\documentclass[12pt,openany]{book}
\usepackage{fancyhdr}
\usepackage{type1cm}         
\usepackage{amssymb}
\usepackage{amsfonts,amsthm}
\usepackage{graphicx}
\usepackage[utf8]{inputenc}
\usepackage{newtxtext}       %
\usepackage[varvw]{newtxmath}   
 \usepackage{subfig}
\newtheorem{definition}{Definition}%
\newtheorem{proposition}{Proposition}
\newtheorem{remark}{Remark}

\newcommand{\half}{\small \mbox{$\frac{1}{2}$}}

\newcommand{\argmin}{\mathop{\rm argmin}}
\newcommand{\argmax}{\mathop{\rm argmax}}

\renewcommand{\chaptermark}[1]{\markboth{}{}}
\renewcommand{\sectionmark}[1]{\markright{}}

\makeindex             

\begin{document}

\title{\vspace{2cm}\bf Minimum Gamma Divergence for\\[10mm] Regression and Classification Problems\vspace{10cm}}

\author{Shinto Eguchi}

\date{ \today}

\maketitle


\chapter*{Preface}
\addcontentsline{toc}{chapter}{Preface}
In an era where data drives decision-making across diverse fields, the need for robust and efficient statistical methods has never been greater. As a researcher deeply involved in the study of divergence measures, I have witnessed firsthand the transformative impact these tools can have on statistical inference and machine learning. This book aims to provide a comprehensive guide to the class of power divergences, with a particular focus on the $\gamma$-divergence, exploring their theoretical underpinnings, practical applications, and potential for enhancing robustness in statistical models and machine learning algorithms.

The inspiration for this book stems from the growing recognition that traditional statistical methods often fall short in the presence of model misspecification, outliers, and noisy data. Divergence measures, such as the $\gamma$-divergence, offer a promising alternative by providing robust estimation techniques that can withstand these challenges. This book seeks to bridge the gap between theoretical development and practical application, offering new insights and methodologies that can be readily applied in various scientific and engineering disciplines.

The book is structured into four main chapters. Chapter 1 introduces the foundational concepts of divergence measures, including the well-known Kullback-Leibler divergence and its limitations. It then presents a detailed exploration of power divergences, such as the $\alpha$, $\beta$, and $\gamma$-divergences, highlighting their unique properties and advantages. Chapter 2 explores minimum divergence methods for regression models, demonstrating how these methods can improve robustness and efficiency in statistical estimation. Chapter 3 extends these methods to Poisson point processes, with a focus on ecological applications, providing a robust framework for modeling species distributions and other spatial phenomena. Finally, Chapter 4 explores the use of divergence measures in machine learning, including applications in Boltzmann machines, AdaBoost, and active learning. The chapter emphasizes the practical benefits of these measures in enhancing model robustness and performance.

By providing a detailed examination of divergence measures, this book aims to offer a valuable resource for statisticians, machine learning practitioners, and researchers. It presents a unified perspective on the use of power divergences in various contexts, offering practical examples and empirical results to illustrate their effectiveness. The methodologies discussed in this book are designed to be both insightful and practical, enabling readers to apply these concepts in their work and research.

This book is the culmination of years of research and collaboration. I am grateful to my colleagues and students whose questions and feedback have shaped the content of this book. Special thanks to Hironori Fujisawa, Masayuki Henmi, Takashi Takenouchi, Osamu Komori, Kenichi Hatashi, Su-Yun Huang, Hung Hung, Shogo Kato, Yusuke Saigusa and Hideitsu Hino for their invaluable support and contributions. 

I invite you to explore the rich landscape of divergence measures presented in this book. Whether you are a researcher, practitioner, or student, I hope you find the concepts and methods discussed here to be both insightful and practical. It is my sincere wish that this book will contribute to the advancement of robust statistical methods and inspire further research and innovation in the field.

\enlargethispage{1\baselineskip}
\vspace{\baselineskip}
\begin{flushright}\noindent
Tokyo, 2024\hfill { Shinto Eguchi}
\\
\end{flushright}

\tableofcontents



\chapter{Power divergence}\label{Power divergence}

\noindent{
We present a mathematical framework for discussing the class of  divergence measures, which are essential tools for quantifying the difference between two probability distributions. 
These measures find applications in various fields such as statistics, machine learning, and data science. 
We begin by discussing the well-known Kullback-Leibler (KL) Divergence, highlighting its advantages and limitations. 
To address the shortcomings of KL-divergence, the paper introduces three alternative types: $\alpha$, $\beta$, and $\gamma$-divergences. 
We emphasize the importance of choosing the right "reference measure," especially for $\beta$ and $\gamma$ divergences, as it significantly impacts the results.

\section{Introduction}

We provide a comprehensive study of divergence measures that are essential tools for quantifying the difference between two probability distributions. 
These measures find applications in various fields such as statistics, machine learning, and data science 
\cite{amari1982differential,burbea1982convexity,eguchi1985differential,rao1987differential,eguchi1992geometry}.  See also \cite{zhang2004divergence,csiszar2004information,liese2006divergences,gneiting2007strictly,nguyen2010estimating,reid2011information,jewson2018principles}.

We present the 
$\alpha$, $\beta$, and $\gamma$ divergence measures, each characterized by distinctive properties and advantages. These measures are particularly well-suited for a variety of applications, offering tailored solutions to specific challenges in statistical inference and machine learning.
We further explore the practical applications of these divergence measures, examining their implementation in statistical models such as generalized linear models and Poisson point processes. Special attention is given to selecting the appropriate 'reference measure,' which is crucial for the accuracy and effectiveness of these methods.
The study concludes by identifying areas for future research, including the further exploration of reference measures. Overall, the paper serves as a valuable resource for understanding the mathematical and practical aspects of divergence measures.

In recent years, a number of studies have been conducted on the robustness of machine learning models using the $\gamma$-divergence, which was proposed in \cite{Fujisawa2008}.
This book highlights that the $\gamma$-divergence can be defined even when the power exponent $\gamma$ is negative, provided certain integrability conditions are met \cite{eguchi2022minimum}. Specifically, one key condition is that the probability distributions are defined on a set of finite discrete values. We demonstrate that the $\gamma$-divergence with 
$\gamma = -1$  is intimately connected to the inequality between the arithmetic mean and the geometric mean of the ratio of two probability mass functions, thus terming it the geometric-mean (GM) divergence. Likewise, we show that the $\gamma$-divergence with 
$\gamma = -2$ can be derived from the inequality between the arithmetic mean and the harmonic mean of the mass functions, leading to its designation as the harmonic-mean (HM) divergence.


\section{Probabilistic framework}

Let  \( Y \)  be a random variable with a set \( \mathcal{Y} \) of possible values in $\mathbb R$. We denote \( \Lambda \) as a \( \sigma \)-finite measure, referred to as the reference measure. The reference measure \( \Lambda \) is typically either the Lebesgue measure for a case where $Y$ is a continuous random variable or the counting measure for a case where $Y$ is a discrete one.
Let us define $\mathcal P$ as the space encompassing all probability measures $P$'s that are absolutely continuous with respect to $\Lambda$ each other. 
The probability \( P(B) \) for an event $B$ can be expressed as:
\begin{align}\nonumber
     P(B)=\int_B p(y) {\rm d}\Lambda(y),
\end{align}
where $p(y)=(\partial P /\partial\Lambda)(y)$ is referred to as the Radon-Nicodym (RN) derivative.
Specifically, $p(y)$ is referred to as the probability density function (pdf) and the probability mass function (pmf) if  $Y$ is a continuous and discrete random variable, respectively.
\begin{definition}\label{def1}
Let $D(P,Q)$ denote a functional defined on ${\mathcal P}\times{\mathcal P}$. 
Then, we call $D(P,Q)$ a divergence measure if $D(P,Q)\geq0$ for all $P$ and $Q$ of $\mathcal P$, and $D(P,Q)=0$ means $P=Q$. 
\end{definition}
Consider two normal distributions \( {\tt Nor}(\mu_1, \sigma_1^2) \) and \( {\tt Nor}(\mu_2, \sigma_2^2) \). If both distributions have the same mean \( \mu \) and variance \( \sigma^2 \), they are identical, and their divergence is zero. However, as the mean and variance of one distribution diverge from those of the other, the divergence measure increases, quantifying how one distribution is different from the other.
Thus, a divergence measure quantifies how one probability distribution diverges from another.
The key properties are  non-negativity,  asymmetry and  zero when identical.
The asymmetry of the divergence measure helps one to discuss model comparisons, variational inferences, generative models, optimal control policies and so on.
Researchers have proposed various divergence measures  in statistics and machine learning to compare two models or to measure the information loss when approximating a distribution.
It is more appropriately termed  `information divergence', although here is simply called `divergence' for simplicity.
As a specific example, the Kullback-Leibler (KL) divergence is given by the following equation:
\begin{equation}\label{KL}
    D_0(P,Q) = \int_{\mathcal{Y}} p(y) \log \frac{p(y)}{q(y)} {\rm d}\Lambda(y),
\end{equation}
where $p(y) = \partial P /\partial\Lambda(y) $ and $q(y) =\partial Q /\partial\Lambda(y) $.
The KL-divergence is essentially independent of
the choice since it is written without $\Lambda$  as
\begin{align}\nonumber
D_0(P,Q)=\int \log \frac{ \partial P}{ \partial Q}\,{\rm d} P.
\end{align}
This implies that any properties for the KL-divergence $D_0(P,Q)$ directly can be thought as  intrinsic properties between probability measures $P$ and $Q$ regardless of the RN-derivatives with respect to the reference measure.   
The definition \eqref{KL} implicitly assumes the integrability. This implies that the integrals of their products with the logarithm of their ratio must be finite.
Such assumptions are almost acceptable in the practical applications in statistics and the machine learning.
However, if we assume a Cauchy  distribution as $P$ and a normal distribution $Q$, then $D_0(P,Q)$ is not a finite value.
Thus, the KL-divergence is associated with instable behaviors, which arises non-robustness for the minimum KL-divergence method or the maximum likelihood. 
This aspect will be discussed in the following chapter.

If we write the cross-entropy as
\begin{align}\label{cross-ent}
H_0(P,Q)=-\int  p \log{q}\, {\rm d}\Lambda,
\end{align}
then the KL-divergence is written by the difference as
\begin{align}\nonumber
D_0(P,Q)=H_0(P,Q )-H_0(P,P ).
\end{align}
The KL-divergence is a divergence measure due to the convexity of the negative logarithmic function.
In foundational statistics, the Neyman-Pearson lemma holds a pivotal role. This lemma posits that the likelihood ratio test (LRT) is the most powerful method for hypothesis testing when comparing a null hypothesis distribution \(P\) against an alternative distribution \(Q\). In this context, the KL-divergence \(D_0(P,Q)\) can be interpreted as the expected value of the LRT under the null hypothesis distribution \(P\). For a more in-depth discussion of the close relationship between \(D_0(P,Q)\) and the Neyman-Pearson lemma, the reader is referred to \cite{eguchi2006interpreting}.

In the context of machine learning, KL-divergence is often used in algorithms like variational autoencoders. Here, KL-divergence helps quantify how closely the learned distribution approximates the real data distribution. Lower KL-divergence values indicate better approximations, thus helping in the model's optimization process.



\section{Power divergence measures}
The KL-divergence is sometimes referred to as log divergence due to its definition involving the logarithmic function.
Alternatively, a specific class of power divergence measures can be derived from power functions characterized by the exponent parameters \( \alpha \), \( \beta \), and \( \gamma \), as detailed below.
Among the numerous ways to quantify the divergence or distance between two probability distributions, `power divergence measure' occupies a unique and significant property. 
Originating from foundational concepts in information theory, these measures have been extended and adapted to address various challenges across statistics and machine learning. As we strive to make better decisions based on data, understanding the nuances between different divergence measures becomes crucial. This section introduces the power divergence measures through three key types: \( \alpha \), \( \beta \), and \( \gamma \) divergences, see \cite{cichocki2010families} for a comprehensive review. 
Each of these offers distinct advantages and limitations, and serves as a building block for diverse applications ranging from robust parameter estimation to model selection and beyond.

\

\noindent
{ (1)  $\alpha$-divergence}:
\begin{align}\nonumber
D_\alpha(P,Q;\Lambda)=\frac{1}{\alpha(\alpha-1)}{\int_{\mathcal Y}\left[
 1- \left(\frac{q(y)}{p(y)}\right)^\alpha\right] p(y) {\rm d}\Lambda(y)}
 ,
\end{align}
where \( \alpha \) belongs to \( \mathbb R \), cf. \cite{chernoff1952measure,amari1982differential} for further details.
Let us introduce 
\begin{align}\nonumber
W_\alpha(R)=\frac{1}{\alpha(\alpha-1)}(1-R^\alpha)-\frac{1}{\alpha-1}(1-R),
\end{align}
as a generator function for $R \geq0$. 
Then the $\alpha$-divergence is written as
\begin{align}\nonumber
D_\alpha(P,Q;\Lambda)=\int_{\mathcal Y}W_\alpha\left(\frac{q(y)}{p(y)}\right)p(y){\rm d}\Lambda(y).
\end{align}
Note that \( W_\alpha(R) \geq 0 \). Equality is achieved if and only if \( R=1 \), indicating that \( W_\alpha(R) \) is a convex function. This implies $D_\alpha(P,Q;\Lambda)\geq0$ with equality if and only if $P=Q$.
This shows that $D_\alpha(P,Q;\Lambda)$ is a divergence measure.
The log expression \cite{chernoff1952measure} is given by
\begin{align}\nonumber
\Delta_\alpha(P,Q;\Lambda)=\frac{1}{\alpha-1}\log {\int_{\mathcal Y} \left(\frac{q(y)}{p(y)}\right)^\alpha p(y) {\rm d}\Lambda(y)}.
\end{align}

The $\alpha$-divergence is associated with the Pythagorean identity in the space $\mathcal P$.
Assume that a triple of $P$, $Q$ and $R$ satisfies
\begin{align}\nonumber
D_\alpha(P,Q;\Lambda)+D_\alpha(Q,R;\Lambda)=D_\alpha(P,R;\Lambda).
\end{align}
This equation reflects a Pythagorean relation, wherein the triple \( P, Q, R \) forms a right triangle if \( D_\alpha(P,Q;\Lambda) \) is considered the squared Euclidean distance between \( P \) and \( Q \).
We define two curves $\{P_t\}_{0\leq t\leq1}$ and $\{R_s\}_{0\leq t\leq1}$ in $\mathcal P$ such that
the RN-derivatives of $P_t$ and $R_s$ is given by $p_t(y)=(1-t)p(y)+tq(y)$ and
\begin{align}\nonumber
r_s(y)=z_s\exp\{(1-s)\log r(y)+s\log q(y)\},
\end{align}
respectively, where $z_s$ is a normalizing constant.
We then observe that the Pythagorean relation remains unchanged for the triple \( P_t, Q, R_s \), as illustrated by the following equation:
\begin{align}\nonumber
D_\alpha(P_t,Q;\Lambda)+D_\alpha(Q,R_s;\Lambda)=D_\alpha(P_t,R_s;\Lambda).
\end{align}
In accordance with this, The $\alpha$-divergence allows for $\mathcal P$ to be as if it were an Euclidean space.  This property plays a central role in the approach of information geometry.
It gives geometric insights for statistics and machine learning \cite{nielsen2021geodesic}.

For example, consider a multinomial distribution {\tt MN}$(\pi,m)$ with a probability mass function (pmf):
\begin{align}\label{Bin}
p(y,\pi,m)=\binom{m}{y_1 \cdots y_k} \pi_1^{y_1}\cdots \pi_k^{y_k} 
\end{align}
for $y=(y_1,...,y_k) \in{\mathcal Y}$ with ${\mathcal Y}=\{\ y\ |\sum_{j=1}^k y_j=m\}$, where $\pi=(\pi_j)_{j=1}^k$.    
The \( \alpha \)-divergence between multinomial distributions \( \text{{\tt MN}}(\pi,m) \) and \( \text{{\tt MN}}(\rho,m) \) can be expressed as follows:
\begin{align}\label{Bin1}
D_\alpha({\tt MN}(\pi,m),{\tt MN}(\rho,m);C)=\frac{1}{\alpha(\alpha-1)}\Big\{1-\sum_{j=1}^k \pi_j^\alpha\rho_j^{1-\alpha}\Big\},
\end{align}
where $C$ is the counting measure.

The $\alpha$-divergence is independent of
the choice of $\Lambda$ since  
\begin{align}\nonumber
D_\alpha(P,Q;\Lambda)=\frac{1}{\alpha(1-\alpha)}{\int \left[
 1- \left(\frac{\partial Q}{\partial P}\right)^\alpha\right]  d P}.
\end{align}
This indicates that \( D_\alpha(P,Q;\Lambda) \) is independent of the choice of \( \Lambda \), as \( D_\alpha(P,Q;\Lambda) = D_\alpha(P,Q;\tilde\Lambda) \) for any \( \tilde\Lambda \). Consequently, equation \eqref{Bin1} is also independent of \( C \).
In general, the Csisar class of divergence is independent of the choice of the reference measure \cite{eguchi2022minimum}.

\

\noindent
{  (2)  $\beta$-divergence}:
\begin{align}\label{Beta}
D_\beta(P,Q;\Lambda)=\frac{1}{\beta(\beta+1)}\int\{  p(y)^{\beta+1}+{\beta}q(y)^{\beta+1}-(\beta+1) p(y)q(y)^{\beta}\}{\rm d}\Lambda(y)
 ,
\end{align}
where \( \beta \) belongs to \( \mathbb R \). For more details, refer to \cite{basu1998robust,Minami2002}.
Let us consider a generator function defined as follows:
\begin{align}\nonumber
  U_\beta(R)=\frac{1}{\beta(\beta+1)}R^{\beta+1}   \   \text{ for } \  R\geq0.
\end{align}
It follows from the convexity of $U_\beta(R)$ that
$
  U_\beta(R_1)-U_\beta(R_0)\geq U_\beta{}^\prime(R_0)(R_1-R_0).
$ 
This concludes  that $D_\beta(P,Q,\Lambda)$ is a divergence measure due to 
\begin{align}\nonumber
D_\beta(P,Q;\Lambda)=\int\{   U_\beta(p)-U_\beta(q)- U_\beta{}^\prime(q)(p-q)
\}{\rm d}\Lambda.
\end{align}
We also observe the property preserving the Pythagorean relation for the $\beta$-divergence.
When $P$, $Q$ and $R$ form a right triangle by the $\beta$-divergence, then the right triangle is preserved for the triple $P_t$, $Q$ anf $R_s$.

It's worth noting that the \( \beta \)-divergence is dependent on the choice of reference measure \( \Lambda \).
For instance, if we choose $\tilde\Lambda$ as the reference measure, then  the \( \beta \)-divergence is given by:
\begin{align}\nonumber
D_\beta(P,Q;\tilde\Lambda)=\frac{1}{\beta(\beta+1)}\int\{ \tilde p ^{\beta+1}+ \beta\tilde  q ^{\beta+1}-(\beta+1) \tilde p  \tilde q ^{\beta}\}{\rm d}\tilde \Lambda .
\end{align}
Here, $\tilde p=\partial P/\partial \tilde \Lambda$ and $\tilde q=\partial Q/\partial \tilde \Lambda$.
This can be rewritten as
\begin{align}\label{homo}
\frac{1}{\beta(\beta+1)}\left\{
\int(\tilde \lambda  p) ^{\beta}dP+ \beta \int (\tilde \lambda q) ^{\beta}dQ
-(\beta+1)\int (\tilde \lambda  q) ^{\beta}dP
\right\},
\end{align}
where $\tilde \lambda ={\partial \tilde\Lambda}/{\partial \Lambda}$.
Hence, the integrands of $D_\beta(P,Q;\tilde\Lambda)$ are given by  the integrands of $D_\beta(P,Q;\Lambda)$ multiplied by $\tilde \lambda^\beta$.
The choice of the reference measure gives a substantial effect for evaluating the $\beta$-divergence.

We again consider a multinomial distribution {\tt MN}$(\pi,m)$ defined in \eqref{Bin}.
Unfortunately, the $\beta$-divergence $D_\beta({\tt MN}(\pi,m),{\tt MN}(\rho,m);C)$ with the counting measure $C$ would have an intractable expression.
Therefore, we select \( \tilde \Lambda \) in a way that the Radon-Nikodym (RN) derivative is defined as
\begin{align}\label{Carrier}
\frac{\partial\tilde\Lambda}{\partial C}(y)=\binom{m}{y_1\cdots y_k}^{-1} 
\end{align}
as a reference measure.
Accordingly,  $\tilde p(y,\pi,n)= \pi_1{}^{y_1}\cdots \pi_k{}^{y_k}$, and hence
\begin{align}\nonumber
\int \tilde p(y,\pi,n)^{\beta} dP(y) 
=\sum_{y\in{\mathcal Y}} \binom{m}{y_1\cdots y_k}\left\{\pi_1{}^{y_1}\cdots \pi_k{}^{y_k}\right\}^{\beta+1},
\end{align}
which is equal to $\big\{\sum_{j=1}^k \pi_j{}^{\beta+1}\big\}^m$.
Using this approach, a closed form expression for \( \beta \)-divergence can be derived:
\begin{align}\nonumber
&D_\beta({\tt MN}(\pi,m),{\tt MN}(\rho,m);\tilde\Lambda) \\[3mm]
&=\frac{\big\{\sum_{j=1}^k \pi_j{}^{\beta+1}\big\}^m}{\beta(\beta+1)}+
\frac{ \big\{\sum_{j=1}^k \rho_j{}^{\beta+1}\big\}^m}{\beta+1} 
-\frac{\big\{\sum_{j=1}^k \pi_j \rho_j{}^{\beta}\big\}^m}{\beta}
\label{Bin2}
\end{align}
due to \eqref{homo}.
In this way, the expression \eqref{Bin2} has a tractable form, in which this is reduced to
the standard one of $\beta$-divergence when $m=1$.
Subsequent discussions will explore the choice of reference measure that provides the most accurate inference within statistical models, such as the generalized linear model and the model of inhomogeneous Poisson point processes.

\ 

\noindent
 { (3) $\gamma$-divergence} \cite{Fujisawa2008}:
\begin{align}\label{Gamma}
D_\gamma(P,Q;\Lambda)=-\frac{1}{\gamma} \frac{\int_{\mathcal Y} p(y)q(y)^\gamma  {\rm d}\Lambda(y)}{\ \ \big(\int_{\mathcal Y}q(y)^{\gamma+1}{\rm d}\Lambda(y)\big)^\frac{\gamma}{\gamma+1}}
+\frac{1}{\gamma}\Big(\int_{\mathcal Y}p(y)^{\gamma+1}{\rm d}\Lambda(y)\Big)^{\frac{1}{\gamma+1}} .
\end{align}
If we define the $\gamma$-cross entropy as:
\begin{align}\nonumber
H_\gamma(P,Q;\Lambda)= -\frac{1}{\gamma}\frac{\int_{\mathcal Y} p(y)q(y)^\gamma  {\rm d}\Lambda(y)}{\ \ \big(\int_{\mathcal Y}q(y)^{\gamma+1}{\rm d}\Lambda(y)\big)^\frac{\gamma}{\gamma+1}}.
\end{align}
then the $\gamma$-divergence is written by the difference: 
\begin{align}\nonumber
D_\gamma(P,Q;\Lambda)=H_\gamma(P,Q;\Lambda)-H_\gamma(P,P;\Lambda).
\end{align}
It is noteworthy that the cross $\gamma$-entropy is a convex-linear functional with respect to the first argument:
\begin{align}\nonumber
\sum_{j=1}^k w_jH_\gamma(P_j,Q;\Lambda)=H_\gamma(\bar P,Q;\Lambda),
\end{align}
where $w_j$'s are positive weights with $\sum_{j=1}^k w_j=1$ and $\bar P=(1/k)\sum_{j=1}^k P_j$.
This property gives explicitly the empirical expression for the $\gamma$-entropy given data set $\{Y_i\}_{1\leq i\leq n}$.
Consider  the empirical distribution $(1/n)\sum_{i=1}^n {1}_{Y_i}$ as $\bar P$,
where $1_{Y_i}$ is the Dirac measure at the atom $Y_i$.
Then 
\begin{align}\nonumber
H_\gamma(\bar P,Q;\Lambda)=-\frac{1}{\gamma} \frac{1}{n}\sum_{i=1}^n \frac{q(Y_i)^\gamma }{\big(\int_{\mathcal Y}q(y)^{\gamma+1}{\rm d}\Lambda(y)\big)^\frac{\gamma}{\gamma+1}}.
\end{align}
If we assume that $\{Y_i\}_{1\leq i\leq n}$ is a sequence of identically and independently distributed with $P$, then 
$$
\mathbb E[H_\gamma(\bar P,Q;\Lambda)]=H_\gamma(P,Q;\Lambda),
$$ 
and hence
$H_\gamma(\bar P,Q;\Lambda)$ almost surely converges to $H_\gamma(P,Q;\Lambda)$ due to the strong law of large numbers.
Subsequently, this will be used to define the empirical loss based on the dataset.
Needless to say, the empirical expression of the cross entropy $H_0(\bar P,Q;\Lambda)$ in \eqref{cross-ent} is the negative log likelihood.
The $\gamma$-diagonal entropy is proportional to the Lebesgue norm with exponent $\lambda=\gamma+1$ as
\begin{align}\nonumber
H_\gamma(P,P;\Lambda)= -\frac{1}{\gamma}
{\Big(\int_{\mathcal Y}p(y)^{\lambda}{\rm d}\Lambda(y)\Big)^\frac{1}{\lambda}}.
\end{align}
Considering the conjugate exponent $\lambda^*=\lambda/(1-\lambda)$, the H\"{o}lder inequality for $p$ and $q^\gamma$ states
\begin{align}\nonumber
{\int_{\mathcal Y} p(y) q(y)^\gamma {\rm d}\Lambda(y)}
\leq{\Big\{\int_{\mathcal Y}p(y)^{\lambda}{\rm d}\Lambda(y)\Big\}^{\frac{1}{\lambda}}\Big\{\int_{\mathcal Y}\{q(y)^{\gamma}\}^{\lambda^*}{\rm d}\Lambda(y)\Big\}^{\frac{1}{\lambda^*}} } .
\end{align}
This holds for any $\lambda>1$ with the equality if and only if $p=q$. 
This implies the $\gamma$-divergence satisfies the definition of  `divergence measure' for any $\gamma>0$. 
It should be noted that the H\"{o}lder inequality is employed for not the pair $p$ and $q$ but that of $p$ and $q^\gamma$, which yields the property
`zero when identical'  as a divergence measure. 
Also, the $\gamma$-divergence approaches the KL-divergence in the limit:
\begin{align}\nonumber
\lim_{\gamma\rightarrow0}D_\gamma(P,Q;\Lambda)=D_0(P,Q;\Lambda).
\end{align}
This is because $\lim_{\gamma\rightarrow0} (s^\gamma-1)/\gamma=\log s$ for all $s>0$
as well as the $\alpha$ and $\beta$ divergence measures.

We observe a close relationship between $\beta$ and $\gamma$ divergence measures.
Consider a maximum problem of the $\beta$-divergence: $\max_{\sigma>0}D_\beta(P,\sigma Q;\Lambda)$. 
By definition, if we write $D_\beta^*(P,Q;\Lambda)=\max_{\sigma>0} D_\beta(P,\sigma Q;\Lambda)$, then $D_\beta^*(P,\sigma Q;\Lambda)=D_\beta^*(P,Q;\Lambda)$ for all $\sigma>0$.
Thus, the maximizer of $\sigma$ is given by
\begin{align}\nonumber
\sigma_{\rm opt}=\frac{\int p(y)q(y)^\beta {\rm d}\Lambda(y)}{\int q(y)^{\beta+1} {\rm d}\Lambda(y)}.
\end{align}
If we confuse $\beta$ with $\gamma$, then the close relationship is found in
\begin{align}\nonumber
D_\beta(P,\sigma_{\rm opt} Q;\Lambda)=
\frac{\{-\gamma H_\gamma(P,Q;\Lambda)\}^{\beta+1}-\{-\gamma H_\gamma(P,P;\Lambda)\}^{\beta+1}}{\beta(\beta+1)}
\end{align}
In accordance with this, the $\gamma$-divergence can be viewed as the $\beta$-divergence interpreted in a projective geometry \cite{eguchi2011projective}.
Similarly, consider a dual problem: $\max_{\sigma>0}D_\beta(\sigma P,Q;\Lambda)$. 
Then, the maximizer of $\sigma$ is given by
\begin{align}\nonumber
\sigma_{\rm opt}^*=\Big\{\frac{\int p(y)q(y)^\beta {\rm d}\Lambda(y)}{\int p(y)^{\beta+1} {\rm d}\Lambda(y)}\Big\}^{\frac{1}{\beta}}.
\end{align}
Hence, the scale adjusted divergence is given by
\begin{align}
    D_{\beta}(\sigma^*_{\rm opt}p, q)
=\frac{-\{-\gamma H_\gamma^*(P,Q;\Lambda)\}^\frac{\beta+1}{\beta}+\beta \{-\gamma H_\gamma^*(P,P;\Lambda)\}^\frac{\beta+1}{\beta}}{\beta(\beta+1)}
\end{align}
Thus, we get   a dualistic version of the $\gamma$-divergence as
\begin{align}\label{gamma-star}
    D^*_{\gamma}(P,Q;\Lambda)=- \frac{1}{\gamma}\frac{\int_{\mathcal Y} p(y) q(y)^\gamma{\rm d}\Lambda(y)}{(\int_{\mathcal Y} p(y)^{\gamma+1}{\rm d}\Lambda(y))^{\frac{1}{\gamma+1}}}+
    \frac{1}{\gamma}\Big(\int_{\mathcal Y} q(y)^{\gamma+1}{\rm d}\Lambda(y)\Big)^{\frac{\gamma}{\gamma+1}}.
\end{align}
We refer $D^*_{\gamma}(p, q)$ to as the dual $\gamma$-divergence.
If we define the dual $\gamma$-entropy as
\begin{align}\nonumber
    H^*_{\gamma}(P,Q;\Lambda)=- \frac{1}{\gamma}\frac{\int_{\mathcal Y} p(y) q(y)^\gamma{\rm d}\Lambda(y)}{(\int_{\mathcal Y} p(y)^{\gamma+1}{\rm d}\Lambda(y))^{\frac{1}{\gamma+1}}},
\end{align}
then $D^*_{\gamma}(P,Q;\Lambda)=H^*_{\gamma}(P,Q;\Lambda)-H^*_{\gamma}(P,P;\Lambda)$.
In effect, the $\gamma$-divergence $D_\gamma$ and its dual $D_\gamma^*$ are connected as follows.
\begin{align}\nonumber
    D^*_{\gamma}(P,Q;\Lambda)=\frac{\big(\int  q^{\gamma+1}{\rm d}\Lambda\big)^{\frac{\gamma}{\gamma+1}}}{\big(\int  p ^{\gamma+1}{\rm d}\Lambda\big )^{\frac{1}{\gamma+1}}}
D_{\gamma}(P,Q;\Lambda).
\end{align}


\section{The $\gamma$-divergence and its dual}
In the evolving landscape of statistical divergence measures, a lesser-explored but highly potent member of the family is the \( \gamma \) divergence. This divergence serves as an interesting alternative to the more commonly used \( \alpha \) and \( \beta \) divergences, with unique properties and advantages that make it particularly suited for certain classes of problems. The dual \( \gamma \) divergence offers further flexibility, allowing for nuanced analysis from different perspectives. The following section is dedicated to a deep-dive into the mathematical formulations and properties of these divergences, shedding light on their invariance characteristics, relationships to other divergences, and potential applications. Notably, we shall establish that \( \gamma \) divergence is well-defined even for negative values of the exponent \( \gamma \), and examine its special cases which connect to the geometric and harmonic mean divergences. This comprehensive treatment aims to illuminate the role that \( \gamma \) and its dual can play in advancing both theoretical and applied aspects of statistical inference and machine learning \cite{eguchi2014duality}.

Let us focus on the $\gamma$-divergence in power divergence measures.
We define a power-transformed function  as follows:
\begin{align}\label{gamma-model}
q^{(\gamma)}(y)=\frac{q(y)^{\gamma+1}}{\int_{\mathcal Y}q(\tilde y)^{\gamma+1}{\rm d}\Lambda(\tilde y)},
\end{align}
which we refer to as the $\gamma$-expression of $q(y)$, where $q\in{\mathcal P}$.
Thus, the measure having the RN-derivative $q^{(\gamma)}$ belongs to $\mathcal P$ since  $\int q^{(\gamma)}(y){\rm d}\Lambda(y)=1$.
We can write  
\begin{align}\nonumber
H_\gamma(P,Q;\Lambda)= -\frac{1}{\gamma}\int_{\mathcal Y}p(y)
\{ {q^{(\gamma)}(y)\}^\frac{\gamma}{\gamma+1}  {\rm d}\Lambda(y)},
\end{align}
and
\begin{align}\nonumber
H_\gamma^*(P,Q;\Lambda)= -\frac{1}{\gamma}\int_{\mathcal Y}
\{ {p^{(\gamma)}(y)\}^\frac{1}{\gamma+1} q(y)^\gamma {\rm d}\Lambda(y)},
\end{align}
These equations directly yield an observation: \( D_\gamma \) and \( D_\gamma^* \) are scale-invariant, but only with respect to one of the two arguments.
\begin{align}
D_\gamma(P,\sigma Q;\Lambda)=D_\gamma(P,Q;\Lambda) 
\text{ and } D_\gamma^*(\sigma P, Q;\Lambda)=D_\gamma(P,Q;\Lambda);
\end{align}
while 
\begin{align}\nonumber
    D_{\gamma}(\sigma P,Q;\Lambda)=\sigma D_{\gamma}(P,Q;\Lambda) \quad \text{ and } \quad D^*_{\gamma}(P,\sigma Q;\Lambda)=\sigma^\gamma D_{\gamma}(P,Q;\Lambda).
\end{align}

The power exponent $\gamma$ is usually assumed to be positive.
However, we extend  $\gamma$ to be a real number in this discussion, see \cite{eguchi2024minimum} for a brief discussion.

\begin{proposition}
\( D_\gamma(P,Q;\Lambda) \) and \( D^*_\gamma(P,Q;\Lambda) \) defined in \eqref{Gamma} and \eqref{gamma-star}, respectively,
are both divergence measures in \eqref{def1} for any real number \( \gamma \).
\end{proposition}
\begin{proof}
We introduce two generator functions defined as:
\begin{align}\label{V}
V_\gamma(R)
=-\frac{1}{\gamma}R^{\frac{\gamma}{1+\gamma}} \quad {\text and} \quad 
V_\gamma^*(R)
=-\frac{1}{\gamma}R^{\frac{1}{1+\gamma}}
\end{align}
for $R>0$.
By definition, the \( \gamma \)  divergence can be expressed as:
\begin{align}\label{V-form1}
D_\gamma(P,Q;\Lambda)=
 \int_{{\mathcal Y}} p(y)\{ V_\gamma(q^{(\gamma)}(y))- V_\gamma(p^{(\gamma)}(y))\}{\rm d}\Lambda(y).
\end{align}
Due to the convexity of $V_\gamma(R)$ in $R$ for any $\gamma\in\mathbb R$, we have
\begin{align}\label{ineq0}
D_\gamma(P,Q;\Lambda)\geq
 \int_{{\mathcal Y}} p(y) V_\gamma^\prime(p^{(\gamma)}(y))\{q^{(\gamma)}(y)-p^{(\gamma)}(y)\}{\rm d}\Lambda(y)
\end{align}
with equality if and only if $P=Q$.
The right-hand-side of \eqref{ineq0}  can be rewritten as:
\begin{align}\nonumber
-\frac{1}{1+\gamma} \Big(\int_{{\mathcal Y}} p^{\gamma+1}{\rm d}\Lambda\Big)^\frac{1}{\gamma+1}
\int_{{\mathcal Y}}(q^{(\gamma)} -p^{(\gamma)}){\rm d}\Lambda.
\end{align}
The second term identically vanishes since $p^{(\gamma)}$ and $q^{(\gamma)}$ have both total mass one. 
Similarly, we observe for any real number $\gamma$ that
\begin{align}\label{V-form}
D_\gamma^*(P,Q;\Lambda)=
 \int_{{\mathcal Y}} \{ V^*_\gamma(p^{(\gamma)}(y))- V^*_\gamma(q^{(\gamma)}(y))\}q(y)^\gamma {\rm d}\Lambda(y)
\end{align}
which is equal or greater than $0$ and the equality holds if and only if $P=Q$ due the convexity of $V^*(R)$.
Therefore, \( D_\gamma(P,Q;\Lambda) \)  and \( D^*_\gamma(P,Q;\Lambda) \) are both divergence measures for any real number \( \gamma \).
\end{proof}

We will discuss \( D_\gamma(P,Q;\Lambda) \) with a negative power exponent $\gamma$
in a context of statistical inferences. 
The $\gamma$-divergence \eqref{Gamma} is implicitly assumed to be integrable as well as the KL-divergence, in which 
the integrability condition for the $\gamma$-divergence with $\gamma<0$ 
is presented.
Let us look into the case of a multinomial distribution {\tt Bin}$(\pi,m)$ defined in \eqref{Bin} with the reference measure $\tilde\Lambda$ given by \eqref{Carrier}.
An argument similar to that on the $\beta$-divergence yields 
\begin{align}\nonumber
D_\gamma({\tt MN}(\pi,m),{\tt MN}(\rho,m);\tilde\Lambda)=-\frac{
m}{\gamma}
\frac{\sum_{j=1}^k\pi_j \rho_j{}^{\gamma}}{
\big\{ \sum_{j=1}^k \rho_j{}^{\gamma+1}\big\}^{\frac{\gamma }{\gamma+1}}}
+\frac{m}{\gamma}
\Big\{ \sum_{j=1}^k \pi_j{}^{\gamma+1}\Big\}^{\frac{ 1}{\gamma+1}}.\label{Bin3}
\end{align}
as the $\gamma$-divergence of the log expression in \eqref{Gamma-log}, where $\pi$ and $\rho$ are cell probability vectors of $m$ dimension.
The $\gamma$-divergence with the counting measure would also have no closed expression.
Therefore, careful consideration is needed when choosing the reference measure $\Lambda$
for $\gamma$ divergence.
Let $\pi(y)$ be the RN-derivative of $\Lambda$ with respect to the Lebesgue measure $L$.  Then, the $\gamma$-divergence  \eqref{Gamma} with
\begin{align}\nonumber
H_\gamma(P,Q;\Lambda)= -\frac{1}{\gamma}\frac{\int_{\mathcal Y} p(y)q(y)^\gamma \pi(y)  d L(y)}{\ \ \big(\int_{\mathcal Y}q(y)^{\gamma+1}\pi(y)  d L(y)\big)^\frac{\gamma}{\gamma+1}},
\end{align}
where $\pi=\partial \Lambda/\partial L$.
Our key objective is to identify a \( \pi \) that ensures stable and smooth behavior under a given statistical model and dataset.

We discuss a generalization of the $\gamma$-divergence.
Let $V$ be a convex function.  Then, $V$-divergence is defined by
\begin{equation}
D_V(P,Q;\Lambda)= \int_{\mathcal Y} p(y)\{V(v^*(z_q q(y)))-V(v^*(z_p p(y)))\}{\rm d}\Lambda(y),
\label{general}
\end{equation}
where $z_q$ is a normalizing constant satisfying $\int  v^*(z_q q(y)){\rm d}\Lambda(y)=1$ and the function $v^*(t)$ satisfies 
\begin{equation}  \label{assump1}
V^\prime(v^*(t))=\frac{1}{t} .
\end{equation}
It is derived from the assumption of the convexity of $V$ that
\begin{align}\nonumber
D_V(P,Q;\Lambda)\geq  \int_{\mathcal Y} p(y)V^\prime(v^*(z_p p(y)))
\{v^*(z_qq(y))-v^*(z_p p(y))\}{\rm d}\Lambda(y)
\end{align}
which is equal to
\begin{align}\label{int}
\int_{\mathcal Y} \{ v^*(z_q q(y))- v^*(z_p p(y))\}{\rm d}\Lambda(y)
\end{align}
up to the proportional factor due to \eqref{assump1}.
By the definition of  normalizing constants, the integral in \eqref{int} vanishes.
Hence, $D_V(P,Q;\Lambda)$ becomes a divergence measure.
Specifically, $D_V(P,\sigma Q;\Lambda)=D_V(P,Q;\Lambda)$ for $\sigma>0$ due to the normalizing constant $z_q$.
For example, if $V(R)=-(1/\gamma)R^\frac{\gamma}{\gamma+1}$ as in \eqref{V},
$D_V(P,Q;\Lambda)$ reduces to the $\gamma$-divergence.
There are various examples of $V$ other than \eqref{V}, for example, 
\begin{align}\nonumber
V(R)=2\log(\sqrt{R+1}-1)
+\sqrt{R+1}, 
\end{align}
which is related to the $\kappa$-entropy discussed in a physical context \cite{kaniadakis2001non,pistone2009kappa}.  
We do not go further into this topic as it is beyond the scope of this paper.

We investigate a notable property of the dual \( \gamma \) divergence.
There exists a strong relationship between the generalized mean of probability measures and the minimization of the average dual \( \gamma \) divergence.
Subsequently, we will explore its applications in active learning.

\begin{proposition}\label{A-func}
Consider an average of $k$ dual $\gamma$-divergence measures as
\begin{align}\nonumber
A (P)=\sum_{j=1}^k w_j D_\gamma^*(P,Q_j;\Lambda).
\end{align}
Let \( P^{\rm opt} = \argmin_{P\in{\mathcal P}} A (P) \).
Then, the Radon-Nikodym (RN) derivative of \( P^{\rm opt} \) is uniquely determined as follows:
\begin{align}\nonumber
  \frac{\partial P ^{\rm opt}}{\partial \Lambda}(y)=z_w \Big\{\sum_{j=1}^k w_j q_j(y)^\gamma\Big\}^{\frac{1}{\gamma}}
\end{align}
where $q_j={\partial Q_j}/{\partial \Lambda}$ and $z_w$ is the normalizing constant.
\end{proposition}
\begin{proof}
If we write $\partial P ^{\rm opt}/{\partial \Lambda}$ by $p ^{\rm opt}$, then
\begin{align}\nonumber
A (P)-A (P ^{\rm opt})=- \frac{1}{\gamma}\sum_{j=1}^k w_j\Big\{
\frac{\int_{\mathcal Y} p(y) q_j(y)^\gamma{\rm d}\Lambda(y)}{(\int_{\mathcal Y} p(y)^{\gamma+1}{\rm d}\Lambda(y))^{\frac{1}{\gamma+1}}}
-\frac{\int_{\mathcal Y} p_w^{\rm opt}(y) q_j(y)^\gamma{\rm d}\Lambda(y)}{(\int_{\mathcal Y} \{p_w^{\rm opt}(y)\}^{\gamma+1}{\rm d}\Lambda(y))^{\frac{1}{\gamma+1}}}
\Big\}
\end{align}
which is equal to
\begin{align}\nonumber
- \frac{1}{z_w}\frac{1}{\gamma}\Big\{
\frac{\int_{\mathcal Y} p(y) \{p ^{\rm opt}(y)\}^\gamma{\rm d}\Lambda(y)}{(\int_{\mathcal Y} p(y)^{\gamma+1}{\rm d}\Lambda(y))^{\frac{1}{\gamma+1}}}
-{\Big(\int_{\mathcal Y} \{p ^{\rm opt}(y)\}^{\gamma+1}{\rm d}\Lambda(y)\Big)^{\frac{\gamma}{\gamma+1}}}
\Big\}.
\end{align}
This expression simplifies to \( (1/z) D_\gamma^*(P,P^{\rm opt};\Lambda) \).
Therefore, $A (P)\geq A (P^{\rm opt})$ and the equality holds if and only if $P=P^{\rm opt}$ 
This is due to the property of \( D_\gamma^*(P,P^{\rm opt};\Lambda) \) as
a divergence measure.
\end{proof}
The optimal distribution $P ^{\rm opt}$ can be viewed as the consensus one integrating $k$ committee members' distributions $Q_j$'s into the average of divergence measures with importance weights $w_j$.
We adopt a "query by committee" approach and examine the robustness against variations in committee distributions.
Proposition \ref{A-func} leads to an average version of the Pythagorean relations:
\begin{align}\nonumber
\sum_{j=1}^k w_j D_\gamma^*(P,Q_j)=D_\gamma^*(P,P ^{\rm opt})+\sum_{j=1}^k w_j D_\gamma^*(P ^{\rm opt},Q_j).
\end{align}
We refer to \( P^{\rm opt} \) as the power mean of the set \( \{Q_j\}_{1\leq j\leq k} \).
In general, a generalized mean is defined as
\begin{align}\nonumber
  GM_\phi =\phi^{-1}\Big(\sum_{j=1}^k w_j \phi(q_j)\Big),
\end{align} 
where $\phi$ is a one-to-one function on $(0,\infty)$.
We confirm that, if $\gamma=1$, then $P ^{\rm opt}=\sum_{j=1}^k w_j Q_j$ that  is the arithmetic mean, or the mixture distribution of $Q_j$'s with mixture proportions $w_j$'s.
If $\gamma=-1$, then 
\begin{align}\nonumber
  \frac{\partial P ^{\rm opt}}{\partial \Lambda}(y)=z_w \Big\{\sum_{j=1}^k w_j q_j(y)^{-1}\Big\}^{-1},
\end{align}
which is the harmonic mean of $q_j$'s with weights $w_j$'s.
As $\gamma$ goes to $0$, the dual $\gamma$-divergence $D_\gamma^*(P,Q;\Lambda)$ converges to the dual KL-divergence $D_0^*(P,Q)$ defined by $D_0(Q,P)$.  
The minimizer is given by 
\begin{align}\nonumber
  \frac{\partial P ^{\rm opt}}{\partial \Lambda}(y)=z  \prod_{j=1}^k \{q_j(y)\}^{w_j},
\end{align}
which is the harmonic mean of $q_j$'s with weights $w_j$'s.
We will discuss divergence measures using the harmonic and geometric means of ratios for RN-derivatives in a later section.

We often utilize the logarithmic expression for the \( \gamma \) divergence, given by
\begin{align}\label{Gamma-log}
\Delta_\gamma(P,Q;\Lambda)=-\frac{1}{\gamma}
\log \frac{\int_{\mathcal Y} p(y) q(y)^\gamma {\rm d}\Lambda(y)}{\big(\int_{\mathcal Y}p(y)^{\gamma+1}{\rm d}\Lambda(y)\big)^{\frac{1}{\gamma+1}}\big(\int_{\mathcal Y}q(y)^{\gamma+1}{\rm d}\Lambda(y)\big)^{\frac{\gamma}{\gamma+1}} } .
\end{align}
We find a remarkable property such that 
\begin{align}\nonumber
\Delta_\gamma(\tau P,\sigma Q;\Lambda)=\Delta_\gamma(P,Q;\Lambda) 
\end{align}
for all $\tau>0$ and $\sigma>0$, noting the log expression is written by
\begin{align}\label{Gamma-log1}
\Delta_\gamma(P,Q;\Lambda)=-\frac{1}{\gamma}
\log {\int_{\mathcal Y} \{p^{(\gamma)}(y)\}^\frac{1}{\gamma+1}  \{q^{(\gamma)}(y)\}^\frac{\gamma}{\gamma+1}  {\rm d}\Lambda(y)} 
\end{align}
by the use of $\gamma$-expression defined in \eqref{gamma-model}.
This implies that $\Delta_\gamma(P,Q;\Lambda)$ measures not a departure between $P$ ad $Q$ but an angle between them.
When $\gamma=1$, this is the negative log cosine similarity for $p$ and $q$.
In effect, the cosine for $a$ and $b$ are in $\mathbb R^d$ is defined by
\begin{align}\nonumber
\cos(a,b) =\frac{\sum_{j=1}^d a_j b_j}{\{\sum_{j=1}^d |a_j|^2\}^\frac{1}{2}\{\sum_{j=1}^|b_j|^2\}^\frac{1}{2} }.
\end{align}
This is closely related to \( \Delta_\gamma(P,Q;\Lambda) \) in a discrete space \( \mathcal Y \) when \( \gamma=1 \).
which is closely related to $\Delta_\gamma(P,Q;\Lambda)$ on $d$ discrete space $\mathcal Y$ when $\gamma=1$.
We will discuss an extension of $\Delta_\gamma$ to be defined on a space of all signed measures that comprehensively gives an asymmetry in measuring the angle.

In summary, the exploration of power divergence measures in this section has illuminated their potential as versatile tools for quantifying the divergence between probability distributions. From the foundational Kullback-Leibler divergence to the more specialized \( \alpha \), \( \beta \), and \( \gamma \) divergence measures, we have seen that each type has its own strengths and limitations, making them suited for particular classes of problems. We have also underscored the mathematical properties that make these divergences unique, such as invariance under different conditions and applicability in empirical settings. As the field of statistics and machine learning continues to evolve, it's evident that these power divergence measures will find even broader applications, providing rigorous ways to compare models, make predictions, and draw inferences from increasingly complex data.


\section{GM  and HM  divergence measures}

We discuss a situation where a random variable $Y$ is discrete, taking values in a finite set of non-negative integers denoted by ${\mathcal Y}=\{1,...,k\}$.  
Let $\mathcal P$ be  the space  of all probability measures on $\mathcal Y$. 
In the realm of statistical divergence measures,  arithmetic, geometric, and harmonic
means for a probability measure $P$ of $\mathcal P$  receives less attention despite their mathematical elegance and potential applications. 
For this, consider 
the RN-derivative of a probability measure relative to $Q$ that equals a ratio of probability mass functions (pmfs) $p$ and $q$ induced by  $P$ and $Q$ in ${\mathcal P}$.
Then,  there is well-known inequality between the arithmetic and geometric means:
\begin{align}\label{GM 1}
\sum_{y\in{\mathcal Y}} \frac{p(y )}{q(y )}r(y) \geq\prod_{y\in{\mathcal Y}} \Big\{\frac{p(y )}{q(y )}\Big\}^{r(y)}
\end{align}
and that between the arithmetic and harmonic means:
\begin{align}\label{HM 0}
\sum_{y\in{\mathcal Y}} \frac{p(y )}{q(y )}r(y) \geq
\Big[\sum_{y\in{\mathcal Y}} \Big\{\frac{p(y )}{q(y )}\Big\}^{-1}r(y)\Big]^{-1},
\end{align}
where $r(y)$ is a weight function that is arbitrarily a fixed pmf on $\mathcal Y$.
Equality in \eqref{GM 1} or \eqref{HM 0} holds if and only if $p=q$.
This well-known inequality relationships among these means serve as the mathematical bedrock for defining new divergence measures. 
Specifically, the Geometric Mean (GM) and Harmonic Mean (HM) divergences are inspired by inequalities involving these means and ratios of probabilities as follows. 

First, we define the GM-divergence as
 \begin{align}\nonumber
D_{\rm GM }(P ,Q ;R)=\sum_{y\in{\mathcal Y}}\frac{p(y )}{q(y )}r(y)\prod_{y\in{\mathcal Y}} q(y)^{r(y)} -\prod_{y=1}^k p(y)^{r(y)}
\end{align}
transforming the expression \eqref{GM 1}, where $p(y)$, $q(y)$ and $r(y)$ are the pmfs with respect to $P$, $Q$  and $R$, respectively. 
Note that $ D_{\rm GM }(P ,Q)$ is a divergence measure on $\mathcal P$ as defined in Definition \ref{def1}.
We restrict $\mathcal Y$ to be a finite discrete set for this discussion; however, our results can be generalized.
In effect, the GM-divergence has a general form: 
 \begin{align}\label{pr1}
D_{\rm GM }(P ,Q;R )=\int_{\mathcal Y}\frac{p(y )}{q(y )}dR(y)
\exp\Big\{\int_{\mathcal Y}\log q(y) dR(y)\Big\}
 -\exp\Big\{\int_{\mathcal Y}\log p(y) dR(y)\Big\}
\end{align}
For comparison, we have a look at the $\gamma$-divergence
 \begin{align}\label{pr2}
D_{\gamma}(P ,Q;R )=-\frac{1}{\gamma} 
\frac{\int_{{\mathcal Y}}p(y ){q(y )}^\gamma dR(y)}
{\big(\int_{{\mathcal Y}}{q(y )}^{\gamma+1} dR(y)\big)^{\frac{\gamma}{\gamma+1}}}
+\frac{1}{\gamma} 
{\big(\int_{{\mathcal Y}}{p(y )}^{\gamma+1} dR(y)\big)^{\frac{1}{\gamma+1}}},
\end{align}
by selecting a probability measure $R$ as a reference measure. 

\begin{proposition}
Let $D_{\rm GM }(P ,Q;R )$ and $D_{\gamma}(P ,Q;R )$ be defined in \eqref{pr1} and \eqref{pr2}, respectively.
Then,
 \begin{align}\label{pr3}
\lim_{\gamma\rightarrow-1} D_{\gamma}(P ,Q;R )=D_{\rm GM }(P ,Q;R )
\end{align}
\end{proposition}
\begin{proof}
It follows from the L'Hôpital's rule that
 \begin{align}\nonumber
\lim_{\gamma\rightarrow-1}
{\frac{1}{\gamma+1}}\log {\int_{{\mathcal Y}}{p(y )}^{\gamma+1} dR(y)}
= \int_{ {\mathcal Y}} \log p(y) dR(y).
\end{align}
This implies
 \begin{align}\nonumber
\lim_{\gamma\rightarrow-1}
\Big\{{\int_{{\mathcal Y}}{p(y )}^{\gamma+1} dR(y)}\Big\}^{\frac{1}{\gamma+1}}
= \exp\Big\{\  \int_{ {\mathcal Y}} \log p(y) dR(y) \Big\}.
\end{align}
In accordance, we conclude \eqref{pr3}.
\end{proof}
We write the GM-divergence by the difference of the cross and diagonal entropy measures: $D_{\rm GM }(P ,Q;R )=H_{\rm GM }(P ,Q;R )-H_{\rm GM }(P ,P;R )$, where
\begin{align}\nonumber
H_{\rm GM }(P ,Q;R )=\int_{ {\mathcal Y}}\frac{p(y )}{q(y )}dR(y)\exp\Big\{ \int_{{\mathcal Y}} \log q(y) dR(y)\Big\}.
\end{align}
The GM-divergence has a log expression:
 \begin{align}\nonumber
\Delta_{\rm GM }(P ,Q;R )=\log \frac{H_{\rm GM }(P ,Q;R )}{H_{\rm GM }(P ,P;R )}.
\end{align}
We note that $\Delta_{\rm GM }(P ,Q;R )$ is equal to $\Delta_\gamma(P,Q;R)$ taking the limit of $\gamma$ to $-1$.
Here we discuss the case of the Poisson distribution family. We choose a Poisson distribution {\tt Po}$(\tau)$ as the reference measure.  
Thus, the GM-divergence of the log-form is given by
 \begin{align}\nonumber
\Delta_{\rm GM }({\tt Po}(\lambda) ,{\tt Po}(\mu);{\tt Po}(\tau) )
=\tau\Big(\log\frac{\lambda}{\mu}-\frac{\lambda}{\mu}+1\Big).
\end{align}


Second, we introduce  the HM-divergence.
Suppose  $r(y)=q(y)^{-1}/\sum_{z\in{\mathcal Y}}q(z)^{-1}$ in the inequality \eqref{HM 0}.
Then,   \eqref{HM 0} is written as
\begin{align}\label{HM 1}
\sum_{y\in{\mathcal Y}} \frac{p(y )}{q(y )^2}\Big\{\sum_{y\in{\mathcal Y}}q(y)^{-1}\Big\}^{-1}
 \geq
\Big\{\sum_{y\in{\mathcal Y}}{p(y )}^{-1}\Big\}^{-1}\Big\{\sum_{y\in{\mathcal Y}}q(y)^{-1}\Big\}.
\end{align}
We define  the harmonic-mean (HM ) divergence arranging the inequality \eqref{HM 1} as 
\begin{align}\nonumber
D_{\rm HM }(P ,Q)=H_{\rm HM }(P ,Q)-H_{\rm HM }(P ,P).
\end{align}
Here  
 \begin{align}\nonumber
H_{\rm HM }(P ,Q)=\sum_{y\in{\mathcal Y}} \frac{p(y )}{q(y )^2}\Big\{\sum_{y\in{\mathcal Y}}q(y)^{-1}\Big\}^{-2}
\end{align}
is the cross entropy, where $p(y)$ and $q(y)$ is the pmfs of $P$ and $Q$, respectively.
The (diagonal) entropy is given by the harmonic mean of $p(y)$'s:
 \begin{align}\nonumber
H_{\rm HM }(P ,P)=\Big\{\sum_{y\in{\mathcal Y}}p(y)^{-1}\Big\}^{-1}.
\end{align}
Note that $D_{\rm HM }(P ,Q)$ qualifies as a divergence measure on $\mathcal P$, as defined in Definition \ref{def1}, due to the inequality \eqref{HM 1}.
When $\gamma=-2$, $D_{\rm HM }(P ,Q)$ is equal to $D_{\gamma}(P ,Q ;C)$ in \eqref{Gamma} with the counting measure $C$.
The log form is given by
 \begin{align}\nonumber
\Delta_{\rm HM }(P ,Q)=\log H_{\rm HM }(P ,Q)-\log H_{\rm HM }(P ,P).
\end{align}

The GM-divergence provides an insightful lens through which we can examine statistical similarity or dissimilarity by leveraging the multiplicative nature of probabilities. The 
HM-divergence, on the other hand, focuses on rates and ratios, thus providing a complementary perspective to the GM-divergence, particularly useful in scenarios where rate-based analysis is pivotal.
By extending the divergence measures to include GM and HM divergence, we gain a nuanced toolkit for quantifying divergence, each with unique advantages and applications. For instance, the GM-divergence could be particularly useful in applications where multiplicative effects are prominent, such as in network science or econometrics. Similarly, the HM-divergence might be beneficial in settings like biostatistics or communications, where rate and proportion are of prime importance.
This framework, rooted in the relationships among arithmetic, geometric, and harmonic means, not only expands the class of divergence measures but also elevates our understanding of how different mathematical properties can be tailored to suit the needs of diverse statistical challenges."


\section{Concluding remarks}

In summary, this chapter has laid the groundwork for understanding the class of power divergence measures in a probabilistic framework. 
We study that divergence measures quantify the difference between two probability distributions and have applications in statistics and machine learning.
It begins with the well-known Kullback-Leibler (KL) divergence, highlighting its advantages and limitations. 
To address limitations of KL-divergence, three types of power divergence measures are introduced.

Let us look at the $\alpha$, $\beta$, and $\gamma$-divergence measures for a Poisson distribution model.
Let ${\tt Po}(\lambda)$ denote a Poisson distribution with the RN-derivative
 \begin{align}\label{Poisson-RN}
p(y,\lambda)=\lambda^ye^{-\lambda}
\end{align}
with respect to the reference measure $R$ having $(\partial R)/(\partial C)(y)=1/y!$.
Seven examples of power divergence between Poisson distributions ${\tt Po}(\lambda_0)$ and ${\tt Po}(\lambda_1)$ are listed in Table \ref{vari}.
Note that this choice of the reference measure enable us to having such an tractable form of 
the $\beta$ and $\gamma$ divergences as well as its variants. 
Here we use a basic formula to obtain these divergence measures:
 \begin{align}\nonumber
\sum_{y=0}^\infty p(y,\lambda)^a {\rm d}R (y)=\exp(\lambda^a - a\lambda)
\end{align}
for an exponent $a$ of $\mathbb R$.
The contour sets of seven divergences between Poisson distributions are plotted in Figure \ref{Div-cont}.  
All the divergences attain the unique minimum $0$ in the diagonal line $\{(\lambda_0,\lambda_1):\lambda_0=\lambda_1\}$.   The contour set of GM and HM divergences are flat compared to those of other divergences.

\begin{table}[hbtp]
\caption{The variants of power divergence}
  \label{vari}
  \centering
\vspace{2mm}
  \begin{tabular}{cc}
 
 \hline \\[-1mm]

$\alpha$-divergence   & $ \frac{1}{\alpha(1-\alpha)}(1-\exp\{\lambda_0(\frac{\lambda_1}{\lambda_0})^\alpha
-(1-\alpha)\lambda_0-\alpha\lambda_1\} $ 
\\[5mm]
$\beta$-divergence & $\displaystyle\frac{e^{\lambda_0^{\beta+1}-(\beta+1)\lambda_0}+\beta e^{\lambda_1^{\beta+1}-(\beta+1)\lambda_1}-
(\beta+1) e^{\lambda_0\lambda_1^{\beta}-\lambda_0-\beta\lambda_1}}{\beta(\beta+1)}$
 \\[5mm]
    $\gamma$-divergence  
 & $-\frac{1}{\gamma}\exp({-\lambda_0})\{\exp(\lambda_0\lambda_1^\gamma-\frac{\gamma}{\gamma+1}\lambda_1^{\gamma+1})-\exp(\frac{1}{\gamma+1}\lambda_0^{\gamma+1})\}$
\\[5mm]
dual $\gamma$-divergence & $-\frac{1}{\gamma}\exp({-\gamma \lambda_1})\{\exp(\lambda_0\lambda_1^\gamma-\frac{1}{\gamma+1}\lambda_0^{\gamma+1})-\exp(\frac{\gamma}{\gamma+1}\lambda_1^{\gamma+1})\}$
 \\[5mm]
log $\gamma$-divergence & $-\frac{1}{\gamma}(\lambda_0\lambda_1^\gamma-\frac{1}{\gamma+1}\lambda_0^{\gamma+1}-\frac{\gamma}{\gamma+1}\lambda_1^{\gamma+1})$
 \\[5mm]
GM-divergence  & $\lambda_1\exp({-\lambda_0})\{\exp(\frac{\lambda_0}{\lambda_1})
-\exp(1)\}$
 \\[5mm]
HM-divergence  & $\half \exp({-\lambda_0})\{\exp(\frac{\lambda_0}{\lambda_1^2}-2\frac{1}{\lambda_1})-\exp(-\frac{1}{\lambda_0})\}$
 \\[5mm]
   \hline
  \end{tabular}

\end{table}

\begin{figure}[htbp]
\begin{center}
\includegraphics[width=140mm]{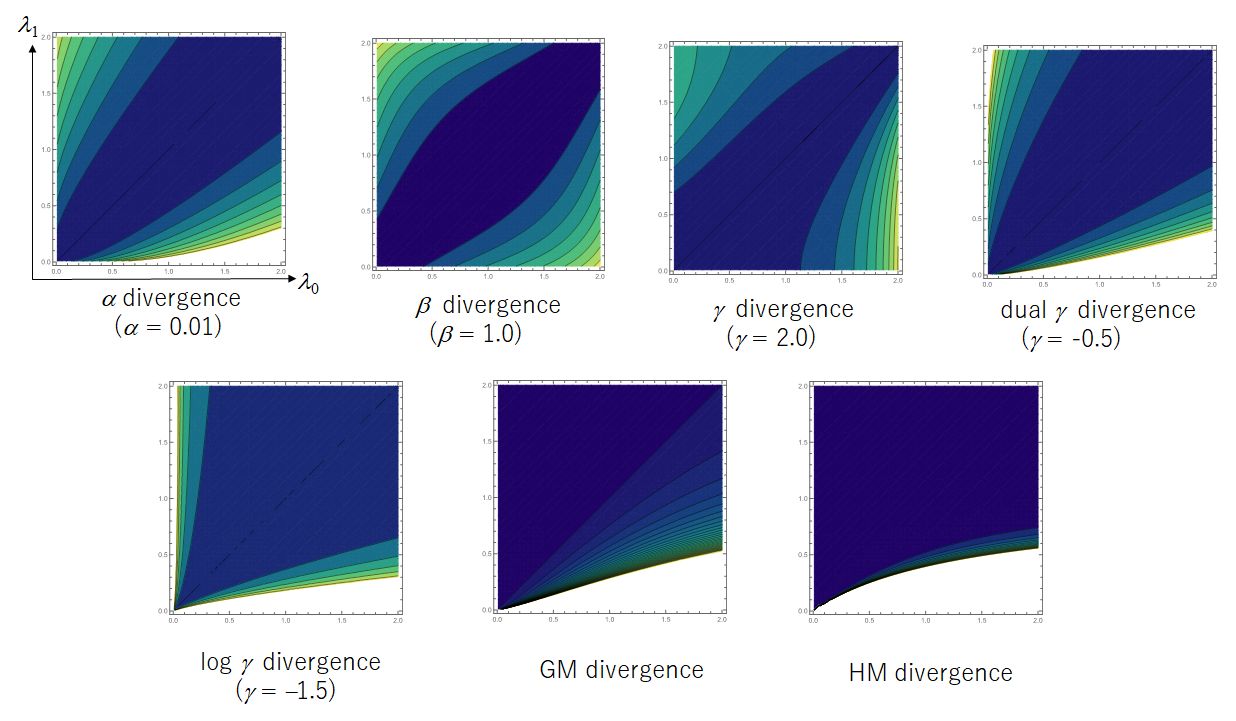}
\end{center}
 \vspace{-5mm}\caption{Contour plots of the power divergences.}
\label{Div-cont}
\end{figure}

The $\alpha$-divergence is intrinsic to assess the divergence between two probability measures.
One of the most important properties  is invariance with the choice of the reference measure that expresses  the Radon-Nicodym derivatives for the two probability measure .
The invariance provides direct understandings for the intrinsic properties beyond  properties for the probability density or mass functions.
A serious drawback is pointed out that an empirical counterpart is not available for a given data set in most practical situations.
This makes difficult for applications for statistical inference for estimation and prediction.
In effect, the statistical applications are limited to the curved exponential family that is modeled in an exponential family.  See \cite{efron1975defining} for the statistical curvature characterizing the second order efficiency.

The $\beta$-divergence and the $\gamma$-divergence are not invariant with the choice of the reference measure.
We have to determine the reference measure from the point of the application to statistics and machine learning.
Subsequently, we discuss the appropriate selection of the reference measure in both cases of $\beta$ and $\gamma$ divergences.  
Both divergence measures are efficient for applications in areas of statistics and machine learning since the empirical loss function for a parametric model distribution is applicable for any dataset. 
For example, the $\beta$-divergence is utilized as a cost function to measure the difference between a data matrix and the factorized matrix in the nonnegative matrix factorization. 
Such applications the minimum $\beta$-divergence method is more robust than the maximum likelihood method which can be viewed as the minimum KL-divergence method.
In practice, the $\beta$ is not scale invariant in the space of all finite measures that involves that of all the probability measures.
We will study that the lack of the scale invariance does not allow simple estimating functions even under a normal distribution model.

Alternatively, the $\gamma$-divergence is scale-invariant with respect to the second argument. 
The $\gamma$-divergence provides a simple estimating function for the minimum $\gamma$-estimator.
This property enables to proposing an efficient learning algorithm for solving the estimating equation.
For example, the $\gamma$-divergence is used for the clustering analysis.
The cluster centers are determined by local minimums of the empirical loss function defined by the $\gamma$-divergence, see 
\cite{notsu2014spontaneous,Notsu2016}  for the learning architecture. 
  A fixed-point type of algorithm is proposed to conduct a fast detection for the local minimums.
Such practical properties in applications will be explored in the following section.
We consider the dual  $\gamma$-divergence that is invariant for the first argument.
We will explore the applicability for defining the consensus distribution in a context of active learning.
It is confirmed that the $\gamma$-divergence is well defined even for negative value of the exponent $\gamma$.
The $\gamma$-divergences with $\gamma=-1$ and $\gamma=-2$  
are reduced to the GM and HM divergences, respectively. 
In a subsequent discussion,  special attentions to the GM and HM divergences  are explored for various objectives in applications, see \cite{eguchi2022minimum} for the application to dynamic treatment regimes in the medical science.


\chapter{Minimum divergence for regression model}\label{Minimum divergence for regression model}

\noindent{
This chapter explores statistical estimation within regression models. We introduce a comprehensive class of estimators known as Minimum Divergence Estimators (MDEs), along with their empirical loss functions under a parametric framework. Standard properties such as unbiasedness, consistency, and asymptotic normality of these estimators are thoroughly examined. Additionally, the chapter addresses the issue of model misspecification, which can result in biases, inaccurate inferences, and diminished statistical power, and highlights the vulnerability of conventional methods to such misspecifications. Our primary goal is to identify estimators that remain robust against potential biases arising from model misspecification. We place particular emphasis on the 
 \(\gamma\)-divergence, which underpins the \(\gamma\)-estimator known for its efficiency and robustness.}

\section{Introduction}
We study statistical estimation in a regression model including a generalized linear model.
The maximum likelihood (ML) method is widely employed and developed for the estimation problem.
This estimation method has been standard on the basis of the solid evidence in which the ML-estimator 
is asymptotically consistent and efficient when the underlying distribution is correctly specified a regression model. See 
\cite{cox1979theoretical,McCullagh1989,casella2024statistical,hastie2009elements}. 
The power of parametric inference in regression models is substantial, offering several advantages and capabilities that are essential for effective statistical analysis and decision-making. 
The ML method has been a cornerstone of parametric inference. 
This principle yields estimators that are asymptotically unbiased, consistent, and efficient,  given that the model is correctly specified. 
Specifically, generalized linear models (GLMs) extend linear models to accommodate response variables that have error distribution models other than a normal distribution, enabling the use of ML estimation across binomial, Poisson, and other exponential family distributions.

However, we are frequently concerned with model misspecification, which occurs when the statistical model does not accurately represent the data-generating process.
This could be due to the omission of important variables, the inclusion of irrelevant variables, incorrect functional forms, and wrong error structures.
Such misspecification can lead to biases, inaccurate inferences, and reduced statistical power because of misspecification for the parametric model.
See \cite{white1982maximum} for the critical issue of model misspecification.  
A misspecified model is more sensitive to outliers, often resulting in more biased estimates. Outliers can also obscure true relationship between variables, making it difficult to detect model misspecification by overshadowing the true relationships between variables. 
Unfortunately, the performance of the ML-estimator is easily degraded in such difficult situations because of the excessive sensitivities to model misspecification.
Such limitations in the face of model misspecification and complex data structures have prompted the development of a broad spectrum of alternative methodologies. 
In this way, we take the MDE approach other than the maximum likelihood.


We discuss a class of estimating methods through minimization of divergence measure
\cite{basu1998robust,Minami2002,Murata2004,Mollah2006,Eguchi2006,Mollah2007,Eguchi2009b,komori2016asymmetric,Komori2019,Komori2021,eguchi2022minimum}. 
These are known as minimum divergence estimators (MDEs).
The empirical loss functions for a given dataset are discussed in  a unified perspective under a parametric model.
Thus, we derive a broad class of estimation methods via MDEs. 
Our primary objective is to find estimators that are robust against potential biases in the presence of model misspecification.
MDEs can be applied to a case where the outcome is a continuous variable in a straightforward manner, in which the reference measure to define the divergence is fixed by the Lebesgue measure.
Alternatively,  more consideration is needed regarding the choice of a reference measure in a discrete variable case for the outcome.  
In particular, $\beta$ and $\gamma$ divergence measures are strongly associated with a specific dependence for the choice of a reference measure.   
We explore effective choices for the reference measure to ensure that the corresponding divergences are tractable and can be expressed feasibly.
We focus on the \(\gamma\)-divergence as a robust MDE through an effective choice of the reference measure.

This chapter is organized as follows.
Section \ref{M-estimator} gives an overview of M-estimation in a framework of generalized linear model.
In Section \ref{gamma-loss-sec} the $\gamma$-divergence is introduced in a regression model and the $\gamma$-loss function is discussed.
In Section \ref{normal-reg-sec} we focus on the $\gamma$-estimator in a normal linear regression model.
A simple numerical example demonstrates a robust property of the $\gamma$-estimator compared to the ML-estimator.
Section \ref{bin-logistic-sec} discusses a logistic model for a binary regression. 
The $\gamma$-loss function is shown to have a robust property where the Euclidean distance of the estimating function  
to the decision boundary is uniformly bounded when $\gamma$ is in a specific range.
In Section \ref{subsec-Multiclass} extends the result in the binary case to a multiclass case.  
Section \ref{Poisson-reg-sec} considers a Poisson regression model focusing on a log-linear model.
The $\gamma$-divergence is given by a specific choice of the reference measure.
The robustness for the $\gamma$-estimator is confirmed for any $\gamma$ in the specific range.
A simple numerical experiment is conducted.
Finally, concluding remarks for geometric understandings are given in \ref{concluding}.


\section{M-estimators in a generalized linear model}\label{M-estimator}

Let us establish a probabilistic framework for a $d$-dimensional covariate variable $X$ in a subset $\mathcal X$ of $\mathbb R^d$, and an outcome $Y$ with a value of a subset $\mathcal Y$ of $\mathbb R$ in a regression model paradigm.
The major objective is to estimate the regression function 
\begin{align}\nonumber
r(x)=\mathbb E[Y|X=x]\end{align}
based on a given dataset.
In a paradigm of prediction, $X$ is often called a feature vector, where to build a predictor defined by a function from $\mathcal X$ to $\mathcal Y$ is one of the most important tasks.    
Let ${\mathcal P}(x)$  be a space of conditional probability measures conditioned on $X=x$.
For any event $B$ in $\mathcal Y$ the conditional probability given $X=x$ is written by 
\begin{align}\nonumber
P(B|x)=\int_B p(y|x){\rm d}\Lambda(y),
\end{align}
where $p(y|x)$ is  the RN-derivative of $Y=y$ given $X=x$ with a reference measure $\Lambda$.
A statistical model embedded in ${\mathcal P}(x)$ is written as
 \begin{align}\label{MODEL}
{{\mathcal M}}=\{P(\cdot|x,\theta):\theta\in\Theta\},
\end{align}
where $\theta$ is a parameter of a parameter space $\Theta$.
Then, the Kullback-Leibler (KL) divergence on ${\mathcal P}(x)$ is given by
 \begin{align}\nonumber
D_0(P_0(\cdot|x),P_1(\cdot|x))=H_0(P_0(\cdot|x),P_1(\cdot|x))-D_0(P_0(\cdot|x),P_0(\cdot|x)),
\end{align}
with the cross entropy,
 \begin{align}\nonumber
H_0(P_0(\cdot|x),P_1(\cdot|x))=-\int_{\mathcal Y}{p_0(y|x )}\log{p_1(y|x)}{\rm d}\Lambda(y).
\end{align}
Note that the KL-divergence is independent of the choice of reference measure $\Lambda$ as discussed in Chapter \ref{Power divergence}. 
Let ${\mathcal D}=\{(X_i,Y_i):1\leq i\leq n\}$ be a random sample drawn from a distribution of  ${\mathcal M}$.
The goal is to estimate the parameter $\theta$ in ${\mathcal M}$ in    \eqref{MODEL}.
Then, the negative log-likelihood function is defined by
 \begin{align}\nonumber
L_0(\theta;\Lambda)=-\frac{1}{n}\sum_{i=1}^n p(Y_i|X_i,\theta),
\end{align}
where $p(y|x,\theta)$ is the RN-derivative of $P(\cdot,\theta)$ with respect to $\Lambda$.
Note that
, for any measure $\tilde\Lambda$ equivalent to $\Lambda$ the negative log-likelihood functions $L_0(\theta;\tilde\Lambda)=L_0(\theta;\Lambda)$ up to a constant.  
The expectation of $L_0(\theta;\Lambda)$ under the model distribution of $\mathcal M$ is equal
to the cross entropy:
\begin{align}\label{Pytha}
\mathbb E_{0}[L_0(\theta;\Lambda)|\underline{X}] =
\frac{1}{n}\sum_{i=1}^n H_{0}(P(\cdot|X_i,\theta_0),P(\cdot|X_i,\theta))
\end{align}
where $\underline{X}=(X_1,...,X_n)$ and
$\theta_0$ is the true value of the parameter and $\mathbb E_{0}$ is
the conditional expectation under the model distribution $P(\cdot|X_i,\theta_0)$'s.
Hence,
\begin{align}\nonumber
\mathbb E_{0}[L_0(\theta;\Lambda)|\underline{X}]- \mathbb E_{0}[L(\theta_0;\Lambda)|\underline{X}] =
\frac{1}{n}\sum_{i=1}^n D_{0}(P(\cdot|X_i,\theta_0),P(\cdot|X_i,\theta)).
\end{align}
which can be viewed as an empirical analogue of the Pythagorean equation.
Due to the property of the KL-divergence as a divergence measure,  
\begin{align}\nonumber
\theta_0=\argmin_{\theta\in\Theta}\mathbb E_0[L_0(\theta;\Lambda)|\underline{X}]  .
\end{align} 
By definition, the ML-estimator $\hat\theta_0$ is the minimizer of $L_0(\theta;\Lambda)$ in $\theta$; while the true value $\theta_0$ is the minimizer of
$E_0[L_0(\theta;\Lambda)|\underline{X}]$ in $\theta$. 
The continuous mapping theorem reveals the consistency of the ML-estimator for the true parameter, see 
\cite{mann1943statistical,van2000asymptotic}    
The estimating function is defined by the gradient of the negative log-likelihood function
\begin{align}\nonumber  
{S}_0(\theta;\Lambda)=\frac{\partial}{\partial\theta}L_0(\theta;\Lambda).
\end{align}
Hence, the ML-estimator $\hat\theta$ is a solution of the estimating equation, ${S}_0(\theta;\Lambda)=0$ under regularity conditions.
We note that the solution of the expected estimating function under the distribution with the true value $\theta_0$ is $\theta_0$ itself, that is, 
\begin{align}\nonumber
\mathbb E_0[{S}_0(\theta_0;\Lambda)|\underline{X}]=0.
\end{align}
This implies that the continuous mapping theorem again concludes the consistency of the ML-estimator for the true value $\theta_0$.

The framework of  a generalized linear model (GLM) is suitable for a wide range of data types other than  the ordinary linear regression model, see \cite{McCullagh1989}. 
While the ordinary linear regression usually assumes that the response variable is normally distributed, GLMs allow for response variables that have different distributions, such as the Bernoulli, categorical, Poisson, negative binomial distributions
and exponential families in a unified manner.
In this way, GLMs provide excellent applicability for a wide range of data types, including count data, binary data, and other types of skewed or non-continuous data.
A GLM consists of three main components:

\begin{itemize}
\item[1]
Random Component: Specifies the probability distribution of the response variable \( Y \). This is typically a member of the exponential family of distributions (e.g., normal, exponential, binomial, Poisson, etc.).

\item[2]
 Systematic Component: Represents the linear combination of the predictor variables, similar to ordinary linear regression. It is usually expressed as \( \eta = \theta^\top X \).

\item[3]
Link Function: Provides the relationship between the random and systematic components. The expected value of \( Y \) given $X=x$, or the regression function is one-to-one with the linear combination of predictors  \( \eta = \theta^\top x \) through the link function $g$.

\end{itemize}
In the framework of the GLM, an exponential dispersion model is employed as 
\begin{align}\nonumber
p_{\exp}(y,\omega,\phi)= \exp\Big\{\frac{y \omega-a(\omega)}{\phi}+c(y,\phi)\Big\},
\end{align}
with respect to a reference measure $\Lambda$, where $\omega$ and $\phi$ is called the canonical and the  dispersion parameters, respectively, see \cite{jorgensen1987exponential}. 
Here we assume that $\omega$ can be defined in $(-\infty,\infty)$.
This allows for a linear modeling $\omega=g(\theta^\top x)$ with a flexible form of the link function $g$.
Specifically, if $g$ is an identity function, then $g$ is referred to as the canonical link function.
This formulation involves most of practical models in statistics such as the logistic and the log linear models.
In practice, the dispersion parameter $\phi$ is usually estimated separately from $\theta$, and hence, we assume  $\phi$  is known to be $1$ for simplicity.
This leads to a generalized linear model:
\begin{align}\label{exp-dispersion}
p(y|x, \theta)= \exp\big\{y \omega-a(\omega)+c(y)\big\}
\end{align} 
with $\omega=g(\theta^\top x)$ as the conditional RN-derivative of $Y$ given $X=x$. 
The regression function is given by
\begin{align}\nonumber 
\mathbb E[Y| X=x, \theta)=a^\prime(g(\theta^\top x))
\end{align}
due to the Bartlett identity.

Let us consider M-estimators for a parameter $\theta$ in the linear model \eqref{exp-dispersion}.
Originally, the M-estimator is introduced to cover robust estimators of a location parameter, see
\cite{huber1992robust} for breakthrough ideas for robust statistics,
and \cite{rousseeuw2005robust} for robust regression. 
We define an M-type loss function for the GLM defined in \eqref{exp-dispersion}:
\begin{align}\label{M-type}
\bar L_\Psi(\theta,{\cal D})=\frac{1}{n}\sum_{i=1}^n \Psi(Y_i,\theta^\top X_i)
\end{align}
for a given dataset ${\cal D}=\{(X_i,Y_i)\}_{i=1}^n$ and we call 
\begin{align}\nonumber
\hat\theta_\Psi:=\argmin_{\theta\in\mathbb R^d} \bar L_\Psi(\theta,{\cal D})
\end{align}
 the M-estimator.
Here the generator function $\Psi(y,s)$ is assumed to be convex with respect to $s$.
If $\Psi(y,s)=yg(s)-a(g(s))$, then the M-estimator is nothing but the ML estimator.
Thus, the estimating function is given by 
\begin{align}\label{psi-est-function}
\bar S_\psi(\theta,{\cal D})=\frac{1}{n}\sum_{i=1}^n \psi(Y_i,\theta^\top X_i)X_i,
\end{align} 
where $\psi(y,s)=(\partial/\partial s)\Psi(y,s)$.
If we confine the generator function $\Psi$ to a form of $\Psi(y-s)$, then this formulation reduces to the original form of  M-estimation \cite{huber2011robust,valdora2014robust}.
In general, the estimating function is characterized by $\psi(Y,\theta^\top X)$.
Hereafter we assume that
\begin{align}\nonumber 
\mathbb E_\theta[\psi(Y,\theta^\top X)|X=x]=0,
\end{align}
where $\mathbb E_\theta$ is the expectation under $p(y|x,\theta)$.
This assumption leads to consistency for the estimator $\hat\theta_\Psi$.
We note that the relationship between the loss function and the estimating function is not one-to-one. 
Indeed, there exist many choices of the estimating function for obtaining the estimator $\hat\theta_\Psi$ other than \eqref{psi-est-function}. 
We have a geometric discussion for an unbiased estimating function.

We investigate the behavior of the score function ${S}_\gamma(x,y,\theta)$  of the $\gamma$-estimator.
By definition, the $\gamma$-estimator is the solution such that the sample mean of the score function is equated to zero. 
We write a linear predictor as $F_\theta(x)=\theta_1^\top x+\theta_0$, where $\theta=(\theta_0,\theta_1)$.
We call 
\begin{align}\label{boundary}
H_\theta=\{x\in \mathbb R^p:F_\theta(x)=0\}
\end{align} 
the prediction boundary.
Then, the following formula is well-known in the Euclidean geometry.

\begin{proposition}\label{dist-formula}
Let $x\in{\cal X}$ and $d(x,H_\theta)$ be the Euclidean distance from $x$ to the prediction boundary $H_\theta$
defined in \eqref{boundary}.
Then,
\begin{align}\label{length}
d(x,H_\theta)=\frac{|F_\theta(x)|}{\|\theta_1\|}
\end{align}
\end{proposition}

\begin{proof}
Let $x^*$ be the projection of $x$ onto $H_\theta$.
Then, $d(x,H_\theta)=\|x-x^*\|$, where $\|\ \|$ denotes the Euclidean norm.
There exists a non zero scalar $\tau$ such that
$
x-x^*=\tau \theta_1
$ noting that a  normal vector to the hyperplane  $H_\theta$ is given by $\theta_1$.
Hence,
$ 
\theta_1^\top (x-x^*) =\tau \|\theta_1\|^2
$ 
and 
\begin{align}
d(x,H_\theta)=|\tau|\| \theta_1\|=\frac{|\theta_1^\top (x-x^*)|}{ \|\theta_1\|},
\end{align}
which concludes \eqref{length} since 
${|\theta_1^\top (x-x^*)|}=|F_\theta(x)|$ due to $F_\theta(x^*)=0$.

\end{proof}

Thus, a covariate vector $X$ of $\cal X$ is decomposed into  the orthogonal and horizontal components as $X= Z_\theta(X)+W_\theta(X)$, where 
 \begin{align}\label{ortho1}
Z_\theta(X)=\frac{\theta^{\top}X}{\|\theta\|^2}\>\theta \ \ {\rm and}\  \
W_\theta(x)=X-Z_\theta(X).
\end{align} 
We note that $Z_\theta(X)^\top W_\theta(X)=0$ and $\|X\|^2=\|Z_\theta(X)\|^2+\|W_\theta(X)\|^2$.
Due to the orthogonal decomposition \eqref{ortho1} of $X$, the estimating function is also decomposed into 
\begin{align}\nonumber 
S_\psi(Y,X,\theta)= S^{\rm (O)}_\psi(Y,X,\theta)+S^{\rm (H)}_\psi(Y,X,\theta),
\end{align} 
where
\begin{align}\nonumber 
S^{\rm (O)}_\psi(Y,X,\theta)=\psi(Y,\theta^\top Z_\theta(X))Z_\theta(X)    ,{\ \ }S^{\rm (H)}_\psi(Y,X,\theta)=\psi(Y,\theta^\top Z_\theta(X))W_\theta(X).
\end{align} 
Here we use  a property: $\theta^\top Z_\theta(X)=\theta^\top X$.
Thus, in  $S^{\rm (O)}_\psi(Y,X,\theta)$, $\psi(Y ,\theta^\top Z_\theta(X ))$  and $Z_\theta(X)$ are strongly connected   each other;
in  $S^{\rm (H)}_\psi(Y,X,\theta)$,  $\psi(Y,\theta^\top Z_\theta(X))$ and $W_\theta(X)$  are less connected. 

The estimating function  \eqref{psi-est-function}
is decomposed into a sum of the orthogonal and horizontal components,
\begin{align}\nonumber %
\bar S_\psi(\theta,{\cal D})=\bar S^{\rm (O)}_\psi(\theta,{\cal D})+\bar S^{\rm (O)}_\psi(\theta,{\cal D}), 
\end{align}
where
\begin{align}\nonumber %
\bar S^{\rm (O)}_\psi(\theta,{\cal D})= \frac{1}{n}\sum_{i=1}^n S^{\rm (O)}_\psi(Y_i,X_i,\theta), \ \ 
\bar S^{\rm (H)}_\psi(\theta,{\cal D})=
\frac{1}{n}\sum_{i=1}^n S^{\rm (H)}_\psi(Y_i,X_i,\theta).
\end{align}
We consider a specific type of contamination in the covariate space $\cal X$.
\begin{proposition}
Let ${\cal D}=\{(X_i,Y_i)\}_{i=1}^n$ and ${\cal D}^*=\{(X^*_i,Y_i)\}_{i=1}^n$, where
$X^*_i=X_i+\sigma(X_i) W_\theta(X_i)$ with arbitrarily a fixed scalar $\sigma(X_i)$ depending on $X_i$.   Then,
$\bar L_\Psi(\theta,{\cal D}^*) = \bar L_\Psi(\theta,{\cal D})$, $\bar S^{\rm (O)}_\psi(\theta,{\cal D}^*) = \bar  S^{\rm (O)}_\psi(\theta,{\cal D})$ and
\begin{align}\nonumber 
\bar S^{\rm (H)}_\psi(\theta,{\cal D}^*) = \frac{1}{n}\sum_{i=1}^n \psi(Y_i,\theta^\top Z_\theta(X_i))(1+\sigma(X_i))W_\theta(X_i).
\end{align}
\end{proposition}
\begin{proof}
By definition, $Z_\theta(X^*_i)=Z_\theta(X_i)$ and $W_\theta(X^*_i)=(1+\sigma(X_i))W_\theta(X_i)$ due to $Z_\theta(X_i)^\top W_\theta(X_i)=0$.  These imply the conclusion.
\end{proof}
We observe that  $\bar L_\Psi(\theta,{\cal D}^*)$ and 
$\bar S_\psi^{(O)}(\theta,{\cal D}^*)$ have both no influence with the contamination in 
$\{X_i^{*}\}$.
Alternatively, $\bar S_\psi^{(H)}(\theta,{\cal D}^*)$ has a substantial influence by scalar multiplication.
Hence, we can change the definition of the horizontal component as
\begin{align}\nonumber 
\bar{\bar S}^{\rm (H)}_\psi(\theta,{\cal D}^*) = \frac{1}{n}\sum_{i=1}^n \psi(Y_i,\theta^\top Z_\theta(X_i))\frac{W_\theta(X_i)}{\|W_\theta(X_i)\|}
\end{align}
 choosing as $\sigma(X_i)={\|W_\theta(X_i)\|}^{-1}-1$.
Then, it has a mild behavior such that
\begin{align}\nonumber 
\|\bar{\bar S}^{\rm (H)}_\psi(\theta,{\cal D}^*)\}\| \leq \frac{1}{n}\sum_{i=1}^n |\psi(Y_i,\theta^\top Z_\theta(X_i))|
\end{align}
In this way, the estimating function \eqref{psi-est-function} of M-estimator $\hat\theta_\Psi$ can be written as
\begin{align}\label{tilde-S} 
\tilde S_\psi(\theta,{\cal D})=\frac{1}{n}\sum_{i=1}^n \psi(Y_i,\theta^\top X_i)
\Big\{Z_\theta(X_i)+\frac{W_\theta(X_i)}{\|W_\theta(X_i)\|}\Big\}.
\end{align}
\begin{proposition}
Assume there exists a constant $c$ such that
\begin{align}\nonumber 
  \sup_{(y,s)\in{\cal Y}\times \mathbb R} |\psi(y,s)s| =c.
\end{align}
Then, the estimating function $\tilde S_\psi(\theta,{\cal D})$ in \eqref{tilde-S} of the M-estimator $\hat\theta_\Psi$ is bounded with respect to any dataset $\cal D$.
\end{proposition}
\begin{proof}
It follows from the assumption such that there exists a constant $c_1$ such that
$c_1=  \sup_{(y,s)\in{\cal Y}\times \mathbb R} |\psi(y,s)|$ since
\begin{align}\nonumber
 |\psi(y,s)|\leq  \sup_{(y,s)\in{\cal Y}\times[-1,1]}|\psi(y,s)|+ \sup_{(y,s)\in{\cal Y}\times\mathbb R}|\psi(y,s) s|.
\end{align}
Therefore, we observe
\begin{align}\nonumber
 \sup_{\cal D}\|\tilde S_\psi(\theta,{\cal D})\|&\leq
\frac{1}{n}\sum_{i=1}^n |\psi(Y_i,\theta^\top X_i)|
\big\{\|Z_\theta(X_i)\|+1\big\}\\[3mm]\nonumber
&\leq \frac{1}{\|\theta\|}\sup_{(y,s)\in{\cal Y}\times \mathbb R} |\psi(y,s)s|
+\sup_{(y,s)\in{\cal Y}\times \mathbb R} |\psi(y,s)|,
\end{align}
which is equal to $c/\|\theta\|+c_1$.

\end{proof}

On the other hand, suppose another type of contamination ${\cal D}^{**}=\{(X^{**}_i,Y_i)\}_{i=1}^n$, where
$X^{**}_i=X_i+\tau(X_i) Z_\theta(X_i)$ with a fixed scalar $\tau(X_i)$ depending on $\{X_i\}$.
Then,   $\bar L_\Psi(\theta,{\cal D}^*)$ and 
$\bar S_\psi^{(O)}(\theta,{\cal D}^*)$ have both  strong influences; $\bar S_\psi^{(H)}(\theta,{\cal D}^*)$ has no influence. 

The ML-estimator is a standard estimator that is defined by maximization of the likelihood for a given data set $\{(X_i,Y_i)\}_{i=1}^n$.   
In effect, the negative log-likelihood function is defined by
\begin{align}\nonumber
L_0(\theta;\Lambda)=-\frac{1}{n}\sum_{i=1}^n \{Y_i g (\theta^\top X_i)-a(g (\theta^\top X_i))+c(Y_i)\}.
\end{align}
The likelihood estimating function is given by
\begin{align}\label{likeEstimating}
{S}_0(\theta;\Lambda)=-\frac{1}{n}\sum_{i=1}^n \{Y_i-a^\prime(g (\theta^\top X_i))\}
g^\prime (\theta^\top X_i)X_i.
\end{align}
Here the regression parameter $\theta$ is of our main interests. 
We note that the ML-estimator $\hat\theta_0$ can be obtained without the nuisance parameter $\phi$ even if it is unknown.  
In effect, there are some methods for estimating $\phi$ using the deviance and the Pearson $\chi^2$ divergence in a case where $\phi$ is unknown. 
The expected value of the negative log-likelihood conditional on $\underline X=(X_1,...,X_n)$ is given by
\begin{align}\nonumber
\mathbb E[L_0(\theta;\Lambda)|\underline X]=-\frac{1}{n}\sum_{i=1}^n \{a^\prime ( g (\theta^\top X_i))g (\theta^\top X_i)-a(g (\theta^\top X_i))\}
\end{align}
up to a constant since the conditional expectation is given by $\mathbb E[Y|X=x]=a^\prime ( g (\theta^\top x))$ due to a basic property of the exponential dispersion model \eqref{exp-dispersion}.


\section{The $\gamma$-loss function and its variants}\label{gamma-loss-sec} 
Let us discuss the the $\gamma$-divergence in the framework of regression  model based on the discussion in the general distribution setting of the preceding section.
The $\gamma$-divergence is given by
 \begin{align}\nonumber
D_\gamma(P(\cdot|X,\theta_0),P(\cdot|X,\theta_1);\Lambda)=
H_\gamma(P(\cdot|X,\theta_0),P(\cdot|X,\theta_1);\Lambda)-H_\gamma(P(\cdot|X,\theta_0),P(\cdot|X,\theta_0);\Lambda)
\end{align}
with the cross entropy,
 \begin{align}\nonumber
H_\gamma(P(\cdot|X,\theta_0),P(\cdot|X,\theta_1);\Lambda)=
-\frac{1}{\gamma} \int_{\mathcal Y}{p(y|X,\theta_0)}\big[\{{p^{(\gamma)}(y|X,\theta_1)}\}^\frac{\gamma}{\gamma+1}.
\end{align}
The loss function derived from the $\gamma$-divergence is 
\begin{align}\nonumber
L_\gamma(\theta;\Lambda)=-\frac{1}{n}\frac{1}{\gamma}\sum_{i=1}^n\{p^{(\gamma)}(Y_i|X_i,\theta)\}^{\frac{\gamma}{\gamma+1}},
\end{align}
where $p^{(\gamma)}(y|x,\theta)$ is the $\gamma$-expression of $p(y|x,\theta)$, that is
\begin{align}\label{g-model}
p^{(\gamma)}(y|x,\theta)=\frac{\{p(y|x,\theta)\}^{\gamma+1}}{\int \{p(\tilde y|x,\theta)\}^{\gamma+1}{\rm d}\Lambda(\tilde y)}.
\end{align}
We define the $\gamma$-estimator for the parameter $\theta$ by  $\hat\theta_\gamma=\argmin_{\theta\in\Theta}L_\gamma(\theta;\Lambda)$.
By definition, the $\gamma$-estimator reduces to the ML-estimator when $\gamma$ is taken a limit to $0$.
\begin{remark}
Let us discuss a behavior of the $\gamma$-loss function when $|\gamma|$ becomes larger in which the outcome $Y$ is finite-discrete in $\cal Y$.
For simplicity, we define the loss function as 
\begin{align}\nonumber
L_\gamma(\theta;\Lambda)=-{\rm sign}( \gamma)\sum_{i=1}^n\{p^{(\gamma)}(Y_i|X_i,\theta)\}^{\frac{\gamma}{\gamma+1}}.
\end{align}
Let $f(x,\theta)=\argmax_{y\in{\cal Y}}p(y|x,\theta)$ and  $g(x,\theta)=\argmin_{y\in{\cal Y}}p(y|x,\theta)$.
Then, the $\gamma$-expression satisfies
\begin{align}\nonumber
p^{(\infty)}(y|x,\theta)&:=\lim_{\gamma\rightarrow\infty}p^{(\gamma)}(y|x,\theta)
\\[3mm] \nonumber
&=\lim_{\gamma\rightarrow\infty}\frac{\{p(y|x,\theta)/\max_{y^*\in{\cal Y}}p(y^*|x,\theta)\}^{\gamma+1}}{\sum_{\tilde y \in{\cal Y}} \{p(\tilde y|x,\theta)/\max_{y^*\in{\cal Y}}p(y^*|x,\theta)\}^{\gamma+1}}
\\[3mm]\nonumber
&={\rm I}(y=f(x,\theta))
\end{align}
Similarly,
\begin{align}\nonumber
p^{(-\infty)}(y|x,\theta)&:=\lim_{\gamma\rightarrow-\infty}p^{(\gamma)}(y|x,\theta)
\\[3mm] \nonumber
&=\lim_{\gamma\rightarrow-\infty}\frac{\{\min_{y^*\in{\cal Y}}p(y^*|x,\theta)/p(y|x,\theta)\}^{-\gamma-1}}{\sum_{\tilde y \in{\cal Y}}\{\min_{y^*\in{\cal Y}}p(y^*|x,\theta)/p(\tilde y|x,\theta)\}^{-\gamma-1}}
\\[3mm]\nonumber
&={\rm I}(y=g(x,\theta))
\end{align}
Hence, $L_\infty(\theta;\Lambda)$
is equivalent to the 0-1 loss function
$
\sum_{i=1}^n
{\rm I}(Y_i\ne f(X_i,\theta))
$
; while 
\begin{align}
L_{-\infty}(\theta;\Lambda)=\sum_{i=1}^n {\rm I}(Y_i=g(X_i,\theta))
\end{align}
This is the number of $Y_i$'s  equal to the worst predictor $g(X_i,\theta)$. 
If we focus on a case of ${\cal Y}=\{0,1\}$, then $L_{-\infty}(\theta;\Lambda)$ is nothing but the 0-1 loss function since ${\rm I}(y=g(x,\theta))={\rm I}(y\ne f(x,\theta))$.
In principle, the minimization of the 0-1 loss is hard due to the non-differentiability.
The $\gamma$-loss function smoothly connects the log-loss and the 0-1 loss without the computational challenge.
See \cite{friedman1997bias,nguyen2013algorithms} for detailed discussion for the 0-1 loss optimization. 

\end{remark}
In a subsequent discussion,  { the $\gamma$-expression} will play an important role on clarifying the statistical properties of the $\gamma$-estimator.
In fact, the $\gamma$-expression function is a counterpart of the log model function: $\log p(y|x,\theta)$ in $L_0(\theta;\Lambda)$.
Here we have a note as one of the most basic properties that
\begin{align}\label{H-g}
 -\frac{1}{\gamma}\mathbb E_0\Big[\{p^{(\gamma)}(Y|X,\theta)\}^{\frac{\gamma}{\gamma+1}}|X=x\Big]
 =H_{\gamma}(P(\cdot|x,\theta_0),P(\cdot|x,\theta);\Lambda),
\end{align}
Equation \eqref{H-g} yields 
\begin{align}\nonumber
\mathbb E_{0}[L_\gamma(\theta;\Lambda)|\underline{X}] =
\frac{1}{n}\sum_{i=1}^n H_{\gamma}(P(\cdot|X_i,\theta_0),P(\cdot|X_i,\theta);\Lambda).
\end{align}
and,  hence,
\begin{align}\label{Pytha-g}
\mathbb E_{0}[L_\gamma(\theta;\Lambda)|\underline{X}]- \mathbb E_{0}[L_\gamma(\theta_0;\Lambda)|\underline{X}] =
\frac{1}{n}\sum_{i=1}^n D_{\gamma}(P(\cdot|X_i,\theta_0),P(\cdot|X_i,\theta);\Lambda).
\end{align}
This implies
\begin{align}\nonumber
\theta_0=\argmin_{\theta\in\Theta}\mathbb E_0[L_\gamma(\theta;\Lambda)|\underline{X}].  
\end{align}
Thus, we observe due to the discussion similar to that for the ML-estimator and the KL-divergence that $\hat\theta_\gamma $ consistent for $\theta_0$.
The $\gamma$-estimating function is defined by
\begin{align}\nonumber 
{S}_\gamma(\theta;\Lambda) =\frac{\partial}{\partial\theta}L_\gamma(\theta;\Lambda).
\end{align}
Then, we have a basic property that the $\gamma$-estimating function should satisfy in general.

\begin{proposition}
The true value of the parameter is the solution of the expected $\gamma$-estimating equation under the expectation of the true distribution.
That is, if $\theta=\theta_0$,  
\begin{align}\label{00}
\mathbb E_0[{S}_\gamma(\theta;\Lambda)|\underline X]=0 ,
\end{align}
where $\mathbb E_{0}$ is
the conditional expectation under the true distribution $P(\cdot|X_i,\theta_0)$'s given $\underline X$.
\end{proposition}
\begin{proof}
By definition,
\begin{align}\nonumber
{S}_\gamma(\theta;\Lambda) =-\frac{1}{n}\frac{1}{\gamma+1}\sum_{i=1}^n \{p^{(\gamma)}(Y_i|X_i,\theta)\}^{-\frac{1}{\gamma+1}}\frac{\partial}{\partial\theta}p^{(\gamma)}(Y_i|X_i,\theta).
\end{align}
Here we  note 
\begin{align}\nonumber
\{p^{(\gamma)}(Y_i|X_i,\theta)\}^{-\frac{1}{\gamma+1}}=\frac{1}{p(Y_i|X_i,\theta)}
\end{align}
up to a proportionality constant.
Hence,
\begin{align}\nonumber
\mathbb E_0[{S}_\gamma(\theta;\Lambda)|\underline X]\propto\sum_{i=1}^n \int_{\mathcal Y}
\frac{p(y|X_i,\theta_0)}{p(y|X_i,\theta)}\frac{\partial}{\partial\theta}p^{(\gamma)}(y|X_i,\theta){\rm d}\Lambda(y).
\end{align}
If $\theta=\theta_0$, then this vanishes identically due to the total mass one of $p^{(\gamma)}(y|X_i,\theta_0)$.
\end{proof}
The $\gamma$-estimator $\hat\theta_\gamma$ is a solution of the estimating equation; while
true value $\theta_0$ is the solution of the expected estimating equation under the true distribution with $\theta_0$.
Similarly, this shows the consistency of the $\gamma$-estimator for the true value of the parameter.
The $\gamma$-estimating function ${S}_\gamma(\theta;\Lambda)$ is said to be unbiased in the sense of \eqref{00}. 
Such a unbiased property leads to the consistency of the estimator.
However, if the underlying distribution is misspecified, then we have to evaluate the expectation in \eqref{00} under the misspecified distribution other than the true distribution.
Thus, the unbiased property is generally broken down, and the Euclidean norm of the estimating function may be divergent at the worst case.  
We will investigate such behaviors in misspecified situations later.

Now, we consider the MDEs via the GM and HM divergences introduced in Chapter \ref{Power divergence}.
First, consider the loss function defined by the GM-divergence:
\begin{align}\nonumber
L_{\rm GM }(\theta;R)=\frac{1}{n}\sum_{i=1}^n
\frac{ r(Y_i )}{p(Y_i|X_i,\theta)} \exp\Big\{\int  \log p(y|X_i,\theta)d R(y )\Big\}.
\end{align}
where $R$ is the reference  probability measure in ${\mathcal P}(x)$.
We define as $\hat\theta_{\rm GM }=\argmin_{\theta\in\Theta}L_{\rm GM }(\theta;R)$, which we refer  to as the GM-estimator.
The $\rm GM $-estimating equation is given by
\begin{align}\nonumber
{S}_{\rm GM } (\theta;R)&:=\frac{1}{n} \sum_{i=1}^n 
\frac{ r(Y_i)}{p(Y_i|X_i,\theta)} \exp\Big\{\int  \log p(y|X_i,\theta)d R(y)\Big\}\\[3mm]
&\times\Big\{S(Y_i|X_i,\theta) - \int  S(y|X_i,\theta)d R(y)\Big\}=0,
\end{align}
where $S(y|x,\theta)=(\partial/\partial\theta)\log p(y|x,\theta)$.
Secondly, consider the loss function defined by the HM-divergence:
\begin{align}\nonumber
L_{\rm HM }(\theta)= \frac{1}{2n}\sum_{i=1}^n \{p^{(-2)}(Y_i|X_i,\theta)\}^2.
\end{align}
The $(-2)$-model can be viewed as an inverse-weighted probability model on the account of
 \begin{align}\nonumber
p^{(-2)}(y|x,\theta)= \frac{\small\displaystyle\frac{1}{p(y|x,\theta)}}{\displaystyle{\large\mbox{$\sum_{j=0}^k$}} \ \frac{1}{p(j|x,\theta)}}
\end{align}
We define the HM estimator by $\hat\theta_{\rm HM }=\argmin_{\theta\in\Theta}L_{\rm HM }(\theta)$.
The $\rm HM $-estimating equation is given by
\begin{align}\nonumber
{S}_{\rm HM } (\theta)=\frac{1}{n}\sum_{i=1}^n p^{(-2)}(Y_i|X_i,\theta)\frac{\partial}{\partial\theta}p^{(-2)}(Y_i|X_i,\theta)
\end{align}
We note from the discussion in Section 2 that $L_{\rm GM }(\theta;R)$ and ${S}_{\rm GM }(\theta;R)$
are equal to $L_{\gamma}(\theta;R)$ and ${S}_{\gamma}(\theta;R)$ with $\gamma=-1$;
$L_{\rm HM }(\theta)$ and ${S}_{\rm HM }(\theta)$
are equal to $L_{\gamma}(\theta;C)$ and ${S}_{\gamma}(\theta;C)$ with $\gamma=-2$.
We will discuss the dependence on the reference measure $R$, in which we like to elucidate which choice of $R$ gives a reasonable  performance in the presence of possible model misspecification.

We focus on the GLM framework in which we look into the formula on the $\gamma$-divergence.
Then, the choice of the reference measure should be paid attentions to the $\gamma$-divergence.
The original reference measure $\Lambda$ is changed to $R$ such that $\partial R/\partial \Lambda(y)=\exp\{c(y)\}.$
Hence, the model is given by $p(y|x,\omega)=\exp\{y\omega-a(\omega)\}$ withr respect to $R$.
We note that $\tilde\Lambda$ is a probability measure since the RN-derivative is equal to $ p(y|x, \theta)$ defined in \eqref{exp-dispersion} when $\theta$ is a zero vector.
This makes the model more mathematically tractable and allows us to use standard statistical methods for estimation and inference.
Then,  the $\gamma$-expression for $p(y|x,\omega)$ is given by
\begin{align}\label{gamma-express}
p^{(\gamma)}(y|x,\omega)=p(y|x,(\gamma+1)\omega).
\end{align}
This property of  reflexiveness is convenient the analysis based on the $\gamma$-divergence. 
First of all, the $\gamma$-loss function is given by
\begin{align}\label{closed}
L_\gamma(\theta;R)=-\frac{1}{n}\frac{1}{\gamma}\sum_{i=1}^n \exp\Big\{\gamma \,Y_i\,g(\theta^\top X_i)-\frac{\gamma}{\gamma+1} a\big((\gamma+1) g(\theta^\top X_i)\big)\Big\}
\end{align}
due to the $\gamma$-expression \eqref{gamma-express}.
The $\gamma$-estimating function is given by
\begin{align}\label{gammaEstimating}
{S}_\gamma(\theta;R)=& \frac{1}{n}\sum_{i=1}^n \exp\Big\{\gamma \,Y_i\,g(\theta^\top X_i)-\frac{\gamma}{\gamma+1} a\big((\gamma+1) g(\theta^\top X_i)\big)\Big\}\\[2mm]\nonumber
&\times \{Y_i - a^\prime \big((\gamma+1) g(\theta^\top X_i)\big)\}g^\prime(\theta^\top X_i) X_i.
\end{align}
We note that the change of the reference measure from $\Lambda$ to  $\tilde\Lambda$ is the key for the minimum 
$\gamma$-divergence estimation.
In fact,  the $\gamma$-loss function would not have a closed form as \eqref{closed} unless the reference measure  is changed.
Here, we remark that the $\gamma$-loss function is a specific example of M-type loss function
$\bar L_\Psi(\theta)$ in \eqref{M-type} with a relationship of 
\begin{align}\nonumber
\Psi(y,s)=\exp\Big\{\gamma y g(s)-\frac{\gamma}{\gamma+1} a\big((\gamma+1) g(s)\big)\Big\}.
\end{align}
The expected $\gamma$ loss function  is given by
\begin{align}\nonumber
\mathbb E_0[L_\gamma(\theta;R)|\underline X]=-\frac{1}{n}\sum_{i=1}^n 
\exp\Big\{a\big(\gamma g(\theta^\top X_i)\big)-\frac{\gamma}{\gamma+1} a\big((\gamma+1) g(\theta^\top X_i)\big)\Big\},
\end{align}
where $\mathbb E_0$ denotes the expectation under the true distribution $P(\cdot|x,\theta_0)$.
This function attains a global minimum at $\theta=\theta_0$ as discussed around \eqref{Pytha-g} in the general framework.
Similarly, the GM-loss function is written by
\begin{align}\nonumber
L_{\rm GM }(\theta;R)=\frac{1}{n}\sum_{i=1}^n
{ r(Y_i )}\exp\big\{- (Y_i-\mu_R)\,g(\theta^\top X_i)+a\big( g(\theta^\top X_i)\big)\big\},
\end{align}
where $\mu_R=\int y r(y){\rm d}\tilde \Lambda(y)$.
The HM loss function is written by
\begin{align}\nonumber
L_{\rm HM }(\theta)= \frac{1}{2n}\sum_{i=1}^n \exp\{-2Y_i-2a\big(\!\!-\!g(\theta^\top x)\big)\}.
\end{align}
since the $\gamma$-expression  becomes
$p^{(\gamma)}(y|x,\theta)= \exp\{-y\,g(\theta^\top x)-a\big(\!\!-\!g(\theta^\top x)\big)\}$ when $\gamma=-2$.
In accordance with these, all the formulas for the loss functions defined in the general model \eqref{MODEL}
are reasonably transported in GLM.  Subsequently, we go on the specific model of GLM to discuss deeper properties.

We have discussed the generalization of the $\gamma$-divergence in the preceding section.
The generalized divergence $D_V(P,Q:\Lambda)$ defined in \eqref{general} in Chapter 2 yields the loss function by
\begin{align}\nonumber
L_V(\theta;\Lambda)=-\frac{1}{n}\sum_{i=1}^nV(v^*(z(\theta,X_i)p(Y_i|X_i,\theta))),
\end{align}
where $z(\theta,X_i)$ is a normalizing factor satisfying 
\begin{align}\label{z*}
\int_{\mathcal Y} v^*(z(\theta,X_i)p(y|X_i,\theta)){\rm d}\Lambda(y)=1.
\end{align}
The similar discussion as in the above conducts
\begin{align}\nonumber
\mathbb E_{0}[L_V(\theta;\Lambda)|\underline{X}]- \mathbb E_{0}[L_V(\theta_0;\Lambda)|\underline{X}] =
\frac{1}{n}\sum_{i=1}^n D_{V}(P(\cdot|X_i,\theta_0),P(\cdot|X_i,\theta);\Lambda).
\end{align}
The estimating function is written as
\begin{align}\nonumber
{S}_V(\theta;\Lambda)=-\frac{1}{n}\sum_{i=1}^n
\frac{1}{z(\theta,X_i)p(Y_i|X_i,\theta)}\frac{\partial}{\partial\theta}v^*(z(\theta,X_i)p(Y_i|X_i,\theta)).
\end{align}
due to assumption \eqref{assump1}.
This implies
\begin{align}\nonumber
\mathbb E_0[{S}_V(\theta_0;\Lambda)|\underline X]=-\frac{1}{n}\sum_{i=1}^n
\frac{1}{z(\theta,X_i)}\int \frac{\partial}{\partial\theta}v^*(z(\theta_0,X_i)p(y|X_i,\theta)){\rm d}\Lambda(y) 
\end{align}
which vanishes since all $v^*(z(\theta_0,X_i)p(y|X_i,\theta))$'s have total mass one as in \eqref{z*}.
Consequently, we can derive the MD estimator based on the generalized divergence $D_V(P,Q,\Lambda)$ with the $\gamma$-divergence as the standard.   In Section 3, we will consider another candidate of $D_V(P,Q,\Lambda)$ for estimation under a Poisson point process model.


\section{Normal linear regression}\label{normal-reg-sec}
Linear regression, one of the most familiar and most widely used statistical techniques, dates back to the 19-th century in the mathematical formulation by Carolus F. Gauss \cite{williams1998prediction}. 
It originally emerged from the eminent observation of Francis Galton on regression towards the mean at the begging of the 20-th century.
Thus, the ordinary least squares method is  evolved with the advancement of statistical theory and computational methods. 
As the application of linear regression expanded, statisticians recognized its sensitivity to outliers. 
Outliers can significantly influence the regression model's estimates, leading to misleading results.
 To address these limitations, robust regression methods were developed. These methods aim to provide estimates that are less affected by outliers or violations of model assumptions like normality of errors or homoscedasticity.

Let $Y$ be an outcome variable  in $\mathbb R$ and  $X$ be a covariate vector in a subset ${\mathcal X}$ of $\mathbb R^d$.
Assume the conditional probability density function (pdf) of $Y$ given $X=x$ as
\begin{align}\label{normal-reg}
p(y|X=x,\theta,\sigma^2)=\frac{1}{\sqrt{2\pi\sigma^2}}\exp\Big\{-\half  \frac{(y-\theta^\top x)^2}{\sigma^2} \Big\}.
\end{align}
Thus, the normal linear regression model \eqref{normal-reg} is one of the simplest examples of GLM with an identity link function
where $\sigma$ is a dispersion parameter.
Indeed, $\sigma$ is a crucial parameter for assessing model fit.
We will discuss the estimation for the parameter later.
The KL-divergence between normal distributions is given by
\begin{align}\nonumber
D_{0}({\tt Nor}(\mu_0,\sigma_0^2),{\tt Nor}(\mu_1,\sigma_1^2))
=\half \frac{(\mu_1-\mu_0)^2}{\sigma_1^2}
+\half \Big(\frac{\sigma_0^2}{\sigma_1^2}-\log \frac{\sigma_0^2}{\sigma_1^2}-1\Big).
\end{align}
For a given dataset ${(X_i,Y_i)}_{i=1}^n$, the negative log-likelihood function is as follows:
\begin{align}\nonumber
L_0(\theta)=\half \frac{1}{n}\sum_{i=1}^n \Big\{\frac{(Y_i-\theta^\top X_i)^2}{\sigma_0}
+\log(2\pi\sigma^2)\Big\}.
\end{align}
The estimating function for $\theta$ is 
 \begin{align}\nonumber
{S}_0(\theta)=\frac{1}{n}\frac{1}{\sigma^2}\sum_{i=1}^n (Y_i- \theta^\top X_i)X_i,
\end{align}
where $\sigma^2$ is assumed to be known.
In fact, it is estimated in a practical situation where $\sigma^2$ is unknown.
Equating the estimating function to zero gives the  likelihood equations in which the ML-estimator
is nothing but the least square estimator.
This is a well-known element in statistics with a wide range of applications, where several standard tools for assessing model fit and diagnostics have been established.

   
On the other hand, robust regression robust methods  aim to provide estimates that are less affected by outliers or violations of model assumptions like normality of errors.
The key is the introduction of M-estimators, which generalize maximum likelihood estimators. They work by minimizing a sum of a loss function applied to the residuals. The choice of the loss function (such as Huber's winsorized loss or Tukey's biweight loss \cite{beaton1974fitting}) determines the robustness and efficiency of the estimator.
The M-estimator, 
$\hat\theta_\Psi$, of a parameter 
$\theta$ is obtained by minimizing an objective function, typically defined by a sum of  
$\Psi$'s applied to the adjusted residuals:
 \begin{align}\label{rho}
\hat\theta_\Psi=
\argmin_{\theta\in\mathbb R^d}\sum_{i=1}^n \Psi\Big(\frac{Y_i-\theta^\top X_i}{\sigma }\Big).
\end{align}
The estimating equation is given by
\begin{align}\nonumber
\sum_{i=1}^n \psi\Big(\frac{Y_i-\theta^\top X_i}{\sigma }\Big)X_i=0,
\end{align}
where $\psi(r)=(\partial/\partial r)\Psi(r)$. 
Here  are typical examples:\\[3mm]
(1). Quadratic loss: $\Psi(r)=r^2$, which is equivalent to the log-likelihood function \\[3mm]
(2). Huber's loss:
$\Psi(r)= \left\{
\begin{array}{lc}\half r^2 & \text{for } |r| \leq k \\[2mm]
k(|r| - \half k) & \text{for } |r| > k
\end{array}\right. $\\[5mm]
(3). Tukey's  loss:
$\Psi(r)= \left\{
\begin{array}{lc}\frac{c^2}{6} \left(1 - \left[1 - \left(\frac{r}{c}\right)^2\right]^3\right) & \text{for } |r| \leq c \\[3mm]
\frac{c^2}{6} & \text{for } |r| > c
\end{array}\right.,$\\[3mm]
where $K$ and $c$ are hyper parameters.

We return the discussion for the $\gamma$-estimator.
The $\gamma$-divergence  is given by
\begin{align}\nonumber
D_\gamma({\tt Nor}(\mu_0,\sigma^2),{\tt Nor}(\mu_1,\sigma^2))
=c_\gamma (\sigma^2){}^\frac{\gamma}{\gamma+1}\Big[
\exp\Big\{-\half\frac{\gamma}{\gamma+1} \frac{(\mu_1-\mu_0)^2}{\sigma^2}\Big\}-1\Big],
\end{align}
where $c_\gamma=(\gamma+1)^{\half\frac{1}{\gamma+1}}$.
The $\gamma$-expression of the normal linear model  is given by
\begin{align}\nonumber
p^{(\gamma)}(y|x,\theta)=p_0(y, \theta^\top x,\sigma_0/(\gamma+1)),
\end{align}
where $p_0(y,\mu,\sigma^2)$ is a normal density function with mean $\mu$ and variance $\sigma^2$.
Hence, the $\gamma$-loss function is given by
\begin{align}\nonumber
L_\gamma(\theta)=-\frac{1}{n}\frac{1}{\gamma}
\sum_{i=1}^n \{p_0(Y_i, \theta^\top X_i,\sigma_0/(\gamma+1))\}^\frac{\gamma}{\gamma+1}.
\end{align}
which is written as
 \begin{align}\label{L-g-Nor} 
 -\frac{1}{n}\frac{1}{\gamma}
\sum_{i=1}^n 
\exp\Big\{-\half{\gamma}\frac{(Y_i-\theta^\top X_i)^2}{\sigma^2}{-\half\frac{\gamma}{\gamma+1}}\sigma^2\Big\}
\end{align}
up to a scalar multiplication. 
Consequently, the $\gamma$-loss function is a specific example of $\Psi$-loss function in \eqref{rho} viewing as $\Psi(r)\propto -(1/\gamma)\exp(-\half\gamma r^2)$.
We note that the $\gamma$-estimator is one of M-estimators.
The $\gamma$-estimating function is defined as
${S}_\gamma(\theta)=\frac{1}{n}\sum_{i=1}^n {S}_\gamma(X_i,Y_i,\theta),$
where the score function is defined by
 \begin{align}\label{score}
{S}_\gamma(x,y,\theta)= (2\pi\sigma^2)^{-\half\frac{\gamma}{\gamma+1}}
\exp\Big\{-\half{\gamma}\frac{(y-\theta^\top x)^2}{\sigma^2}\Big\}\frac{y-\theta^\top x}{\sigma_0}x.
\end{align}
The generator function is given as $\psi(r,\gamma)=r\exp(-\half \gamma r^2)$ as an M-estimator.
Fig \ref{GENE} displays the plots of the generator functions: \\[3mm]
(1). $\gamma$-loss, $\psi(r,\gamma)=r\exp(-\half \gamma r^2)$, \\[3mm]
(2). Huber's loss,  $\psi(r,k)={\mathbb I}(|r|\leq k)r+{\mathbb I}(|r|> k){\rm sign(r)}$ 
\\[3mm]
(3). Tukey's loss, $\psi(r,c)={\mathbb I}(|r|\leq k)r\{1-(r/c)^2\}^2$.\\  [3mm]
It is observed  the generator functions of the $\gamma$-loss and Tukey's loss are both redescending. This means the influence of each data point on the estimation decreases to zero beyond a certain threshold, effectively eliminating the impact of extreme outliers.
Unlike  the quadratic loss and Huber's loss functions, such redescending loss functions are non-convex. This characteristic makes it more robust but also introduces challenges in optimization, as it can lead to multiple local minimums.

\begin{figure}[htbp]
\begin{center}
  \includegraphics[width=90mm]{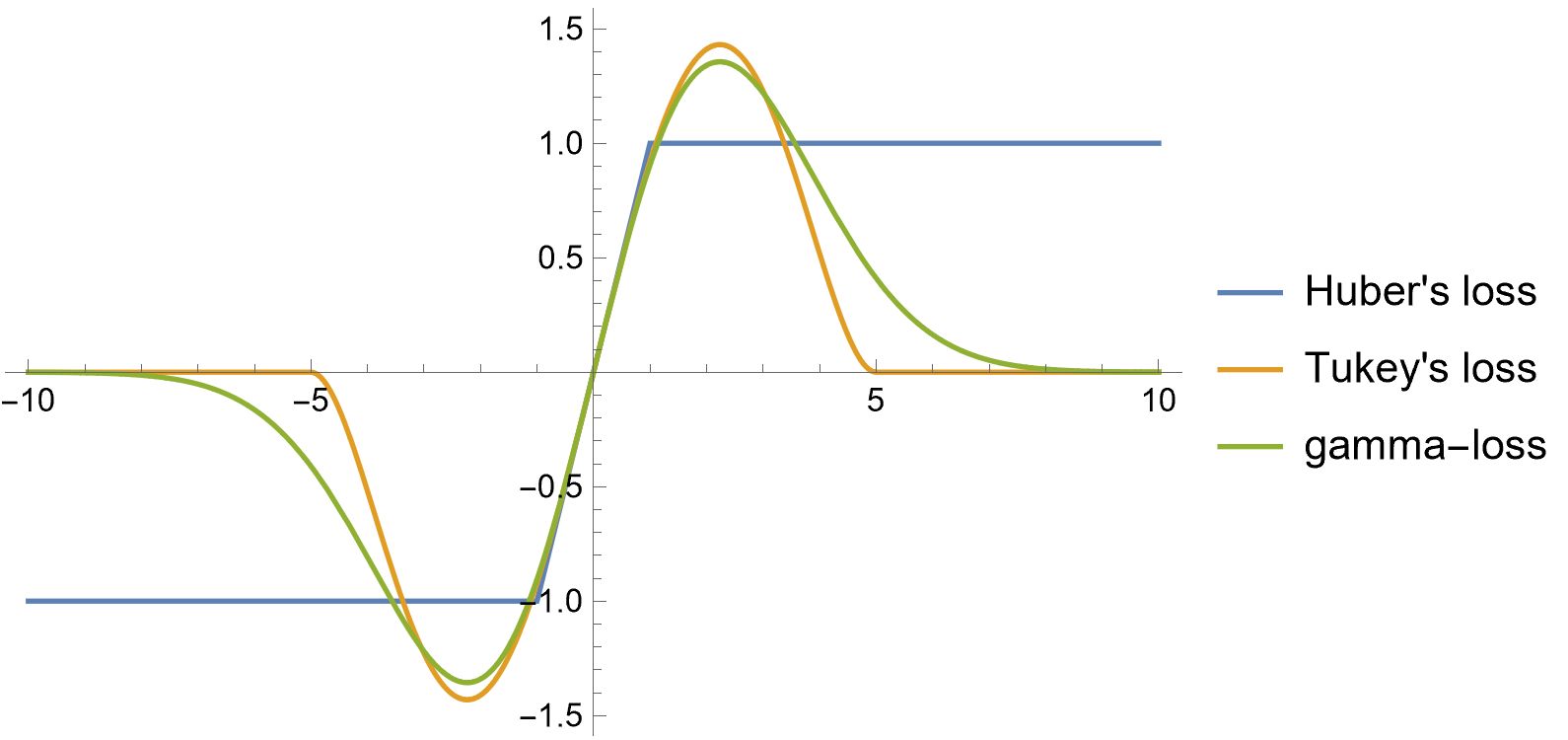}
\end{center}
 \vspace{-2mm}\caption{Plots of  the generator functions}
\label{GENE}
\end{figure} 

The variance parameter $\sigma^2$ in the normal regression model is referred to as a dispersion parameter in GLM.
In a situation where $\sigma^2$ is unknown the likelihood method is similar to the known case.
The ML-estimator for $\sigma^2$ is derived by 
\begin{align}\nonumber
 \hat\sigma^2=\half\frac{1}{n}\sum_{i=1}^n {(Y_i-\hat\theta_0^\top X_i)^2}
\end{align}
plugging $\theta$ in $\hat\theta_0$.
Alternatively, the $\gamma$-estimator for $(\theta,\sigma^2)$ is is derived by the solution of the joint estimating equation combining
 \begin{align}\nonumber
\sigma^2 =\frac{\gamma+1}{n}\sum_{i=1}^n \exp\Big\{-\half{\gamma}\frac{(Y_i-\theta^\top X_i)^2}{\sigma^2}\Big\}(Y_i-\theta^\top X_i)^2
\end{align}
with the estimating equation for $\theta$. 
Similarly, we can find that the boundedness property for the $\gamma$-score function for $\sigma^2$ holds. 

Let us apply the geometric discussion associated with the decision boundary $H_\theta$
in \eqref{boundary} to the normal regression model.
We write the estimating function of M-estimator in \eqref{rho} as
\begin{align}\nonumber
S_\psi(\theta,{\cal D})=\sum_{i=1}^n \psi\Big(\frac{Y_i-\theta^\top X_i}{\sigma_0}\Big)X_i
\end{align}
for a given dataset ${\cal D}=\{X_i,Y_i)\}_{i=1}^n$.
Due to the orthogonal decomposition  of $X$, the estimating function is also decomposed into a sum of the orthogonal and horizontal components,
$\bar S^{\rm (O)}_\psi(\theta,{\cal D})+\bar S^{\rm (H)}_\psi(\theta,{\cal D})$, where
\begin{align}\nonumber %
\bar S^{\rm (O)}_\psi(\theta,{\cal D})= \frac{1}{n}\sum_{i=1}^n  \psi\Big(\frac{Y_i-\theta^\top X_i}{\sigma_0}\Big)Z_\theta(X_i), \ \ 
\bar S^{\rm (H)}_\psi(\theta,{\cal D})=
\frac{1}{n}\sum_{i=1}^n  \psi\Big(\frac{Y_i-\theta^\top X_i}{\sigma_0}\Big)W_\theta(X_i).
\end{align}
We note that this decomposition is the same as that for the GLM in Section \label{M-estimator}.
We consider a specific type of contamination in the covariate space $\cal X$ such that ${\cal D}^*=\{(X^*_i,Y_i)\}_{i=1}^n$, where
$X^*_i=X_i+\sigma(X_i) W_\theta(X_i)$ with a fixed scalar $\sigma(X_i)$ depending on $X_i$.  
As in the discussion for the general setting of the GLM,   $\bar L_\Psi(\theta,{\cal D}^*)$ and 
$\bar S_\psi^{(O)}(\theta,{\cal D}^*)$ have both  strong influences; $\bar S_\psi^{(\rm H)}(\theta,{\cal D}^*)$ has no influence. 
Let us investigate a preferable property for the $\gamma$-estimator applying the decomposition formula above.

\begin{proposition}\label{RobustNor}
Let ${S}_\gamma(x,y,\theta)$ be the $\gamma$-score function defined in \eqref{score}.
Then, 
\begin{align}\label{sup_Nor}
\sup_{x\in{\mathcal X}}d({S}_\gamma(x,y,\theta),H_\theta)<\infty
\end{align}
for any fixed $y$ of $\mathbb R$ and any $\gamma>0$, where $d$ is the Euclidean  distance. 
\end{proposition}

\begin{proof}
It is written  that
\begin{align}\nonumber
d({S}_{\gamma}(x,y,\theta),H_\theta)
=\exp(-\half \gamma  z^2)|z(z-y)|,
\end{align}
where $z=y-\theta^\top x$.
Therefore,
\begin{align}\nonumber
\sup_{x\in{\mathcal X}}d({S}_{\gamma}(x,y,\theta),H_\theta)
\leq \exp(-\half {\gamma} z^2)(z^2+|yz|),
\end{align}
which is bounded by
\begin{align}\nonumber
\sup_{z>0}z^2\exp(-\half {\gamma} z^2)
+|y|\sup_{z>0}|z|\exp(-\half {\gamma} z^2).
\end{align}
This is simplified as
\begin{align}\nonumber
\frac{2}{\gamma}\exp(-1)
+|y|\frac{1}{\gamma}\exp(-\frac{1}{\gamma}).
\end{align}
Therefore,  \eqref{sup_Nor} is concluded for the fixed $y$.
\end{proof}
It follows from Proposition \ref{RobustNor}  that all the estimating scores of the $\gamma$-estimator  appropriately lies in a tubular neighborhood
\begin{align}\label{tube}
{\mathcal N}_\theta(\delta)=\big\{ z\in\mathbb R^d: d(z,{\mathcal H}_\theta)\leq\delta  \big\}
\end{align}
surrounding $H_\theta$.
As a result,  the distance from the estimating function to the boundary $H_\theta$ is
bounded, that is,
\begin{align}\nonumber
\sup_{x\in{\mathcal X}}d({S}_\gamma(\theta),{\mathcal H}_\theta)\leq \frac{2}{\gamma}\exp(-1)
+\frac{1}{\gamma}\exp(-\frac{1}{\gamma}) \max_{1\leq i\leq n}|Y_i|.
\end{align}
However, in the limit case of $\gamma=0$ or the ML-estimator, this boundedness property for covariate outlying is broken down.
Tukey's biweight loss estimating function satisfies the boundedness; Huber's loss estimating function does not satisfy that.

We have a brief study for numerical experiments.
Assume that covariate vectors $X_i$'s are generated from a bivariate normal distribution 
${\tt Nor}(0,{\rm I})$, where $\rm I$ denotes a 3-dimensional identity matrix.
This simulation was designed based on a scenario about the conditional distribution of the response variables 
$ Y_i $'s as follows.
\begin{description}
\item[Specified model] 
{$\hspace{9mm} Y_i\sim {\tt Nor}(\theta_1^\top X_i+\theta_0,\sigma).$}

\item[Misspecified model ]
{$Y_i\sim (1-\pi){\tt Nor}(\theta_1^\top X_i+\theta_0,\sigma)
+\pi{\tt Nor}(-\theta_{*1}^\top X_i+\theta_{*0},\sigma_{*}).$}

\end{description}
Here parameters were set as $(\theta_0,\theta_1)=(0.5,1.5,1.0)^\top$,  and  $\pi=0.1$ with $\sigma=1$; $(\theta_{0*},\theta_{1*})=(0.5,-1.5,-1.0)^\top$ with $\sigma_*=1$.

We compared the estimates the ML-estimator $\hat\theta_0$ and the $\gamma$-estimator $\hat\theta_{\gamma}$ with $\gamma=0.3$, where the simulation was conducted by 300 replications.
In the the case of specified model, the ML-estimator was slightly superior to the $\gamma$-estimator in a point of the
root means square estimate (rmse), however the superiority is almost negligible.
Next, we suppose a mixture distribution of two normal regression modes in which one was the same model as the above with the mixing probability $0.9$; the other was still a normal  regression model but the the minus  slope vector with the mixing probability $0.1$.
Under such a misspecified setting, $\gamma$-estimator was crucially superior to the ML-estimator, where the rmse of the ML-estimator is more than double that of the $\gamma$-estimator.
Thus, the ML-estimator is sensitive to the presence of such a heterogeneous subgroup; the $\gamma$-estimator is robust.
Proposition \ref{RobustNor} suggests that the effect of the subgroup is substantially suppressed in the estimating function of the $\gamma$-estimator.    
See Table \ref{gestNor} and Figure \ref{Boxplot} for details. 

\begin{table}[hbtp]

\caption{Comparison between the ML-estimator and the $\gamma$-estimator. }

  \centering
\vspace{3.2mm}

(a). The case of specified model

\vspace{3mm}
  \begin{tabular}{ccc}
  \hline 

    Method  & estimate  &  rmse
 \\
    \hline \hline

\vspace*{1mm}   
 ML-estimate & $(0.495672, 1.50753, 1.00211)$  & $ 0.173617$\\
    $\gamma$-estimate   & $(0.497593, 1.50754, 1.00301)$  & $0.176443$ \\

   \hline
  \end{tabular}
\vspace{5.2mm}

(b). The case of  misspecified model

\vspace{3mm}
  \begin{tabular}{ccc}
  \hline 

    Method  & estimate  &  rmse \vspace{1mm}
 \\
    \hline \hline

    \vspace*{1mm}   
 ML-estimate & $(0.50788, 1.19093, 0.774613)$  & $ 0.486919$\\
    $\gamma$-estimate   & $(0.500093, 1.43289, 0.941501)$  & $0.219798$ \\

   \hline
  \end{tabular}

\label{gestNor}
\end{table}

\

\

\begin{figure}[htbp]
\begin{center}
  \includegraphics[width=80mm]{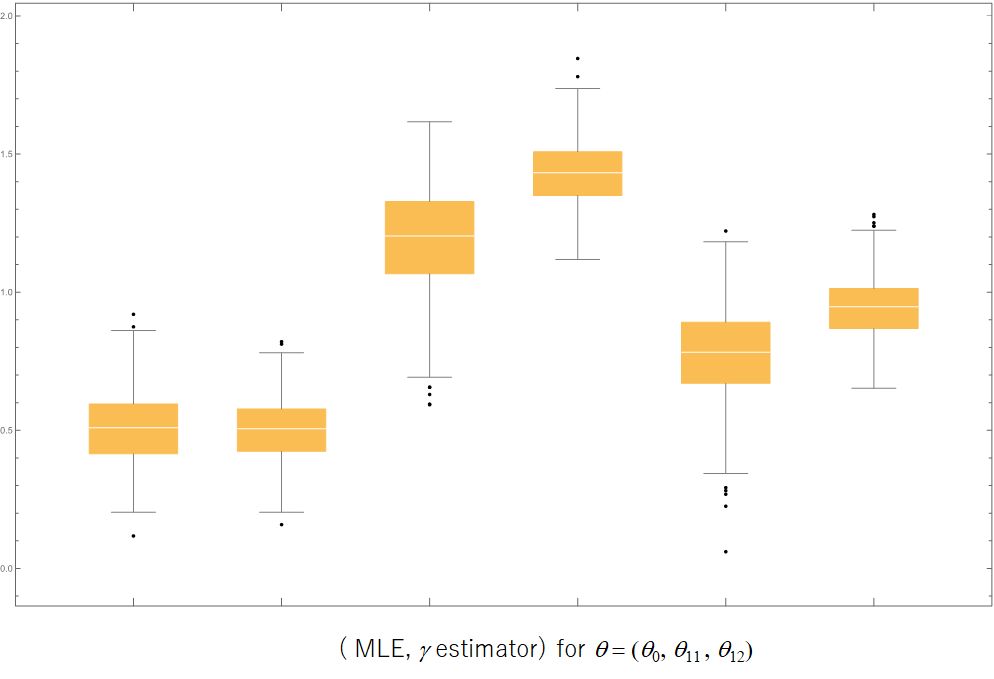}
 \end{center}
 \vspace{-2mm}\caption{Box-whisker Plots of  the ML-estimator and the $\gamma$-estimator}
\label{Boxplot}
\end{figure}


\section{Binary logistic regression}\label{bin-logistic-sec}
We consider a binary outcome $Y$ with a value  in ${\mathcal Y}=\{0,1\}$ and a covariate $X$ in a subset ${\mathcal X}$ of $\mathbb R^d$.
The probability distribution is characterized by a probability mass function (pmf) or the RN-derivative with respect to a counting measure $C$:
\begin{align}\nonumber
p(y,\pi)=\pi^y(1-\pi)^{1-y},
\end{align}
which is referred to as the Bernoulli distribution ${\tt Ber}(\pi)$, where $\pi$ is the probability of $Y=1$.
A binary regression model is defined by a link function of the systematic component $\omega$  into the
random component:
$ 
g(\eta)={\exp(\omega)}/{\{1+\exp(\omega)\}}, 
$ 
so that the conditional pmf given $X=x$ with a linear model $\omega =\theta^\top x$ is given by
\begin{align}\label{logistic}
p(y|x,\theta)=\frac{\exp(y\theta^\top x)}{1+\exp(\theta^\top x)},
\end{align}
which is referred as a logistic model \cite{cox1958some,hosmer2013applied}.
 
The KL-divergence between Bernoulli distributions is given by
\begin{align}\nonumber
D_{0}({\tt Ber}(\pi),{\tt Ber}(\rho))=\pi\log\frac{\pi}{\rho}+(1-\pi)\log\frac{1-\pi}{1-\rho}.
\end{align}
For a given dataset $\{(X_i,Y_i)\}_{i=1,...,n}$, the negative log-likelihood function is given by
\begin{align}\nonumber
L_0(\theta)=-\frac{1}{n}\sum_{i=1}^n Y_i\theta^\top X_i-\log\{1+\exp(\theta^\top X_i)\}
\end{align}
and the likelihood equation is written by
 \begin{align}\label{g-est-logistic}
{S}_0(\theta )=-\frac{1}{n}\sum_{i=1}^n \Big\{Y_i-\frac{\exp(\theta^\top X_i)}
{1+\exp(\theta^\top X_i)}\Big\}X_i=0.
\end{align}
On the other hand, the $\gamma$-divergence  is given by
\begin{align}\nonumber
D_\gamma({\tt Ber}(\pi),{\tt Ber}(\rho);C)=-\frac{1}{\gamma}\frac
{\pi\rho^\gamma+(1-\pi)(1-\rho)^\gamma}
{\big\{\rho^{\gamma+1}+(1-\rho)^{\gamma+1}\big\}^{\frac{\gamma}{\gamma+1}}}
+\frac{1}{\gamma}
{\big\{\pi^{\gamma+1}+(1-\pi)^{\gamma+1}\big\}^{\frac{1}{\gamma+1}}},
\end{align}
where $C$ is the counting measure on $\cal Y$.
Note that this depends on the choice of $C$ as the reference measure on $\cal Y$.
The $\gamma$-expression of the logistic model \eqref{logistic} is given by
\begin{align}\nonumber
p^{(\gamma)}(y|x,\omega)=\frac{\exp\{(\gamma+1)y\theta^\top x\}}{1+\exp\{(\gamma+1)\theta^\top x\}}.
\end{align}
Hence, the $\gamma$-loss function is written by
\begin{align}\label{GM1}
L_\gamma(\theta;C)=-\frac{1}{n}\frac{1}{\gamma}
\sum_{i=1}^n \Big[\frac{\exp\{(\gamma+1) Y_i\theta^\top X_i\}}{1+\exp\{(\gamma+1)\theta^\top X_i\}}\Big]^{\frac{\gamma}{\gamma+1}}.
\end{align}
and the $\gamma$-estimating function is written as
 \begin{align}\nonumber 
{S}_\gamma(\theta;C)=\frac{1}{n}\sum_{i=1}^n {S}_\gamma(X_i,Y_i,\theta;C),
\end{align}
where
 \begin{align}
{S}_\gamma(X,Y,\theta;C)= \Big[
\frac{\exp\{(\gamma+1) Y_i\theta^\top X\}}{1+\exp\{(\gamma+1)\theta^\top X\}}\Big]^{\frac{\gamma}{\gamma+1}}
\Big\{Y-\frac{\exp\{(\gamma+1)\theta^\top X\}}
{1+\exp\{(\gamma+1)\theta^\top X\}}\Big\}X \label{g-est1},
\end{align}
see \cite{hung2018robust} for the discussion for robust mislabel.
See \cite{Eguchi2002,Murata2004,Takenouchi2004,Takenouchi2008,Takenouchi2012,komori2016asymmetric} for other type of MDE approaches than the $\gamma$estimation.

The $\gamma$-divergence on the space of Bernoulli distributions is well defined for all real number $\gamma$. 
Let us fix as $\gamma=-1$, and thus the GM-divergence between Bernoulli distributions is given by
\begin{align}\nonumber
D_{\rm GM }({\tt Ber}(\pi),{\tt Ber}(\rho);R)=
\Big\{\frac{\pi}{\rho}r+  \frac{1-\pi}{1-\rho}(1-r)\Big\}
{\pi}^{r}({1-\pi})^{1-r}-
{\rho}^{r}({1-\rho})^{1-r},
\end{align}
where the reference measure $R$ is chosen by ${\tt Ber}(r)$. 
Hence, the GM-loss function is given by
\begin{align}\nonumber
L_{\rm GM }(\theta;R)=\frac{1}{n}
\sum_{i=1}^n r^{Y_i}(1-r)^{1-Y_i}\exp\{(r-Y_i)\theta^\top X_i\}.
\end{align}
The GM-loss function with the reference measure ${\tt Ber}(\half)$ is equal to the exponential loss function for AdaBoost algorithm discussed for an ensemble learning
\cite{freund1997decision}. 
The integrated discrimination improvement index via odds \cite{hayashi2024new} is
based on the GM-loss function to assess prediction performance.
We will give a further discussion in a subsequent chapter.
The GM-estimating function is written by
\begin{align}\nonumber
{S}_{\rm GM }(\theta;{\tt Ber}(r))=\frac{1}{n}
\sum_{i=1}^n (2Y_i-1)\exp\{(r-Y_i)\theta^\top X_i\}X_i
\end{align}
due to $r^Y (1-r)^{1-Y}(r-Y)=(2Y-1)r(1-r)$ for $Y=0,1$.
Therefore, this estimating function is unbiased for any $r, 0<r<1$, that is, the expected estimating function conditional on $(X_1,...,X_n)$
under the logistic model \eqref{logistic} is equal to a zero vector.

We discuss which $r$ is effective for practical problems in logistic regression applications.
In particular, we focus on a problem of imbalanced samples that is an important issue in the binary regression.
An imbalanced dataset is one where the distribution of samples across these two classes is not equal.
For example, in a medical diagnosis dataset, the number of patients with a rare disease (class 1) may be significantly lower than those without it (class 0).
In this way, it is characterized as \begin{align}\nonumber0\approx P(Y=1)\ll P(Y=0)\approx1,\end{align}
There are difficult issues for the model bias, the poor generalization and the inaccurate performance metrics for the prediction.
Imbalanced samples can lead to biased or inconsistent estimators, affecting hypothesis tests and confidence intervals.
For these problem resampling techniques have been exploited by oversampling the minority class or undersampling the majority class can balance the dataset.
Also, the cost-sensitive Learning introduces a cost matrix to penalize misclassification of the minority class more heavily.
The asymmetric logistic regression is proposed introducing a new parameter to account for data complexity
\cite{komori2016asymmetric}. 
They observe that this parameter controls the influence from imbalanced sampling.
Here we tackle with this problem by the GM-estimator choosing an appropriate reference distribution $R$ in the GM-loss function.
We select ${\tt Ber}(\hat\pi_0)$ as the reference measure, where $\hat\pi_0$ is the proportion of the negative sample, namely
$\hat\pi_0=\sum_{i=1}^n {\mathbb I}(Y_i=0)/n$.
Then,  the resultant loss function is given by 
\begin{align}\label{iw}
L^{\rm (iw)}_{\rm GM }(\theta)=\frac{1}{n}\sum_{i=1}^n \hat\pi_0^{Y_i}(1-\hat\pi_0)^{1-Y_i}
\exp\{(\hat\pi_0-Y_i)\theta^\top X_i\}.
\end{align}
We refer this to as the inverse-weighted  GM-loss function since the weight 
$$
\hat\pi_0^{Y_i}(1-\hat\pi_0)^{1-Y_i}\propto \frac{1}{(1-\hat\pi_0)^{Y_i}\hat\pi_0^{1-Y_i}}.
$$
Hence, the estimating function is given by
\begin{align}\nonumber
{S}^{\rm (iw)}_{\rm GM }(\theta)=\frac{1}{n}\sum_{i=1}^n (2Y_i-1) \exp\{(\hat\pi_0-Y_i)\theta^\top X_i\} X_i.
\end{align}
Equating the estimating function to zero gives  the equality between two sums of positive and negative samples:
\begin{align}\nonumber
\frac{1}{n}\sum_{i=1}^n {\mathbb I}(Y_i=1)\exp\{(\hat\pi_0-1)\theta^\top X_i\}X_i=\frac{1}{n}\sum_{i=1}^n {\mathbb I}(Y_i=0)\exp\{\hat\pi_0 \theta^\top X_i\}X_i.
\end{align}
Alternatively, the likelihood estimating equation is written as
\begin{align}\nonumber
\frac{1}{n}\sum_{i=1}^n {\mathbb I}(Y_i=1)\frac{1}{1+\exp(\theta^\top X_i)}X_i=\frac{1}{n}\sum_{i=1}^n {\mathbb I}(Y_i=0)\frac{\exp(\theta^\top X_i)}{1+\exp(\theta^\top X_i)}X_i.
\end{align}
Both of estimating equations are unbiased, however the weightings are contrast each other.

We conduct a brief study for numerical experiments.
Assume that covariate vectors are generated from a mixture of bivariate normal distributions as  
\begin{align}\nonumber
   X_i \sim (1-\epsilon)\ {\tt Nor}(\mu_0,{\rm I})+ \epsilon\ {\tt Nor}(-\mu_0,{\rm I}).
\end{align}
where $\rm I$ denotes a 2-dimensional identity matrix.
Here, we set as $n=1000,\mu_0=(2.,2.)^\top$, and the mixture ratio $\epsilon$ will be taken by some fixed values.
The outcome variables are generated from Bernoulli distributions as
$   Y_i \sim {\tt Ber}(\pi(X_i))$,
where
\begin{align}\nonumber
    \pi(X_i,\theta_0)=\frac{\exp(\theta_0+\theta_1 X_{i1}+\theta_2 X_{i2})}{1+\exp(\theta_0+\theta_1 X_{i1}+\theta_2 X_{i2})}
\end{align}
where we set as $n=1000$ and $(\theta_0,\theta_1^\top)=(0.5,1.0,1.5)^\top.$
This simulation is designed to have imbalanced samples such that  the positive sample proportion  approximately 
becomes near $\epsilon$. 
  
We compared the ML-estimator $\hat\theta$ with the inverse-weighted GM estimator $\hat\theta_{\rm GM}$ with 30 replications.
Thus, we observe that the GM estimator have a better performance over the ML-estimator in the sense of true positive rate.
Table \ref{TPRTNR} is the list of the true positive and negative rates based on test samples with size $1000$. 
Note that two label conditional distributions are ${\tt Nor}(\mu_0,{\rm I})$ are ${\tt Nor}(-\mu_0,{\rm I})$.
These are set to be sufficiently separated from each other.  
Hence, the classification problem becomes extremely an easy task when $\epsilon$ is a moderate value.
Both ML-estimator and GM estimator have good performance in cases of $\epsilon=0.3, 0.1$.
Alternatively, we observe that the true positive rate for GM estimator is considerably higher than that of ML-estimator in a situation of imbalance samples as in either case of $\epsilon=0.03, 0.01$.
\begin{table}[hbtp]
\caption{The comparison between MLE vs  GME }

  \centering
\vspace{2mm}

\vspace{4mm}
  \begin{tabular}{ccc}
  \hline 

    $\epsilon$  & MLE& GME
 \\
    \hline \hline

    $0.3$ & $ (0.969,0.995)$& $( 0.956,0.994 )$\\
    $0.1$   & $ (0.897, 0.998)$& $( 0.902, 0.995)$ \\
    $0.05$   & $( 0.800, 0.999)$& $(0.817,0.994)$ \\
    $0.03$   & $ (0.705, 0.999)$& $(0.733,0.995)$ \\
   $0.01$   & $ (0.462, 0.999)$& $(0.538,0.996)$ \\
   \hline
\label{TPRTNR}
  \end{tabular}

\vspace{4mm}
$(a , b)$ denotes a pair of the true positive and  negative rates $a$ and $b$.
  \label{table1}

\end{table}

\

We next focus on  the HM-divergence ($\gamma$-divergence, $\gamma=-2$):
\begin{align}\nonumber
D_{\rm HM }({\tt Ber}(\pi),{\tt Ber}(\rho))=\pi(1-\rho)^2+(1-\pi)\rho^2-\pi(1-\pi),
\end{align}
where the reference measure is determined by ${\tt Ber }(\rho)$.
The HM-loss function is derived as
 \begin{align}\nonumber
L_{\rm HM }(\theta)= \frac{1}{n}\sum_{i=1}^n\Bigg[ \frac{\exp\{(1-y_i) \theta^\top X_i\}}{1+\exp(\theta^\top X_i)}\Bigg]^2,
\end{align}
for the logistic model \eqref{logistic}.
Note that the HM-loss function is the $\gamma$-loss function with $\gamma=-2$, which the $\gamma$-expression is
reduced to
\begin{align}\nonumber
p^{(-2)}(y|x,\omega)=\frac{\exp\{(1-y)\theta^\top x\}}{1+\exp(\theta^\top x)}.
\end{align}
Hence, the HM-estimating function is written as
\begin{align}\nonumber
{S}_{\rm HM }(\theta)=\frac{1}{n}\sum_{i=1}^n 
\frac{\exp(\theta^\top X_i)}{\{1+\exp(\theta^\top X_i)\}^2}
\Big\{Y_i-\frac{\exp(\theta^\top X_i)}{1+\exp(\theta^\top X_i)}\Big\}X_i=0
\end{align}
This is a weighted likelihood score function with the conditional variance of $Y$ as the weight function. 
We will observe that this weighting has an effective key for the HM estimator to be   robust
for covariate outliers.

\

Let us investigate the behavior of the estimating function ${S}_\gamma(\theta;C)$  of the $\gamma$-estimator.
In general,  ${S}_\gamma(\theta;C)$ is unbiased,
that is, $\mathbb E_0[{S}_\gamma(\theta;C)|\underline X]=0$ under the conditional expectation with the true distribution with the pmf $p(y|x,\theta_0)$.
However, this property easily violated if the expectation is taken by a misspecified distribution $Q$ with the pmf $q(y|x)$ other than the true distribution
\cite{Copas1988,Copas2005,Komori2019}.
Hence, we look into the expected estimating function under the misspecified model.

\begin{proposition}\label{prop5}
Consider the $\gamma$-estimating function under a logistic model \eqref{logistic}.
Assume  $\gamma>0$ or $\gamma<-1$.
Then, 
\begin{equation}\label{sup}
\sup_{x\in{\mathcal X}}|\theta^\top \mathbb E_{Q}[{S}_{\gamma }(X,Y,\theta)|X=x]| < \infty, 
\end{equation}
where $\mathbb E_{Q}[\ \cdot\ |X=x]$ is the conditional expectation under a misspecified distribution $Q$
outside the model \eqref{logistic}. 

\end{proposition}

\begin{proof}
It is written from \eqref{g-est1} that
\begin{align}\nonumber
&\mathbb E_{Q} [{S}_{\gamma}(X,Y,\theta)|X=x]\\[3mm]
&=\sum_{y=0}^1 \bigg[
\frac{\exp\{(\gamma+1)y\theta^\top x\}}{1+\exp\{(\gamma+1)\theta^\top x\}}\bigg]^{\frac{\gamma}{\gamma+1}}
\Big[y -\frac{\exp\{(\gamma+1)  \theta^\top x\}}{1+\exp\{(\gamma+1)\theta^\top x\}}\Big]q(y|x)(\gamma+1)x
\end{align}
Hence, if $s=(\gamma+1)\theta^\top x$, then 
\begin{align}\nonumber
\big|\theta^\top \mathbb E_{Q} [{S}_{\gamma}(X,Y,\theta)|X=x]\big|
\leq \Psi_\gamma (s) +\Psi_\gamma (-s),
\end{align}
where
\begin{align}\label{Psi}
\Psi_\gamma (s)=\frac{|s|}{|\gamma+1|}\Big\{
\frac{1}{1+\exp(-s)}\Big\}^{\frac{\gamma}{\gamma+1}}
\frac{\exp(-s)}{1+\exp(-s)}.
\end{align}
We observe that, if   $\gamma>0$ or $\gamma<-1$, then
\begin{align}\nonumber
\sup_{s\in\mathbb R} \Psi_\gamma (s)=\sup_{s\in\mathbb R} \Psi_\gamma (-s)<\infty.
\end{align}
This concludes \eqref{sup}.

\end{proof}
 
.

We note that Proposition \ref{prop5} focuses only on the logistic model \eqref{logistic}, however such a boundedness property
holds in both the probit model and the complementary log-log model.

We consider a geometric understanding for the bounded property in \eqref{sup}.
In GLM, the linear predictor is written by $\theta^\top x=\theta_1^\top x_1 +\theta_0$, where $\theta_1$ and
$\theta_0$ are referred to as a slope vector and intercept term, respectively.
The decision boundary is defined as ${H}_\theta$ as in \eqref{boundary}.
The Euclidean distance of $x$ into ${H}_\theta$, 
\begin{align}\nonumber
  d(x,{\mathcal H}_\theta)=  \frac{|\theta_1^\top x_1 -\theta_0|}{\|\theta_1\|}.
\end{align}
is referred to as the margin of $x$ from the decision boundary ${H}_\theta$, which plays.a central role
on the support vector machine \cite{cortes1995support}.
Let   
\begin{align}\label{tube}
{\mathcal N}_\theta(\delta)=\big\{ x\in{\mathcal X}: d(x,{\mathcal H}_\theta)\leq\delta  \big\}.
\end{align}
This is the $\delta$-tubular neighborhood including ${\mathcal H}_\theta$.
In this perspective, Proposition \ref{prop5} states for any $\gamma, \gamma<-1 \text{ or } \gamma>0$ 
that the conditional expectation of $\gamma$-estimating function  is in the tubular neighborhood with probability one even under  the misspecified distribution outside  the parametric model \eqref{logistic}.
On the other hands, the likelihood estimating function does not satisfy such a stable property because the margin of the conditional expectation becomes unbounded.
Therefore, we result that the $\gamma$-estimator is robust for misspecification for the model for $\gamma>0$ or $\gamma>-1$; while  the ML-estimator is not robust.

We observe in the Euclidean geometric view that, for a feature vector $x$ of $\cal X$, the decision hyperplane ${\mathcal H}_\theta$ decompose $x$ into  orthogonal and tangential components as $x= z+w$, where $z=(\theta^{\top}x)\theta/\|\theta\|^2$ and
$w=x-z$. 
Note $z\perp w$ and $\|x\|^2=\|z\|^2+\|w\|^2$.
In accordance with this geometric view, we give more insights on the robust performance for the $\gamma$-estimator class.
We write the $\gamma$-estimating function \eqref{g-est1} by
$S_\gamma(x,y,\theta)=\eta_\gamma(y,\theta^\top x)(z+w)$.
Then,
\begin{align}
|S_\gamma(x,y,\theta)|\leq|\eta_\gamma(y,z^\top \theta)( \|z\|+\|w\|).
\end{align}
Therefore, we conclude that
\begin{align}
|S_\gamma(x,y,\theta)|\leq
\sup_{s\in\mathbb R} |\eta_\gamma(y,s)|\Big( \frac{|s|}{\|\theta\|^2}+\|w\|\Big).
\end{align}

Thus, we observe a robust property of the $\gamma$-estimator in a more direct perspective.

\begin{proposition}\label{robust-g}
Assume  $\gamma>0$ or $\gamma<-1$.
Then, the $\gamma$-estimating function ${S}_\gamma(\theta;C)$ based on a dataset ${\cal D}=\{(X_i.Y_i)\}_{i=1}^n$ satisfies
\begin{align}\label{concl1}
\sup_{\cal D}| \theta^\top{S}_\gamma(\theta;C)| <\infty.
\end{align}

\end{proposition}

\begin{proof}
It is written from \eqref{g-est1} that
\begin{align}\nonumber
 &\theta^\top{S}_{\gamma}(\theta:C)\\[.1mm]
 &\hspace{-10mm}=\sum_{i=1}^n \bigg[
\frac{\exp\{(\gamma+1)Y_i\theta^\top X_i\}}{1+\exp\{(\gamma+1)\theta^\top X_i\}}\bigg]^{\frac{\gamma}{\gamma+1}}
\Big[Y_i -\frac{\exp\{(\gamma+1)  \theta^\top X_i\}}{1+\exp\{(\gamma+1)\theta^\top X_i\}}\Big]\theta^\top X_i
\end{align}
which is decomposed into the sum of the positive and negative samples as
\begin{align}\nonumber
\frac{1}{n}\sum_{i=1}^n \big\{  {\mathbb I}(Y_i=1)\Psi_\gamma(S_i)-{\mathbb I}(Y_i=0)\Psi_\gamma(-S_i)\big\}
\end{align}
where $S_i=(\gamma+1)\theta^\top X_i$ and
$\Psi_\gamma(s) $ is defined in \eqref{Psi}. 
Hence, we get 
\begin{align}\nonumber
|\theta^\top{S}_{\gamma}(\theta:\Lambda)|\leq \frac{1}{n}\sum_{i=1}^n \big\{  {\mathbb I}(Y_i=1)|\Psi_\gamma(S_i)|+{\mathbb I}(Y_i=0)|\Psi_\gamma(-S_i)|\big\}
\end{align}
which is bounded by
\begin{align}\nonumber
\frac{n_1}{n} \sup_{s\in\mathbb R}|\Psi_\gamma(s)| +\frac{n_0}{n}  \sup_{s\in\mathbb R}|\Psi_\gamma(-s)|
\end{align}
which is equal to $\delta_\gamma$, where $n_y=(1/n)\sum_{i=1}^n {\mathbb I}(Y_i=y)$ for $y=0,1$.  This concludes \eqref{concl1}. 
\end{proof}
 
The log-likelihood estimating function is given by
\begin{align}\label{like1}
 {S}_{0}(\theta;\Lambda)=-\frac{1}{n}\sum_{Y_i=0}^n \Big\{
{\mathbb I}(Y_i=1)\frac{\exp(\theta^\top X_i)}{1+\exp(\theta^\top X_i)}+
{\mathbb I}(Y_i=0) 
\frac{1}{1+\exp(\theta^\top X_i)}\Big\}X_i.
\end{align}
Hence,
$|\theta^\top {S}_{0}(\theta|\Lambda)|$ is unbounded in $\{\theta^\top X_i:i=1,...,n\}$
since either of two terms in \eqref{like1} diverges to infinity as $|\theta^\top X_i|$ goes to infinity.
The GM-estimating function is written by
\begin{align}\nonumber
{S}_{\rm GM }(\theta,R)=\frac{1}{n}
\sum_{i=1}^n\Big[{\mathbb I}(Y_i=1)\exp\{r(0)\theta^\top X_i\}r(0)+{\mathbb I}(Y_i=0)\exp\{r(1)\theta^\top X_i\}r(1)\Big]X_i.
\end{align}
This implies that $|\theta^\top{S}_{\rm GM}(\theta;R)|$ is unbounded.

\

We have a brief study for numerical experiments in two types of sampling.
One is based on the covariate distribution conditional on the outcome $Y$, which is widely analyzed in case-control studies. 
The other is based on the conditional distribution of $Y$ given the covariate vector $X$, which is common in cohort-control studies.
First, we consider a model of an outcome-conditional distribution.
Assume that the conditional distribution of $X$ given $Y=y$ is a bivariate normal distribution ${\tt Nor}(\mu_y,{\rm I})$, where $\rm I$ is a 2-dimensional identity matrix.
Then, the marginal distribution of $X$ is written as $p_1{\tt Nor}(\mu_1,{\rm I})+p_0{\tt Nor}(\mu_0,{\rm I})$, where $p_y=P(Y=y)$.
The the conditional pmf of $Y$ given $X=x$ is given by
\begin{align}\nonumber
p(y|x,\theta)=\frac{\exp\{y(\theta_1^\top x+\theta_0)\}}{1+\exp(\theta_1^\top x+\theta_0)}
\end{align}
due to the Bayes formula, where $\theta_1=\mu_1-\mu_0$ and $\theta_0=\half(\mu_0^\top\mu_0-\mu_1^\top\mu_1)-\log(p_1/p_0)$.
Let  $N\sim{\tt Bin}(p_1,n).$
The simulation was conducted  based on a scenario about the $N$ positive and $n-N$ negative samples with $Y_i=1$ and $Y_i=0$, respectively, as follows.

\

\noindent
(a). Specified model:  
{$\hspace{9mm} \{X_i\}_{i=1}^N \sim {\tt Nor}(\mu_1,{\rm I}).$} and {$ \{X_i\}_{i=N+1}^n \sim {\tt Nor}(\mu_0,{\rm I}).$}

\vspace{5mm}
\noindent
(b). Misspecified model:
{\hspace{4mm}$\{X_i\}_{i=1}^N \sim (1-\pi){\tt Nor}(\mu_1,{\rm I})+\pi{\tt Nor}(\sigma\mu_0,{\rm I}).$ and  \\
{\hspace{50mm}{$ \{X_i\}_{i=N+1}^n \sim {\tt Nor}(\mu_0,{\rm I}).$}

\vspace{3mm}
\noindent
Here parameters were set as $\mu_1=(0.5,0.5)^\top$, $\mu_0=-(0.5,0.5)^\top$, $p_1=0.5$   and  $(\pi, \sigma)=(0.1,-4.0)$, so that $(\theta_0,\theta_1^\top)=(0.0,1.0,1.0)$.
Figure \ref{3Dplot1} shows the plot of 103 negative samples (Blue), 87 negative samples (Green), 10 negative outliers (Red) on the logistic model surface $\{(x_1,x_2,p(1|(x_1,x_2),(\theta_0,\theta_1)): -3.5\leq x_1\leq3.5,-3.5\leq x_2\leq3.5\}$. Thus, 10 negative outliers are away from the hull of 87 negative samples.
 
\

\begin{figure}[htbp]
\begin{center}
  \includegraphics[width=100mm]{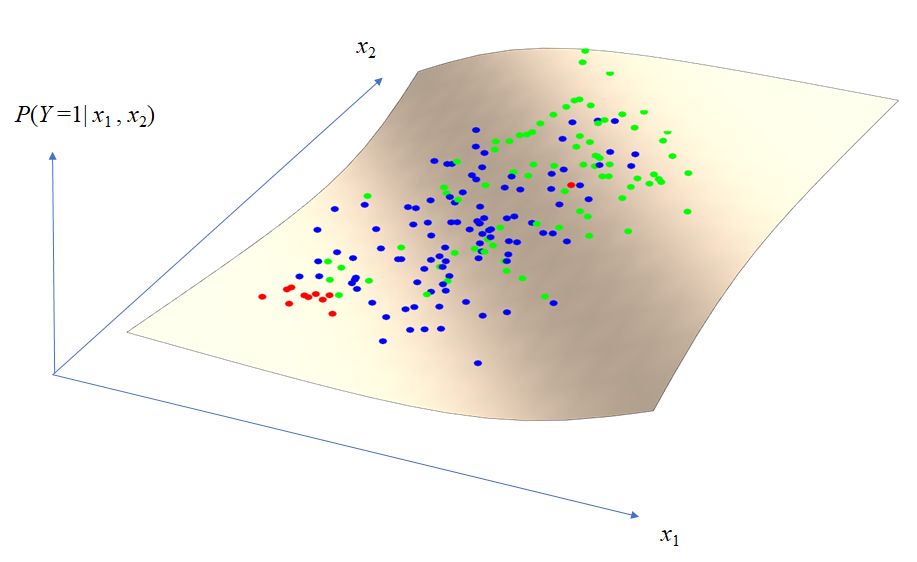}
 \end{center}
 \vspace{-3mm}\caption{Covariate vectors on the logistic model}

\label{3Dplot1}
\end{figure} 


We compared the estimates the ML-estimator $\hat\theta_0$, the $\gamma$-estimator $\hat\theta_{\gamma}$ with $\gamma=0.8$, the GM-estimator
$\hat\theta_{\rm GM}$ and $\rm HM$-estimator $\hat\theta_{\rm HM}$ , where the simulation was conducted by 300 replications.
See Table \ref{gestLogistic} for the performance of four estimators in case (a) and (b) and Figure \ref{Boxplot2} for the box-whisker plot in case (b). 
In the case (a) of specified model, the ML-estimator was superior to other estimators in a point of the
root means square estimate (rmse), however the superiority is subtle.
Next, we observe for case (b) of misspecified model in which the conditional distribution given $Y=1$ is contaminated with a normal distribution ${\tt Nor}(\sigma\mu_0,{\rm I})$ with mixing ratio $0.1$.
Under this setting, $\gamma$-estimator $(\gamma=0.8)$ and the HM -estimator were substantially robust;
the ML-estimator and GM-estimator were sensitive to the misspecification.
Upon closer observation, it becomes apparent that $\gamma$-estimator $(\gamma=0.8)$ and the HM -estimator were superior to  the ML-estimator and GM-estimator in the bias behaviors rather than the variance ones
as shown in Figure \ref{Boxplot2}. 
This observation is consistent with what Proposition \eqref{robust-g} asserts:
The $\gamma$ estimator has a boundedness property if $\gamma<-1$ or $\gamma>0$.
Because the ML-estimator, the GM-estimator and the HM-estimator equal the $\gamma$-estimators with
$\gamma=0,-1,-2$, respectively.

\begin{table}[hbtp]

\caption{Comparison between the ML-estimator and the $\gamma$-estimator. }

  \centering
\vspace{3.2mm}

(a). The case of specified model

\vspace{3mm}
  \begin{tabular}{ccc}
  \hline 

    Method  & estimate  &  rmse
 \\
    \hline \hline

\vspace*{.1mm}   
 ML-estimator & $({0.011, 1.021, 1.014})$  & $0.341 $\\
    $\gamma$-estimate   & $({0.012, 1.062, 1.045})$  & $0.407$ \\
    GM-estimate   & $({0.009, 1.031, 1.029})$  & $0.365$ \\
    HM-estimate   & $({0.013, 1.051, 1.037})$  & $0.390$ \\

   \hline
  \end{tabular}
\vspace{5.2mm}

(b). The case of  misspecified model

\vspace{3mm}
  \begin{tabular}{ccc}
  \hline 

    Method  & estimate  &  rmse \vspace{1mm}
 \\
    \hline \hline

    \vspace*{.1mm}   
 ML-estimator & $({0.102, 0.481, 0.503})$  & $0.758 $\\
    $\gamma$-estimate   & $({0.081, 0.889, 0.911})$  & $0.441$ \\
    GM-estimate   & $({0.161, 0.428, 0.452})$  & $0.839$ \\
    HM-estimate   & $({-0.070, 0.862, 0.885})$  & $0.464$ \\

   \hline
  \end{tabular}

\label{gestLogistic}
\end{table}

\

\begin{figure}[htbp]
\begin{center}
  \includegraphics[width=100mm]{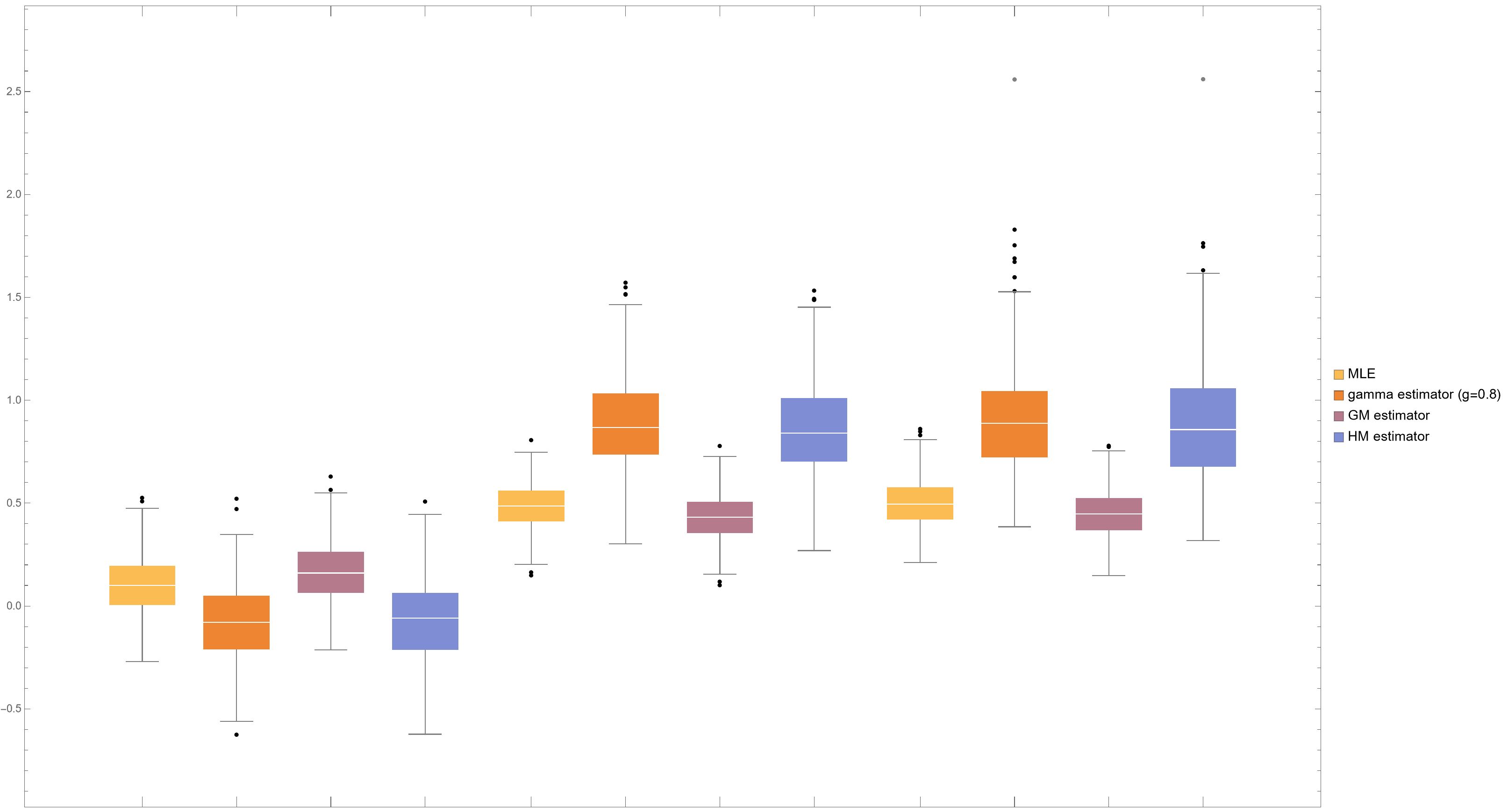}
\end{center}
 \vspace{-2mm}\caption{Box-whisker Plots of  the ML-estimator and the $\gamma$-estimator $(\gamma=0.8)$, GM-estimator, HM-estimator}
\label{Boxplot2}
\end{figure}

Second, we consider a model of a covariate-conditional distribution of $Y$.
Assume that $X$ follows a standard normal distribution ${\tt Nor}(0, I_2)$ and
a conditional distribution of $Y$ given $X=x$ follows a logistic model
\begin{align}\nonumber
p(y|x,\theta)=\frac{\exp\{y(\theta_1^\top x+\theta_0)\}}{1+\exp(\theta_1^\top x+\theta_0)}.
\end{align}
The simulation was conducted  based on a scenario  as follows.

\

\noindent
(a). Specified model:  
{$\hspace{9mm} (Y_i=y)|(X_i=x) \sim {\tt Ber}(p(y|x,\theta)).$}

\vspace{3mm}
\noindent
(b). Misspecified model:
{\hspace{4mm}$(Y_i=y)|(X_i=x) \sim  (1-\epsilon){\tt Ber}(p(y|x,\theta))+ \epsilon {\tt Ber}(p(y|x,\theta_{\rm out})).$}

\

\noindent
Here parameters were set as $(\theta_0,\theta_1^\top)=(0.0,1.0,1.0)$.
Figure \ref{3Dplot1} shows the plot of 103 negative samples (Blue), 87 negative samples (Green), 10 negative outliers (Red) on the logistic model surface $\{(x_1,x_2,p(1|(x_1,x_2),(\theta_0,\theta_1)): -3.5\leq x_1\leq3.5,-3.5\leq x_2\leq3.5\}$. Thus, 10 negative outliers are away from the hull of 87 negative samples.

Similarly, a comparison among  the ML-estimator $\hat\theta_0$, the $\gamma$-estimator $\hat\theta_{\gamma}$ with $\gamma=0.8$, the GM-estimator
$\hat\theta_{\rm GM}$ and $\rm HM$-estimator $\hat\theta_{\rm HM}$ with $100$ replications.
See Table \ref{gestLogistic1}. 
In the case (a),  the ML-estimator was slightly superior to other estimators.
For case (b), $\gamma$-estimator $(\gamma=0.8)$ and the HM -estimator were more robust;
the ML-estimator and GM-estimator, which was the same tendency as the case of the outcome-conditional model.  


\begin{table}[hbtp]

\caption{Comparison between the ML-estimator and the $\gamma$-estimator. }

  \centering
\vspace{3.2mm}

(a). The case of specified model

\vspace{3mm}
  \begin{tabular}{ccc}
  \hline 

    Method  & estimate  &  rmse
 \\
    \hline \hline

\vspace*{.1mm}   
 ML-estimator & $({0.011, 1.021, 1.014})$  & $0.341 $\\
    $\gamma$-estimate   & $({0.012, 1.062, 1.045})$  & $0.407$ \\
    GM-estimate   & $({0.009, 1.031, 1.029})$  & $0.365$ \\
    HM-estimate   & $({0.013, 1.051, 1.037})$  & $0.390$ \\

   \hline
  \end{tabular}
\vspace{5.2mm}

(b). The case of  misspecified model

\vspace{3mm}
  \begin{tabular}{ccc}
  \hline 

    Method  & estimate  &  rmse \vspace{1mm}
 \\
    \hline \hline

    \vspace*{.1mm}   
 ML-estimator & $({0.102, 0.481, 0.503})$  & $0.758 $\\
    $\gamma$-estimate   & $({0.081, 0.889, 0.911})$  & $0.441$ \\
    GM-estimate   & $({0.161, 0.428, 0.452})$  & $0.839$ \\
    HM-estimate   & $({-0.070, 0.862, 0.885})$  & $0.464$ \\

   \hline
  \end{tabular}

\label{gestLogistic1}
\end{table}

\section{Multiclass logistic regression}\label{subsec-Multiclass}

We  consider a situation where an outcome variable $Y$ has a value  in ${\mathcal Y}=\{0,...,k\}$ and a covariate $X$ with a value  in a subset ${\mathcal X}$ of $\mathbb R^d$.
The probability distribution is given by a probability mass function (pmf)
\begin{align}\nonumber
p(y,\pi)=\prod_{j=0}^k \pi_j{}^{{\mathbb I}(y=j)} ,
\end{align}
which is referred to as the categorical distribution ${\tt Cat}(\pi)$, where $\pi=(\pi_j)_{j=1}^k$ is the probability vector, $(P(Y=j))_{j=1}^k$ with  $\pi_0$ being $1-\sum_{j=1}^k \pi_j$.
\begin{remark}
We begin with a simple case of estimating $\pi$ without any covariates.
Let $\{Y_i\}_{1\leq i\leq n}$ be a random sample drawn from  ${\tt Cat}(\pi)$.
Then, the estimators discussed here equal the observed frequency vector as follows.
First of all, the ML-estimator is the observed frequency vector with components $(\hat \pi_0,...,\hat\pi_k)$, where $\hat \pi_j=1/n\sum_{i=1}^n {\mathbb I}(Y_i=j)$.
Next, the $\gamma$-loss function
\begin{align}\nonumber
L_\gamma( \pi,C)=-\frac{1}{n}\frac{1}{\gamma}
\sum_{i=1}^n\frac{\prod_{j=0}^k \pi_j{}^{\gamma{\mathbb I}(Y_i=j)}}{\big(\sum_{j=0}^k\pi_j{}^{\gamma+1}\big)^\frac{\gamma}{\gamma+1}}
\end{align}
is written as
$ %
1/\gamma \sum_{j=0}^k { \hat \pi_j\pi_j{}^{\gamma}}/{\big(\sum_{j=0}^k\pi_j{}^{\gamma+1}\big)^\frac{\gamma}{\gamma+1}}
$.  %
We observe
\begin{align}\nonumber
L_\gamma( \pi,C)=D_\gamma({\tt Cat}(\hat\pi),{\tt Cat}(\pi))
\end{align}
up to a constant.
Therefore,  the $\gamma$-estimator for $\pi$ is equal to $\hat\pi$ for all $\gamma$.
Similarly, the $\beta$-estimator is equal to $\hat\pi$ for all $\beta$.
However the $\alpha$-estimator does not satisfy that except for the limit case of $\alpha$ to $0$, or the ML-estimator.
\end{remark}

We  return the discussion for the  regression model with a covariate vector $X$.
A multiclass logistic regression model is defined by a soft max function as a link function of the systematic component $\eta$  into the
random component. 
The conditional pmf given $X=x$ 
is given by
\begin{align}\label{76}
p(y|x,\theta)=\left\{\begin{array}{cl}
\displaystyle{\frac{1}{1+\sum_{j=1}^k \exp(\eta_j)} } & \text{ if } y=0,\\[5mm]
\displaystyle{\frac{\exp(\eta_y)}{1+\sum_{j=1}^k \exp(\eta_j)}} & \text{ if } y=1,...,k
\end{array}\right.,
\end{align}
which is referred as a multinomial logistic model, where $\theta=(\theta_1,...,\theta_k)^\top$ and $\eta_j=\theta_j{}^\top x$.
The KL-divergence between categorical distributions is given by
\begin{align}\nonumber
D_{0}({\tt Cat}(\pi),{\tt Cat}(\rho))=\sum_{j=0}^k\pi_j \log\frac{\pi_j}{\rho_j}.
\end{align}
For a given dataset $\{(X_i,Y_i)\}_{i=1,...,n}$, the negative log-likelihood function is given by
\begin{align}\nonumber
L_0(\theta;C)=-\frac{1}{n}\sum_{i=1}^n \bigg[\theta_{Y_i}{}^\top X_j-\log\Big\{ 1+\sum_{j=1}^k \exp(\theta_j{}^\top X_i)\Big\}\bigg]
\end{align}
where we set $\theta_y=0$ if $y=0$.
The likelihood equation is written by the $j$-th component:
 \begin{align}\nonumber
{S}_0{}_j(\theta;C)=-\frac{1}{n}\sum_{i=1}^n \Big\{{\mathbb I}(Y_i=j)-\frac{\exp(\theta_j{}^\top X_i)}{1+\sum_{l=1}^k \exp(\theta_l{}^\top X_i)}
\Big\}X_i=0.
\end{align}
for $j=1,...,k$.
The $\gamma$-divergence is given by
\begin{align}\nonumber
D_\gamma({\tt Cat}(\pi),{\tt Cat}(\rho))=-\frac{1}{\gamma}
\frac{\sum_{j=0}^k\pi_j\rho_j{}^\gamma}{\big(\sum_{j=0}^k\rho_j{}^{\gamma+1}\big)^\frac{\gamma}{\gamma+1}}+
\frac{1}{\gamma} \bigg(\sum_{j=0}^k\pi_j{}^{\gamma+1}\bigg)^\frac{1}{\gamma+1} ,
\end{align}
We remark that  the $\gamma$-expression defined in \eqref{g-model} is given by
\begin{align}\nonumber
p^{(\gamma)}(y|x,\theta)=p(y|x,(\gamma+1)\theta),
\end{align}
where $p(y|x,\theta)$ is  in the multi logistic model \eqref{76}. 
Hence, the $\gamma$-loss function is given by
\begin{align}\nonumber
L_\gamma(\theta;C)=-\frac{1}{n}\frac{1}{\gamma}
\sum_{i=1}^n \sum_{j=0}^k \{p (Y_i|x,(\gamma+1)\theta)\}^{\frac{\gamma}{\gamma+1}}
\end{align}
and the $\gamma$-estimating equation is written as
 \begin{align}\nonumber
{S}_\gamma{}_j(\theta;C):=\frac{1}{n}\sum_{i=1}^n {S}_\gamma(X_i,Y_i,\theta;C)=0,
\end{align}
where ${S}_\gamma(X,Y,\theta;C)$ is defined by
 \begin{align}\label{multi-g-est} 
 \{p(Y|X,(\gamma+1)\theta)\}^{\frac{\gamma}{\gamma+1}}
\big\{{\mathbb I}(Y=j)-p(j|X,(\gamma+1)\theta)\big\}X. 
\end{align}
The GM-divergence between categorical distributions:
\begin{align}\nonumber
D_{\rm GM }({\tt Cat}(\pi),{\tt Cat}(\rho);R)=
\sum_{y=0}^k\frac{\pi_y}{\rho_y}r(y)\prod_{y=0}^k \rho_y^{r(y)}-\prod_{y=0}^k \pi_y^{r(y)},
\end{align}
where the reference distribution $R$ is chosen by ${\tt Cat}(r)$. 
Hence, the GM-loss function is given by
\begin{align}\nonumber
L_{\rm GM }(\theta;R)=\frac{1}{n}\sum_{i=1}^n r(Y_i)\exp\{(\bar\theta_R-\theta_{Y_i})^\top X_i\}.,
\end{align}
where $\bar\theta_R=\sum_{j=1}^k r(j)\theta_j$.
We will have a further discussion later such that the GM-loss is closely related to the exponential loss for Multiclass AdaBoost algorithm.
The GM-estimating function is given by
\begin{align}\nonumber
{S}_{\rm GM}{}_j(\theta;R)=\frac{1}{n}
\sum_{i=1}^nr(Y_i)\exp\{(\bar\theta-\theta_{Y_i})^\top X_i\} \{r(Y_i)-{\mathbb I}(Y_i=j)\}X_j.
\end{align}
Finally, the HM-divergence is
\begin{align}\nonumber
D_{\rm HM }({\tt Cat}(\pi),{\tt Cat}(\rho))=\sum_{y=0}^k\pi_y(1-\rho_y)^2-\sum_{y=0}^k \pi_y(1-\pi_y)^2.
\end{align}
The HM-loss function is derived as
 \begin{align}\nonumber
\label{false}
L_{\rm HM }(\theta)= \frac{1}{2n}\sum_{i=1}^n \sum_{j=0}^k \{p(Y_i|X_i,-\theta)\}^{2}
\end{align}
for the logistic model \eqref{76} noting $p^{(-2)}(y|x,\theta)=p (y|x,-\theta)$.
This is the sum of squared probabilities of the inverse label.
Hence, the HM-estimating function is written as
\begin{align}\nonumber
{S}_{\rm HM }(\theta)=\frac{1}{n}\sum_{i=1}^n p (Y_i|X_i,-\theta)\frac{\partial}{\partial\theta}p(Y_i|X_i,-\theta).
\end{align}

\begin{figure}[htbp]
\begin{center}
  \includegraphics[width=120mm]{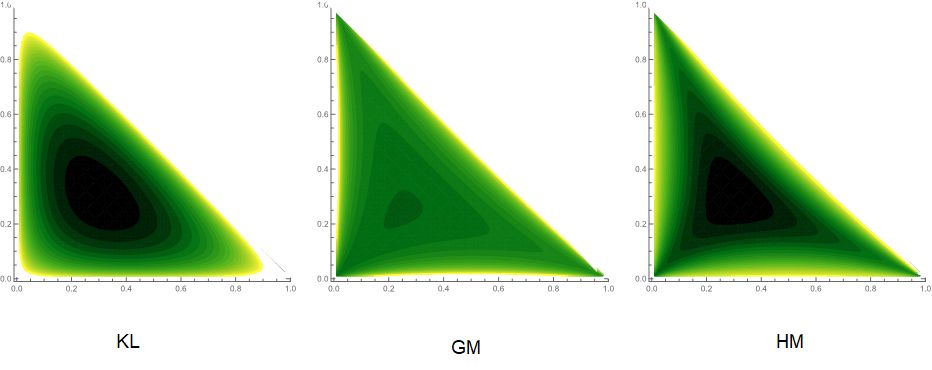}
 \end{center}
 \vspace{-5mm}\caption{Plots of contours of KL, GM and HM divergence measures.}
\end{figure}

Let us have a brief look at the behavior of the $\gamma$-estimating function ${S}_\gamma(\theta;C)$ in the presence of misspecification of the parametric model in the multiclass logistic distribution \eqref{76}.
Basically, most of properties are similar to those in the Bernoulli logistic model.

\begin{proposition}
Consider the $\gamma$-estimating function under a multiclass logistic model \eqref{76}.
Assume  $\gamma>0$ or $\gamma<-1$.
Then, 
\begin{equation}\label{sup-multi}
\sup_{x\in{\mathcal X}}|\theta_j^\top  {S}_{\gamma }{}_j(\theta;\Lambda)| \leq \delta_\gamma{}_j,
\end{equation}
where
\begin{align}\label{delta-multi}
\delta_\gamma{}_j=\sup_{s\in\mathbb R^k}\ \sum_{l\ne j}\frac{|s_j|}{|\gamma+1|}[\{f_j(s)\}^\frac{\gamma}{\gamma+1}f_l(s)+\{f_l(s)\}^\frac{\gamma}{\gamma+1}f_j(s)]
\end{align}
with $f_j(s)=\exp(s_j)/\{1+\sum_{l=1}^k \exp(s_l)\}.$
\end{proposition}

\begin{proof}
We confirm that $\delta_\gamma{}_j$ is a finite definite value if $\gamma<-1$ or $\gamma>0$.
It is written from \eqref{multi-g-est} that
\begin{align}\nonumber
& \big|\theta_j^\top{S}_\gamma{}_j(\theta;C)\big|\leq\frac{1}{n}\sum_{i=1}^n 
\bigg[ \{p(j|X_i,(\gamma+1)\theta)\}^{\frac{\gamma}{\gamma+1}}
\big\{1-p(j|X_i,(\gamma+1)\theta)\big\}\\[.3mm]
 &\hspace{26mm}+\sum_{l\ne j} p(l|X_i,(\gamma+1)\theta)^{\frac{\gamma}{\gamma+1}}p(j|X_i,(\gamma+1)\theta)
\bigg]|\theta_j^\top X_i|. 
\end{align}
Hence, if $S_i{}_j=(\gamma+1)\theta_j^\top X_i$, then 
\begin{align}\nonumber
\big|\theta^\top{S}_\gamma{}_j(\theta;C)\big|\leq\frac{1}{n}\sum_{i=1}^n\sum_{l\ne j} 
\bigg[ \{f_j(S_i)\}^{\frac{\gamma}{\gamma+1}}f_l(S_i)+ \{f_l(S_i)\}^{\frac{\gamma}{\gamma+1}}f_j(S_i)
\bigg]\frac{|S_i{}_j|}{\gamma+1},
\end{align}
where $S_i=(S_i{}_1,...,S_i{}_k)$.
This  concludes \eqref{sup-multi} by taking the supremum to the right-hand-side in all $S_i$'s .
\end{proof}

The $j$-th linear predictor is written by $\theta_j^\top x=\theta_{1j}^\top x_1 +\theta_{j0}$ with the slope vector  $\theta_1$ and
the intercept $\theta_0$.
The $j$-th decision boundary is given by 
\begin{align}\label{DB}
{{H}(\theta_j)}=\{x\in{\mathcal X}: \theta_{1j}^\top x_1+\theta_{0j}=0\},
\end{align} 
In a context of prediction, a predictor for the label $Y$ based on a given feature vector $x$ is give by
\begin{align}\nonumber
f(x)=\argmax_{y\in{\mathcal Y}} \theta_{y}^\top x,
\end{align}
which is equal to the Bayes rule under the multiclass logistic model, where $\theta_0=0$ in the parametrization as in \eqref{76}.
 We observe through the discussion similar to Proposition \ref{prop5} in the situation of the Bernoulli logistic model that
$\theta_j^\top\mathbb E_Q[{S}_\gamma{}_j(\theta,X,Y)|X=x]$ is uniformly bounded in $x\in{\mathcal X}$ even 
under any misspecified distribution  $Q$ outside the parametric model.   
Therefore, we result that the $\gamma$-estimator has such a stable behavior for all the sample $\{(X_i,Y_i)\}_{1\leq i\leq n}$ if  $\gamma$ is in a range $(1,-1)\cup(0,\infty)$.
The ML-estimator and the GM-estimators equals $\gamma$-estimators with $\gamma=0$ and $\gamma=-1$, respectively.
Therefore, both $\gamma$'s are outside the range, which suggests they are suffered from the unboundedness.

We next study ordinal regression, also known as ordinal classification.
Consider an ordinal outcome \(Y\) having values in ${\cal Y}=\{0,...,k\}$.
The probability of $Y$ falling into a certain category $y$ or lower is modeled as
\begin{align}\label{cum}
  {\mathbb P}(Y\leq y|X, \theta)=\frac{\exp(\theta_{0y}+\theta_1^\top X)}{1+\exp(\theta_{0y}+\theta_1^\top X)}
\end{align}
for $y=0,...,k-1$, where $\theta=(\theta_{00},...\theta_{0k},\theta_1)$.
The model \eqref{cum} is referred to as ordinal logistic model noting 
${\mathbb P}(Y\leq y|X, \theta)=F(\theta_{0y}+\theta_1^\top X)$  with  a logistic distribution $F(z)=\exp(z)/(1+\exp(z))$.
Here, the thresholds are assumed  $\theta_{00}\leq\cdots\leq\theta_{0k-1}$ to  ensure that the probability statement \eqref{cum} makes sense.
Each threshold $\theta_{0y}$ effectively sets a boundary point on the latent continuous scale, beyond which the likelihood of higher category outcomes increases. The difference between consecutive thresholds also gives insight into the "distance" or discrimination between adjacent categories on the latent scale, governed by the predictors.

For a given $n$ observations $\{(X_i,Y_i)\}_{i=1}^n$  the negative log-likelihood function 
\begin{align}\nonumber
 L_0(\theta) =-\sum_{i=1}^n \log p(Y_i|X_i,\theta)
\end{align}
where $p(y|x,\theta)=F(\theta_{0y}+\theta_1^\top x)
-F(\theta_{0y-1}+\theta_1^\top x)$.
Similarly, the $\gamma$-loss function can be given in a straightforward manner.
However, these loss functions seem complicated since the conditional probability $p(y|x,\theta)$
is introduced indirectly as a difference between the cumulative distribution functions $F(\theta_{0y}+\theta_1^\top x)$'s.

To address this issue, it is treated that each threshold as a separate binarized response, effectively turning the ordinal regression problem into multiple binary regression problems. 
Let $P(y)$ and $F(y)$ be cumulative distribution functions on $\cal Y$.
We define a  dichotomized cross entropy
\begin{align}\nonumber
H_0^{\rm (d)}(P,F) =\sum_{y=0}^k P(y)\log F(y)+(1-P(y))\log (1-F(y)).
\end{align}
This is a sum of the cross entropys between a Bernoulli distributions ${\tt Ber}(P(y))$ and ${\tt Ber}(F(y))$.
The KL divergence is given as $D_0^{\rm (d)}(P,F)=H_0^{\rm (d)}(P,F)-H_0^{\rm (d)}(P,P)$.
Thus, the dichotomized log-likelihood function is  given by
\begin{align}\nonumber
 L_0^{\rm (d)}(\theta) =\sum_{i=1}^n\sum_{y=0}^k 
Z_{iy}\log F(\theta_{0y}+\theta^\top X_i)+
(1-Z_{iy})\log \{1-F(\theta_{0y}+\theta^\top X_i)\},
\end{align}
where $Z_{iy}={\rm I}(Y_i \leq y)$.  
Note $\mathbb E[L_0^{\rm (d)}(\theta)]=H_0^{\rm (d)}(P,F(\cdot,\theta))$, where $\mathbb E$ denotes the expectation under the distribution $P$ and $F(y,\theta)=F(\theta_{0y}+\theta^\top x)$.
Under the ordinal logistic model \eqref{cum},
\begin{align}\nonumber
 L_0^{\rm (d)}(\theta) =\sum_{i=1}^n\sum_{y=0}^k 
\left[Z_{yi}(\theta_{0y}+\theta^\top X_i)-
\log \{1+\exp(\theta_{0y}+\theta^\top X_i)\}\right].
\end{align}
On the other hand, the dichotomized $\gamma$-loss function is given by
\begin{align}\nonumber
 L_\gamma^{\rm (d)}(\theta) =-\frac{1}{\gamma}\sum_{i=1}^n\sum_{y=0}^k 
\left[Z_{iy}\{F^{(\gamma)}(\theta_{0y}+\theta^\top X_i)\}^\frac{\gamma}{\gamma+1}+
(1-Z_{iy})\{1-F^{(\gamma)}(\theta_{0y}+\theta^\top X_i)\}^\frac{\gamma}{\gamma+1}\right],
\end{align}
where $F^{(\gamma)}(A)$ is the $\gamma$-expression for $F(A)$, that is,
\begin{align}\nonumber
F^{(\gamma)}(A)=\frac{\{F(A)\}^{\gamma+1}}{F(A)^{\gamma+1}+\{1-F(A)\}^{\gamma+1}}.
\end{align}
Under the ordinary logistic model \eqref{cum},
\begin{align}\nonumber
 L_\gamma^{\rm (d)}(\theta) = -\frac{1}{\gamma}\sum_{i=1}^n\sum_{y=0}^k 
\left\{\frac{\exp\{Z_{iy}(\gamma+1)(\theta_{0y}+\theta^\top X_i)\}}
{1+\exp\{(\gamma+1)(\theta_{0y}+\theta^\top X_i)\}}\right\}^\frac{\gamma}{\gamma+1}
\end{align}
If $\gamma$ is taken a limit to $-1$, then it is reduced the GM-loss function
\begin{align}\nonumber
 L_{\rm GM}^{\rm (d)}(\theta,C) =\sum_{i=1}^n\sum_{y=0}^k 
\exp\{(\half-Z_{iy})(\theta_{0y}+\theta^\top X_i)\};
\end{align}
if $\gamma=-2$, then it is reduced the HM-loss function
\begin{align}\nonumber
 L_{\rm HM}^{\rm (d)}(\theta) = \frac{1}{2}\sum_{i=1}^n\sum_{y=0}^k 
\left\{\frac{\exp\{(1-Z_{iy})(\theta_{0y}+\theta^\top X_i)\}}
{1+\exp (\theta_{0y}+\theta^\top X_i) }\right\}^2.
\end{align}
\begin{remark}
Let us discuss an extension for dichotomized loss functions to a setting where
the outcome space $\cal Y$ is a subset of $\mathbb R^d$.
Consider a partition of  ${\cal Y}$ such that ${\cal Y}=\oplus_{j=1}^k B_k$. 
Then, the model is reduced to a categorical distribution
${\tt Cat}(\pi(x,\theta))$, where $\pi(x,\theta)=(\pi_1(x,\theta),...,\pi_k(x,\theta))$ with $\pi_j(x,\theta)=\int_{B_j}p(y|x,\theta){\rm d}\Lambda(y)$. 
The cross entropy is reduced to
\begin{align}\nonumber
H_0^{\rm (d)}(\pi,\pi(x,\theta)) =\sum_{j=1}^k \pi_j\log \pi_j(x,\theta)
\end{align}
and the negative log-likelihood function is reduced to
\begin{align}\nonumber
 L_0^{\rm (d)}(\theta) =-\sum_{i=1}^n \sum_{y=0}^k {\rm I}(Y_i\in B_j) \log \pi_j(X_i,\theta).
\end{align}
Similarly, the $\gamma$-cross entropy is reduced to 
\begin{align}\nonumber
H_\gamma^{\rm (d)}(\pi,\pi(x,\theta)) =\sum_{j=1}^k \pi_j\left\{ 
\frac{\pi_j(x,\theta)^{\gamma+1}}{\sum_{j'=1}^k \pi_{j'}(x,\theta)^{\gamma+1}}\right\}^\frac{\gamma}{\gamma+1}
\end{align}
and the $\gamma$-loss function is reduced to
\begin{align}\nonumber
 L_0^{\rm (d)}(\theta) =-\sum_{i=1}^n \sum_{y=0}^k {\rm I}(Y_i\in B_j)
\left\{ 
\frac{\pi_j(X_i,\theta)^{\gamma+1}}{\sum_{j'=1}^k \pi_{j'}(X_i,\theta)^{\gamma+1}}\right\}^\frac{\gamma}{\gamma+1}.
\end{align}
\end{remark}

There are some parametric models similar to the present model including
ordered probit models, continuation ratio model and adjacent categories logit model.
The coefficients in ordinal regression models tell us about the change in the odds of being in a higher ordered category as the predictor increases. Importantly, because of the ordered nature of the outcomes, the interpretation of these coefficients gets tied not just to changes between specific categories but to changes across the order of categories.
Ordinal regression is useful in fields like social sciences, marketing, and health sciences, where rating scales (like agreement, satisfaction, pain scales) are common and the assumption of equidistant categories is not reasonable.
This method respects the order within the categories, which could be ignored in standard multiclass approaches.


\section{Poisson regression model}\label{Poisson-reg-sec}

The Poisson regression model is a member of generalized linear model (GLM),
which is typically used for count data. When the outcome variable is a count (i.e., number of times an event occurs), the Poisson regression model is a suitable approach to analyze the relationship between the count and explanatory variables. The key assumptions behind the Poisson regression model are that the mean and variance of the outcome variable are equal, and the observations are independent of each other.
The primary objective of Poisson regression is to model the expected count of an event occurring, given a set of explanatory variables. The model provides a framework to estimate the log rate of events occurring, which can be back-transformed to provide an estimate of the event count at different levels of the explanatory variables.

Let $Y$ be a response variable having a value  in ${\mathcal Y}=\{0,1,...\}$ and $X$ be a covariate variable with a value  in a subset ${\mathcal X}$ of $\mathbb R^d$.
 A Poisson distribution ${\tt Po}(\lambda)$ with an intensity parameter $\lambda$ has 
a probability mass function (pmf)  given by
\begin{align}\nonumber
p(y,\lambda)=\frac{\lambda^y}{y!}\exp(-\lambda)
\end{align}
for $y$ of ${\mathcal Y}$.
A Poisson regression model to a count $Y$ given $X=x$ is  defined by the probability distribution $P(\cdot|x,\theta)$
with pmf
\begin{align}\label{log-linear}
p(y|x,\theta)=\frac{1}{y!}\exp\{y\theta^\top x-\exp(\theta^\top x)\} .
\end{align}
The link function of the regression function to the canonical variable is a logarithmic function, $g(\lambda)=\log \lambda$,
in which \eqref{log-linear} is referred to as a log-linear model. 
The likelihood principle gives the negative log-likelihood function by
\begin{align}\nonumber
L_0(\theta)=-\frac{1}{n}\sum_{i=1}^n  \{Y_i\theta^\top X_i-\exp(\theta^\top X_i)-\log Y_i!\}.
\end{align}
for a given dataset $\{(X_i,Y_i)\}_{i=1}^n$.
Here the term $\log Y_i!$ can be neglected since it is a constant in $\theta$.
In effect, the estimating function is give by
\begin{align}\nonumber
{S}_0(\theta)=\frac{1}{n}\sum_{i=1}^n  \{Y_i-\exp(\theta^\top X_i)\}X_i.
\end{align}
We see  from the general theory for the likelihood method that the ML-estimator for $\theta$ is consistent with $\theta$. 

Next, we consider the $\gamma$-divergence and its applications to the Poisson model.
For this, we  fix a reference measure as $R={\tt Po}(\mu)$.
Then, the RN-derivative of a conditional probability measure $P(\cdot|x,\theta)$ with respect to 
$R$ is given by 
\begin{align}\nonumber
\frac{\partial P(y|x,\theta)}{\partial R}=\mu^{-y}\exp\{y\theta^\top x+\mu-\exp(\theta^\top x)\}
\end{align}
and hence the $\gamma$-expression for this is given by
\begin{align}\nonumber
p^{(\gamma)}(y|x,\theta)=\exp[(\gamma+1)y\theta^\top x-\exp\{(\gamma+1)\theta^\top x\}].
\end{align}

The $\gamma$-cross entropy between Poisson distribution is given by
\begin{align}\nonumber
& \hspace{3mm} H_{\gamma}(P(\cdot|x,\theta_0),P(\cdot|x,\theta_1); R )
\\[3mm]\nonumber
 &=-\frac{1}{\gamma}\Big[\exp\{\exp(\theta^\top_0 x+\gamma \theta^\top_1 x)\}-\frac{\gamma}{\gamma+1}\exp\{\exp((\gamma+1) \theta^\top_1 x)\}\Big],
\end{align}
where $ R $ is the reference measure defined by $ R (y)=1/y!$ for $y=0,1,...$.
Note that this choice of $ R $ enable us to having such an tractable form of this entropy. 
Hence, the $\gamma$-loss function is given by
\begin{align}\nonumber
L_{\gamma}(\theta; R )=-\frac{1}{n}\frac{1}{\gamma}\sum_{i=1}^n 
\exp\Big[\gamma Y_i \theta^\top  X_i -\frac{\gamma}{\gamma+1} \exp\{(\gamma+1) \theta^\top X_i)\}\Big].
\end{align}
The estimating function is given by
\begin{align}\nonumber
{S}_\gamma(\theta; R )=\frac{1}{n}\sum_{i=1}^n 
{S}_\gamma(\theta,X_i,Y_i; R ),
\end{align}
where
\begin{align}\label{g-poisson }
{S}_\gamma(\theta,X,Y; R )=w_\gamma (X,Y ,\theta)[Y -\exp\{(\gamma+1) \theta^\top X )\}]X .
\end{align}
where 
\begin{align}\nonumber
w_\gamma (X,Y,\theta)=
\exp\Big[\gamma Y \theta^\top  X -\frac{\gamma}{\gamma+1} \exp\{(\gamma+1) \theta^\top X)\}\Big].
\end{align}
We investigate the unbiased property for the estimating function.
\begin{proposition}
Let
\begin{align}\nonumber
\Phi_\gamma(x,y,\theta)=\theta^\top {S}_\gamma(\theta,x,y; R ),
\end{align}
where ${S}_\gamma(\theta,x,y; R ))$ is the estimating function defined in \eqref{g-poisson }.
Then, if $\gamma>0$,
\begin{align}\nonumber
\sup_{x\in{\cal X}}|\Phi_\gamma(x,y,\theta)|<\infty
\end{align}
for any fixed $y\in{\cal Y}$.
\end{proposition}
\begin{proof}
By definition, we have
\begin{align}\nonumber
|\Phi_\gamma(x,y,\theta)|\leq|\omega|
e^{\gamma y\omega-\frac{\gamma}{\gamma+1}e^{(\gamma+1)\omega}}\big(y+e^{(\gamma+1)\omega}\big).
\end{align}
where  $\omega=\theta^\top x$.
The following limit holds for positive constants $c_1,c_2, c_3$:
\begin{align}\label{ineq-gamma}
 \lim_{|\omega|\rightarrow\infty}|\omega| e^{c_1 \omega-c_2 e^{c_3\omega}}=0
\end{align}
Thus, we immediately observe
\begin{align}\nonumber
 \lim_{|\omega|\rightarrow\infty}|\omega|
e^{\gamma y\omega-\frac{\gamma}{\gamma+1}e^{(\gamma+1)\omega}}\big(y+e^{(\gamma+1)\omega}\big)=0
\end{align}
due to \eqref{ineq-gamma}.
This concludes that $|\Phi_\gamma(x,y,\theta)|$ is a bounded function in $x$ for any $y\in{\cal Y}$.
\end{proof}

It is noted that the function in \eqref{ineq-gamma}  has a mild shape as in Figure \ref{Poissonplot}. 
Thus, the property of redescending is characteristic in the $\gamma$-estimating function.
The graph is rapidly approaching to $0$ when the absolute value of the canonical value $\omega$ increases.
\begin{figure}[htbp]
\begin{center}
  \includegraphics[width=100mm]{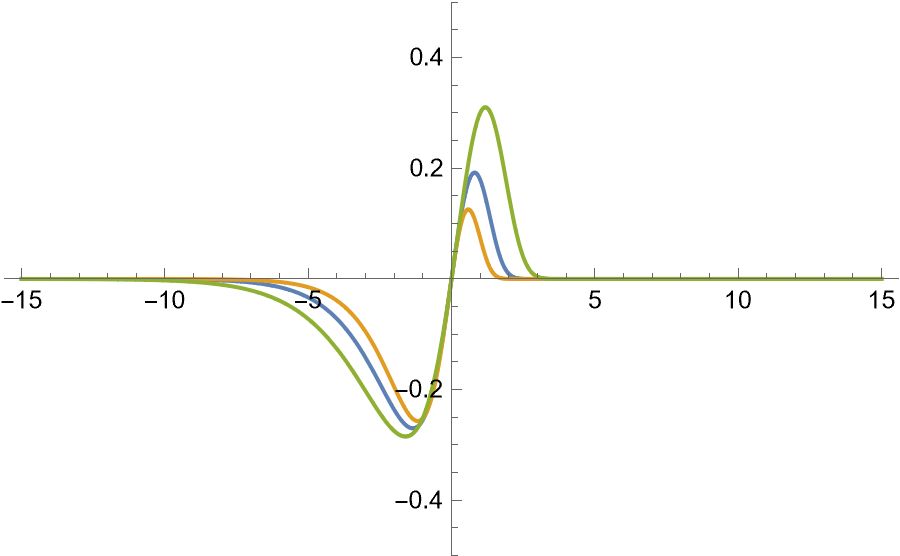}
 \end{center}
 \vspace{-2mm}\caption{Plots of $|\omega| e^{c_1 \omega-c_2 e^{c_3\omega}}$ for $(c_1,c_2,c_3)=(1,1,1),(0.8,0.8,0.8),(1.2,1.2,1.2).$}

\label{Poissonplot}
\end{figure}

We remark that $\Phi_\gamma(x,y,\theta)$ denotes the margin of the estimating function to the boundary
$\{x\in{\cal X}:\theta^\top x=0\}$.   The margin is bounded in $X$  but unbounded in $Y$, in which the behavior
 is delicate as seen in Figure \ref{Poisson-3D}.
When $\gamma>0.2$, the boundedness is almost broken down in a practical numerical sense. 
The green lines are plotted for a curve $\{(y,w,0): y=e^{(\gamma+1)w}\}.$
In this way, the margin becomes a zero on the green line.    The behavior is found mild in a region away from the green line when $\gamma$ is  a small positive value; that is found unbounded there when $\gamma$ equals a zero, or the likelihood case.
This suggests a robust and efficient property for the $\gamma$-estimator with a positive small $\gamma$.
To check this, we consider conduct a numerical experiment where there occurs a misspecification for the Poisson log-linear model $p(y|x,\theta)$ in \eqref{log-linear}.
The synthetic dataset is generated from a mixture distribution, in which   a heterogeneous subgroup is generated from a Poisson distribution  $p(y|x, \theta_{\rm hetero})$ with a small proportion $\pi$
in addition to a normal group from $p(y|x,\theta)$ with the proportion $1-\pi$.  
Here $\theta_{\rm hetro}$ is determined from plausible scenarios.
We generate  $X_i$'s  from a trivariate normal distribution ${\tt Nor}(0,0.2\,{\rm I})$ and $Y_i$'s from
\begin{align}\nonumber
 (1-\pi){\tt Po}(\exp(\theta_1^\top X_i+\theta_0))+\pi{\tt Po}(\exp(\theta_{{\rm hetero}1}^\top X_i+\theta_0)).
\end{align}
Here the intercept is set as $\theta_0=0.5$ and the slope vector is as $\theta_1=(0.5,1.5,-1.0)$ in the normal group; while the slope vector $\theta_{{\rm hetro}1}$ is set as either $-\theta_1$ or a zero vector  in the minor group. 
It is suggested that the minor group has a reverse association to the normal group, or no reaction to the covariate. 
If there is no misspecification above, or equivalently $\pi=0$, then the ML-estimator performs better than the $\gamma$-estimator. 
However, the ML-estimator is sensitive to such misspecification; the $\gamma$-estimator has robust performance, see Table \ref{gestPoi}. 
Here, the sample size is set as $n=100$ and the replication number is as $m=300$.
The $\gamma$ is selected as $0.05$, in which larger values of $\gamma$ yield unreasonable estimates.
This is because the the margin of the $\gamma$-estimator has extreme behavior as noted around Figure \ref{Poisson-3D}.

\begin{figure}[htbp]
\vspace{5mm}
\begin{center}
  \includegraphics[width=100mm]{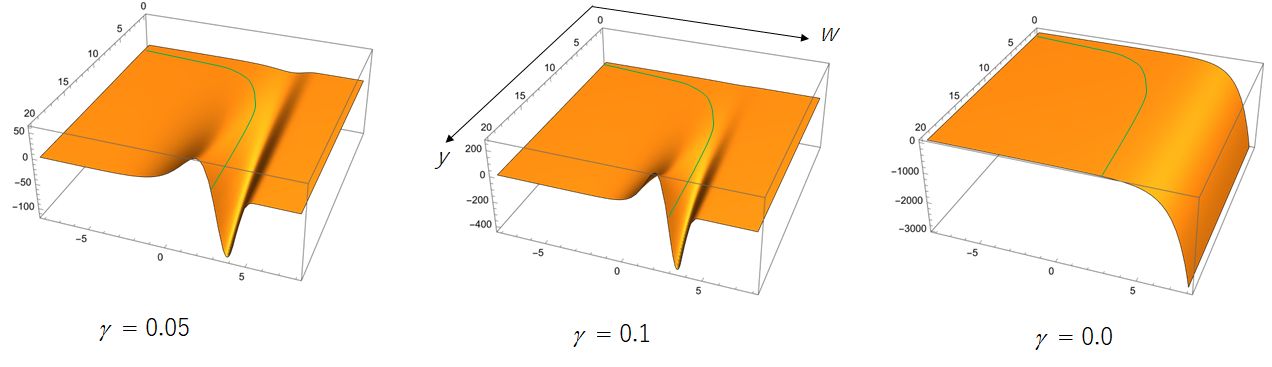}
\end{center}
 \vspace{-.5mm}\caption{3D plots of $\Phi_\gamma(x,y,\theta)$ against $y$ and $w$ where $\gamma=0.05,0.1,0.0$.}

\label{Poisson-3D}
\vspace{.5mm}
\end{figure}

\begin{table}[hbtp]

\caption{Comparison between the ML-estimator and the $\gamma$-estimator. }

  \centering
\vspace{3.2mm}

(a). The case of $\pi=0$

\vspace{3mm}
  \begin{tabular}{ccc}
  \hline 

    Method  & estimate  &  rmse
 \\
    \hline \hline

    ML-estimator & $(0.50, 1.50, -1.00, 0.49)$  & $ 0.062$\\
    $\gamma$-estimate   & $(0.59 , 1.40 , -0.93 , 0.46)$  & $0.187$ \\
  
   \hline
  \end{tabular}
\vspace{5.2mm}

(b). The case of $\pi=0.3$ and $\theta_{1 \rm hetero }=0.$

\vspace{3mm}
  \begin{tabular}{ccc}
  \hline 

    Method  & estimate  &  rmse
 \\
    \hline \hline

    ML-estimator & $(0.40 , 1.32 , -0.88 , 0.45 )$  & $ 0.310$\\
    $\gamma$-estimate   & $(0.39 , 1.47 , -0.97 , 0.48 )$  & $0.235$ \\
  
   \hline
  \end{tabular}

\vspace{5.2mm}

(c).  The case of $\pi=0.3$ and $\theta_{1 \rm hetero}=-\theta_1.$

\vspace{3mm}
  \begin{tabular}{ccc}
  \hline 

    Method  & estimate  &  rmse
 \\
    \hline \hline

    ML-estimator & $(0.41 , 1.29 , -0.88 , 0.43)$  & $ 0.386 $\\
    $\gamma$-estimate   & $(0.40 , 1.47 , -0.99 , 0.48 )$  & $0.293$ \\
  
   \hline
  \end{tabular}
  
\label{gestPoi}
\end{table}

In this section, we focus on the $\gamma$-divergence, within the framework of the Poisson regression model. The $\gamma$-divergence provides a robust alternative to the traditional ML estimator, which are sensitive to model misspecification and outliers.
The robustness of the estimator was examined from a geometric viewpoint, highlighting the behavior of the estimating function in the feature space and its relationship with the prediction level set.
The potential of $\gamma$-divergence in enhancing model robustness is emphasized, with suggestions for future research exploring its application in high-dimensional data scenarios and machine learning contexts, such as deep learning and transfer learning.
This work not only contributes to the theoretical understanding of statistical estimation methods but also offers practical insights for their application in various fields, ranging from biostatistics to machine learning.
For future work, it would be beneficial to further investigate the theoretical underpinnings of $\gamma$-divergence in a wider range of statistical models and to explore its application in more complex and high-dimensional data scenarios, including machine learning contexts like multi-task learning and meta-learning .



\section{Concluding remarks}\label{concluding}

In this chapter we provide a comprehensive exploration of MDEs, particularly $\gamma$-divergence, within regression models, see 
\cite{Naito2013,Notsu2016,Omae2017}
for other applications of unsupervised learning. 
This addresses the challenges posed by model misspecification, which can lead to biased estimates and inaccuracies, and proposes MDEs as a robust solution. 
We have discussed various regression models, including normal, logistic, and Poisson, demonstrating the efficacy of $\gamma$-divergence in handling outliers and model inconsistencies. 
In particular, the robustness for the estimator is pursued in a geometric perspective for the estimation function in the feature space.
This elucidates the intrinsic relationship between the feature space outcome space such that the behavior of the estimating function in the product space of the feature and outcome spaces is characterized the projection length to
the prediction level set.
It concludes with numerical experiments, showcasing the superiority of $\gamma$-estimators over traditional maximum likelihood estimators in certain misspecified models, thereby highlighting the practical benefits of MDEs in statistical estimation and inference.
For a detailed conclusion, it is important to recognize the significant role of $\gamma$-divergence in enhancing model robustness against biases and misspecifications. 
Emphasizing its applicability across different statistical models, the chapter underscores the potential of MDEs to improve the reliability and accuracy of statistical inferences, particularly in complex or imperfect real-world data scenarios. 
This work will not only contribute to the theoretical understanding of statistical estimation methods but also offer practical insights for their application in diverse fields, ranging from biostatistics to machine learning.

For future work, considering the promising results of $\gamma$-divergence in regression models, it could be beneficial to explore its application in more complex and high-dimensional data scenarios. 
This includes delving into machine learning contexts, such as deep learning or neural networks, where robustness against data imperfections is crucial. 
The machine learning is rapidly developing  activity areas towards generative models for documents, images and movies, in which the architecture
is in a huge scale for high-dimensional vector and matrix computations to establish pre-trained models such as large language models.
There is a challenging direction to incorporate the $\gamma$-divergence approach into such areas including multi-task leaning, transfer leaning, meta leaning and so forth.   
For example,  transfer learning is important to strengthen the empirical knowledge for target source. 
Few-shot learning is deeply intertwined with transfer learning. In fact, most few-shot learning approaches are based on the principles of transfer learning. The idea is to pre-train a model on a related task with ample data (source domain) and then fine-tune or adapt this model to the new task (target domain) with limited data. This approach leverages the knowledge (features, representations) acquired during the pre-training phase to make accurate predictions in the few-shot scenario.
Additionally, investigating the theoretical underpinnings of $\gamma$-divergence in a wider range of statistical models could further solidify its role as a versatile and robust tool in statistical estimation and inference.


%

In transfer learning, the goal is to leverage knowledge from a source domain to improve learning in a target domain. The $\gamma$-divergence can be used to ensure robust parameter estimation during this process.
Let \(\mathcal{S}\) be the source domain with distribution \(P_s\) and parameter \(\theta_s\), and \(\mathcal{T}\) be the target domain with distribution \(P_t\) and parameter \(\theta_t\).
   The objective is to minimize a loss function that incorporates both source and target domains:
   \[
   L(\theta_s, \theta_t) = L_{\mathcal{S}}(\theta_s) + \lambda L_{\mathcal{T}}(\theta_t | \theta_s)
   \]
   where \(\lambda\) is a regularization parameter balancing the influence of the source model on the target model.
   Using $\gamma$-divergence, the loss functions \(L_{\mathcal{S}}(\theta_s)\) and \(L_{\mathcal{T}}(\theta_t | \theta_s)\) are defined as:
   \[
   L_{\mathcal{S}}(\theta_s) = \frac{1}{n_s} \sum_{i=1}^{n_s} D_\gamma (P_{s, i}(\cdot | \theta_s), Q_s)
   \]
   \[
   L_{\mathcal{T}}(\theta_t | \theta_s) = \frac{1}{n_t} \sum_{i=1}^{n_t} D_\gamma (P_{t, i}(\cdot | \theta_t), P_{s, i}(\cdot | \theta_s))
   \]
   where \(Q_s\) is the empirical distribution in the source domain, and \(D_\gamma\) denotes the $\gamma$-divergence.
   The gradients for updating the parameters are given by:
   \[
   \nabla_{\theta_s} L(\theta_s, \theta_t) = \nabla_{\theta_s} L_{\mathcal{S}}(\theta_s) + \lambda \nabla_{\theta_s} L_{\mathcal{T}}(\theta_t | \theta_s)
   \]
   \[
   \nabla_{\theta_t} L(\theta_s, \theta_t) = \lambda \nabla_{\theta_t} L_{\mathcal{T}}(\theta_t | \theta_s)
   \]
   These gradients take into account the robustness properties of the $\gamma$-divergence, reducing sensitivity to outliers and model misspecifications.

In multi-task learning, the aim is to learn multiple related tasks simultaneously, sharing knowledge among them to improve overall performance. The $\gamma$-divergence helps in creating a robust shared representation.
   Let \(\mathcal{T}_i\) denote the \(i\)-th task with parameter \(\theta_i\), and let \(\Theta\) be the shared parameter space.
   The combined loss function for multiple tasks is:
   \[
   L(\Theta) = \sum_{i=1}^{m} \alpha_i L_i(\theta_i, \Theta)
   \]
   where \(\alpha_i\) are weights for each task, and \(L_i\) is the loss for task \(\mathcal{T}_i\).
   The task-specific losses \(L_i\) are defined using $\gamma$-divergence:
   \[
   L_i(\theta_i, \Theta) = \frac{1}{n_i} \sum_{j=1}^{n_i} D_\gamma (P_{i, j}(\cdot | \theta_i, \Theta), Q_i)
   \]
   where \(Q_i\) is the empirical distribution for task \(\mathcal{T}_i\).
   The gradients for updating the shared parameters \(\Theta\) and task-specific parameters \(\theta_i\) are:
   \[
   \nabla_{\Theta} L(\Theta) = \sum_{i=1}^{m} \alpha_i \nabla_{\Theta} L_i(\theta_i, \Theta)
   \]
   \[
   \nabla_{\theta_i} L(\Theta) = \alpha_i \nabla_{\theta_i} L_i(\theta_i, \Theta)
   \]
   The $\gamma$-divergence ensures that the updates are robust to outliers and anomalies within each task's data. By using a divergence measure that penalizes discrepancies between distributions, the model learns shared features that are less sensitive to noise and specific to individual tasks.

 Efficiently optimizing the $\gamma$-divergence in high-dimensional parameter spaces remains challenging.
Developing scalable algorithms that maintain robustness properties is crucial.
Further theoretical exploration of the convergence properties and bounds of $\gamma$-divergence-based estimators in transfer and multi-task learning scenarios.
Applying these robust methods to diverse real-world datasets in fields like healthcare, finance, and natural language processing to validate their practical effectiveness and robustness.
By integrating $\gamma$-divergence into transfer and multi-task learning frameworks, we can enhance the robustness and adaptability of machine learning models, making them more reliable in varied and complex data environments.


\chapter{Minimum divergence for Poisson point process}\label{Minimum divergence for Poisson point process}

\noindent{
This study introduces a robust alternative to traditional species distribution models (SDMs) using Poisson Point Processes (PPP) and new divergence measures. We propose the $F$-estimator, a method grounded in cumulative distribution functions, offering enhanced accuracy and robustness over maximum likelihood (ML) estimation, especially under model misspecification. Our simulations highlight its superior performance and practical applicability in ecological studies, marking a significant step forward in ecological modeling for biodiversity conservation.}

\section{Introduction}

Species distribution models (SDMs) are crucial in ecology for mapping the distribution of species across various habitats and geographical areas \cite{guisan2005predicting,elith2009species,merow2017integrating,sofaer2019development,monette2019ecological}. 
These models play an essential role in enhancing our understanding of biodiversity patterns, predicting changes in species distributions due to climate change, and guiding conservation and management efforts \cite{mainali2015projecting}. 
The MaxEnt (Maximum Entropy) approach to species distribution modeling represents a pivotal methodology in ecology, especially in the context of predicting species distributions under various environmental conditions \cite{phillips2004maximum}. 
This approach is particularly favored for its ability to handle presence-only data, a common scenario in ecological studies where the absence of a species is often unrecorded or unknown \cite{eguchi2022minimum}.
Alternatively, the approach based on Poisson Point Process (PPP) gives 
more comprehensive understandings for random events scattered across a certain space or time \cite{renner2015point}. 
It is particularly powerful in various fields including ecology, seismology, telecommunications, and spatial statistics. 
We review quickly  the framework for a PPP, cf. \cite{streit2010poisson} for practical applications focusing on ecological studies and \cite{komori2023statistical,hoang2015cauchy} for statistical learning perspectives.
The close relation between the MaxEnt and PPP approaches are rigorously discussed in  \cite{renner2013equivalence}. 

In this chapter, we introduce an innovative approach that employs Poisson Point Processes (PPP) along with alternative divergence measures to enhance the robustness and efficiency of SDMs \cite{saigusa2024robust}. 
We propose the use of the $F$-estimator, a novel method based on cumulative distribution functions, which offers a promising alternative to ML-estimator, particularly in the presence of model misspecification. 
Traditional approaches, such as ML estimation, often grapple with issues of model misspecification, leading to inaccurate predictions. 
Our approach is evaluated through a series of simulations, demonstrating its superiority over traditional methods in terms of accuracy and robustness. The paper also explores the computational aspects of these estimators, providing insights into their practical application in ecological studies. By addressing key challenges in SDM estimation, our methodology paves the way for more reliable and effective ecological modeling, essential for biodiversity conservation and ecological research.

Let $A$ be a subset of $\mathbb R^2$ to be provided observed points.
Then the event space is given by the collection of all possible finite subsets of $A$ as
a union of pairs $\{(m, \{s_1,...,s_m\})\}_{m=1}^\infty$ in addition to $(0,\emptyset)$,
where $\emptyset$ denotes an empty set.
Thus, the event space comprises pairs of the set of observed points $\{s_1,...,s_m\}$ and the number $m$.
Let $\lambda(s)$ be a positive function on $A$ that is called an intensity function.
A PPP defined on ${S} $ is described by the intensity function $\lambda(s)$ in a two-step procedure for any realization of ${S} $.
\begin{itemize}
\item[(i)]
The number $M$ is non-negative and generated from a Poisson distribution. This distribution, denoted as ${\tt Po}(\Lambda)$, has a probability mass function (pmf) given by
\begin{align}\nonumber
p(m, \mu)=\frac{\Lambda ^m}{m!}\exp\{-\Lambda \}
\end{align}
where $\Lambda=\int_A \lambda(s){\rm d}s$ with an intensity function $\lambda(s)$ on $A$.
\item[(ii)]
The sequence $(S_1,...,S_M)$ in $A$ is obtained by {\em independent
and identically distributed}  sample of a random variable $S$ on $A$ with
probability density function (pdf) given by
\begin{align}\nonumber
p(s)=\frac{\lambda(s)}{\Lambda }
\end{align}
for $s \in A$.
\end{itemize} 
This description  covers the basic statistical
structure of the Poison point process.  
The joint random variable $\Xi=(M,\{S_1,...,S_M\})$ has a pdf 
written as
\begin{equation}\label{pdfPPP}
p(\xi) =\exp\{-\Lambda \}\prod_{i=1}^m {\lambda(s_i)},
\end{equation}
where $\xi=(m,\{s_1,...,s_m\})$.
 Thus, the intensity function $\lambda(s)$ characterizes the pdf $p(\xi)$ of the PPP.
The set of all the intensity function has a one-to-one correspondence with the set of all the distributions of the PPPs. 
Subsequently, we will discuss divergence measures on intensity functions rather than pdfs.

\section{Species distribution model}

Species Distribution Models (SDMs) are crucial tools in ecology for understanding and predicting species distributions across spatial landscapes. 
The inhomogeneous PPP plays a significant role in enhancing the functionality and accuracy of these models due to its ability to handle spatial heterogeneity, which is often a characteristic of ecological data.
Ecological landscapes are inherently heterogeneous with varying attributes such as vegetation, soil types, and climatic conditions. 
The inhomogeneous PPP accommodates this spatial heterogeneity by allowing the event rate to vary across space, thereby enabling a more realistic modeling of species distributions.
This  can incorporate environmental covariates to model the intensity function of the point process, which in turn helps in understanding how different environmental factors influence species distribution. This is crucial for both theoretical ecological studies and practical conservation planning \cite{Komori2020}.

If presence and absence data are available, we can employ familiar statistical method such as the logistic model,
the random forest and other binary classification algorithms.
However, ecological datasets often consist of presence-only records, which can be a challenge for traditional statistical models. 
We focus on a statistical analysis for presence-only-data, in which the inhomogeneous modeling for PPPs
can effectively handle presence-only data, making it a powerful tool for species distribution model in data-scarce scenarios.

Let us introduce a SDM in the framework of  PPP discussed above.
Suppose that we get a presence dataset, say  $\{S_1,...,S_M\}$, or a set  of observed points for a species in a study area $A$.
Then, we build a statistical model of an intensity function driving a PPP on  ${A}$, in which a parametric model is given by 
\begin{align}\label{SDM}
{\mathcal M}=\{\lambda(  s, \theta): \theta\in\Theta\},
\end{align} 
called a species distribution model (SDM), 
where $ \theta$ is an unknown parameter in the space $\Theta$.
The pdf of the joint random variable  $\Xi=(M,\{S_1,...,S_M\})$ is written as
\begin{align}\nonumber
p(\xi,\theta) =
\exp\{-\Lambda(\theta) \}\prod_{i=1}^m {\lambda(s_i,\theta)}
\end{align}
due to \eqref{pdfPPP}, where $\xi=(m,\{s_1,...,s_m\})$ and $\Lambda( \theta)=\int_A\lambda(  s, \theta)d  s$.
 In ecological terms, this can be understood as recording the locations (e.g., GPS coordinates) where a particular species has been observed. The pdf here helps in modeling the likelihood of finding the species at different locations within the study area, considering various environmental factors.
Typically, we shall consider a log-linear model 
\begin{align}\label{log-linear-PPP}
\lambda(s,\theta)=\exp\{\theta_1^\top x(s)+\theta_0\}
\end{align} 
with $\theta=(\theta_0,\theta_1)$,  a feature vector $x(s)$,  a slope vector $\theta_1$ and an intercept $\theta_0$.
Here $x(s)$ consists environmental characteristics such as geographical, climatic and other factors influencing the habitation of the species.  
Then, parameter estimation is key in SDMs to understand the relationships between species distributions and environmental covariates. 
The ML-estimator is a common approach used in PPP to estimate these parameters, which in turn, refines the SDM.

The negative log-likelihood function based on an observation sequence  $(M,\{  S_1,...,  S_M\})$ is given by
\begin{align}\label{ell}
L_0( \theta)=
-\sum_{i=1}^M  \log \lambda(  S_i,\theta)+ \Lambda(\theta).
\end{align}
Here the cumulative intensity is usually approximated as
\begin{align}
 \Lambda(\theta)
= \sum_{i=1}^n  w_i \lambda(  S_i, \theta)
\end{align}
by Gaussian quadrature, where ${  S_{M+1},...,   S_n}$ are the centers of the grid cells containing no presence location and $w_i$ is a quadrature weight for a grid cell area. 
The approximated estimating equation is given by
\begin{align}\label{like}
{S}_0(\theta) =
  \sum_{j=1}^n \{Z_i-w_i \lambda(  S_i , \theta)\}
  \frac{\partial}{\partial \theta}\log  \lambda(  S_i , \theta) ={0},
\end{align}
where $Z_j$ is an indicator for presence, that is, $Z_j=1$ if $1\leq j \leq M$ and $0$ otherwise.

Let $p(\xi)$ and $q(\xi)$ be probability density functions (pdf) of two PPPs,
where $\xi=(m,\{s_1,...,s_m\})$ is a realization.
Due to the discussion above, the pdfs are written as
\begin{align}\label{pq}
p(\xi) =
\exp\{-\Lambda \}\prod_{i=1}^m {\lambda(s_i)}, \ \ \ q(\xi) =
\exp\{-H \}\prod_{i=1}^m {\eta(s_i)},
\end{align}
in which $p(\xi)$ and $\lambda(s)$ have a one-to-one correspondence, and
 $q(\xi)$ and $\eta(s)$ have also the same property.
The Kullback-Leibler (KL) divergence between $p$ and $q$ is defined by the difference between the cross entropy
and the diagonal entropy as ${ D}_{0}(p,q)={  H}_0(p,q)-{H}_0(p,p)$, where the cross entropy is defined by
\begin{align}\nonumber{  H}_0(p,q)=-\mathbb E_{p} [\log {q(\Xi)}]\end{align}
with the expectation $\mathbb E_p$ with the pdf $p(\xi)$.
This is written as
\begin{align}
{H}_0(p,q)
&= -\sum_{m=0}^\infty \frac{\Lambda^m}{m!}e^{-\Lambda}
\int_{A\times\cdots\times A}  \log \{e^{-H}\prod_{j=1}^m {\eta(s_j)}\}\prod_{j=1}^m\frac{\lambda(s_j)}{\Lambda}{\rm d}s_j
\nonumber\\[2mm]
&= \sum_{m=0}^\infty \frac{\Lambda^m}{m!}e^{-\Lambda}
\Big[-H + \frac{m}{\Lambda}
 \int_A {\lambda(s)}\log{\eta(s)}{\rm d}s\Big]
 \nonumber\\[2mm]
&=  \int_A\{ \lambda(s) \log {\eta(s)}-\eta(s)\}{\rm d}s.
\end{align}
Thus,   the KL-divergence is 
\begin{align}\label{kl}
{ D}_{0}(p,q)= \int_A\Big\{ \lambda(s) \log \frac{\lambda(s)}{\eta(s)}-\lambda(s)+\eta(s)\Big\}{\rm d}s,
\end{align}
see \cite{komori2023statistical} for detailed derivations.
This can be seen as a way to assess the effectiveness of an ecological model. 
For instance, how well does our model predict where a species will be found, based on environmental factors like climate, soil type, or vegetation? 
The closer our model's predictions are to the actual observations, the better it is at explaining the species' distribution.
In effect,  ${D}_0(p,q)$ coincides with the extended KL-divergence between intensity functions
$\lambda(s)$ and $\eta(s)$.
Here, the term $-\lambda(s)+\eta(s)$ in the integrand of \eqref{kl} should be added
to the standard form since both $\lambda(s)$ and $\eta(s)$ in general do not have total mass one. 

Let $(M,\{  S_1,...,  S_M\})$ be an observations having a pdf $p(\xi, \theta_0)$.
We consider the expected value under the true pdf $p(\xi,\theta_0)$ that is given by
\begin{align}\nonumber
\mathbb E_0[L_0( \theta)]= 
-\int_A\{ \lambda(s,\theta_0) \log \lambda(s,\theta ) -\lambda(s,\theta )\}{\rm d}s
\end{align}
noting a familiar formula for a random sum in PPP:
 \begin{align}\label{random-sum}
\mathbb E_0\Big[\sum_{i=1}^M \log \lambda(S_i,\theta)\Big]=\int_{A} \lambda(s,\theta_0)\log\lambda(s,\theta){\rm d}s,
\end{align}
where $\mathbb E_0$ denotes the expectation with the pdf $p(\xi,\theta_0)$.
This is nothing but the cross entropy between intensity functions $\lambda(s,\theta_0)$ and $\lambda(s,\theta)$.
In accordance, we observe a close relationship between the log-likelihood and the KL-divergence that is
parallel to the discussion around \eqref{Pytha} in chapter 3.  
In effect,
\begin{align}\nonumber
\mathbb E_0[L_0( \theta)]-\mathbb E_0[L_0( \theta_0)]=D_0(p(\cdot,\theta_0),p(\cdot,\theta)) . 
\end{align}
This relation concludes the consistency of the ML-estimator for the true value $\theta_0$ noting
$\theta_0=\argmin_{\theta\in\Theta}E_0[L_0( \theta)]$.
This suggests that the method used to estimate the impact of environmental factors on species distribution is dependable. In practical terms, this means ecologists can trust the model to make accurate predictions about where a species might be found, based on environmental data.

\section{Divergence measures on intensity functions}\label{Div-on-Int}

 We would like to extend the minimum divergence method for estimating to estimating a SDM.
The main objective is to propose an alternative to the maximum likelihood method, aiming to enhance robustness and expedite computation.
We have observed  the close relationship between the log-likelihood and the KL-divergence in the previous section.
Fortunately, the empirical form of the KL-divergence is matched with the log-likelihood function in the framework of the SDM.
We remark that  a fact that the KL-divergence between PPPs is equal to the KL-divergence  between its intensity functions is essential for ensuring this property.
However, this key relation does not hold in the situation for the power divergence. 

First, we review a formula for random sum and product in PPP, which is gently and comprehensively discussed in \cite{streit2010poisson}.  

\begin{proposition}\label{random}
Let $\Xi=(M,\{S_1,...,S_M\})$ be a realization of PPP with an intensity function $\lambda(s)$ on an area $A$. 
Then, for any integrable function $g(s)$,
\begin{align}\label{random-sum}
\mathbb E\Big[\sum_{i=1}^M g(S_i)\Big]=\int_{A} g(s)\lambda(s){\rm d}s 
\end{align}
and
\begin{align}\nonumber 
\mathbb E\Big[\prod_{i=1}^M g(S_i)\Big]=\exp\Big\{\int_{A} \{g(s)-1\}\lambda(s)\Big\}{\rm d}s .
\end{align}

\end{proposition}
\begin{proof}
By definition,
\begin{align}\nonumber 
& \mathbb E\Big[\sum_{i=1}^M g(S_i)\Big]=\sum_{m=0}^\infty \frac{\Lambda^m}{m!}e^{-\Lambda}
\int_{A\times\cdots\times A}  \sum_{i=1}^m g(s_i)\prod_{i=1}^m\frac{\lambda(s_i)}{\Lambda}{\rm d}s_j
\nonumber\\[2mm]
&= \sum_{m=0}^\infty \frac{\Lambda^m}{m!}e^{-\Lambda}
 \frac{m}{\Lambda}
 \int_A {\lambda(s)} {g(s)}{\rm d}s
 \nonumber\\[2mm]
&=  
\int_{A} g(s)\lambda(s){\rm d}s .
\nonumber
\end{align}
Similarly,
\begin{align}\nonumber 
& \mathbb E\Big[\prod_{i=1}^M g(S_i)\Big]=\sum_{m=0}^\infty \frac{\Lambda^m}{m!}e^{-\Lambda}
\int_{A\times\cdots\times A}  \prod_{i=1}^m g(s_i)\prod_{i=1}^m\frac{\lambda(s_i)}{\Lambda}{\rm d}s_j
\\[2mm]
&= \sum_{m=0}^\infty \frac{\Lambda^m}{m!}e^{-\Lambda}
 \Big\{
 \int_A \frac{g(s)\lambda(s)}{\Lambda}{\rm d}s\Big\}^m
 \nonumber\\[2mm]
&=  \exp\Big\{\int_{A} \{g(s)-1\}\lambda(s){\rm d}s\Big\} .
\nonumber
\end{align}

\end{proof}
Proposition \ref{random} gives interesting properties of the random sum with the random product, see
section 2.6 in \cite{streit2010poisson} for further discussion and historical backgrounds.
 In ecology, this can be interpreted as predicting the total impact or effect of a particular environmental factor (represented by $g(s)$) across all locations where a species is observed within a study area $A$. For example, $g(s)$ could represent the level of a specific nutrient or habitat quality at each observation point $S_i$. The integral then sums up these effects across the entire habitat, providing a comprehensive view of how the environmental factor influences the species across its distribution.
This formula can be used in SDMs to quantify the cumulative effect of environmental variables on species presence. 
For instance, it could help in assessing how total food availability or habitat suitability across a landscape influences the likelihood of species presence. 
By integrating such ecological factors into the SDM, researchers can gain insights into the species' habitat preferences and distribution patterns.
Understanding the cumulative impact of environmental factors is crucial for conservation planning and management. This approach helps identify critical areas that contribute significantly to species survival and can guide habitat restoration or protection efforts. For instance, if the model shows that certain areas have a high cumulative impact on species presence, these areas might be prioritized for conservation.

Second, we introduce divergence measures to apply the estimation for a species distributional model employing the formula introduced in Proposition \ref{random}.
The $\gamma$-divergence between probability measures $P$ and $Q$ of PPPs with RN-derivatives    $p(\xi)$ and $q(\xi)$  in  \eqref{pq} is given by  
 \begin{align}\nonumber
D_{\gamma}(P,Q)=H_{\gamma}(P,Q)-H_{\gamma}(P,P).
\end{align} 
Here the cross $\gamma$-entropy is defined by
\begin{align}\nonumber
H_{\gamma}(P,Q)=-\frac{1}{\gamma}\frac{\mathbb E_P[q(\Xi)^\gamma]}
{\{\mathbb E_Q[q(\Xi)^\gamma]\}^\frac{\gamma}{\gamma+1}},
\end{align}
where $\mathbb E_P$ denote the expectation with respect to $P$.
Accordingly, the $\gamma$-cross entropy between probability distributions  $P_0$ and $P_\theta$  having the intensity functions  $\lambda_0(s)$ and $\lambda(s,\theta)$, respectively, is written as
\begin{align}\nonumber
H_{\gamma}(P_0,P_\theta)=-\frac{1}{\gamma}\exp\Big[
\int_A \Big\{\lambda_0(s)(\lambda(s,\theta)^\gamma-1)-\frac{\gamma}{\gamma+1}\lambda(s,\theta)^{\gamma+1}    \Big\}{\rm d}s\Big]
\end{align}
since
\begin{align}\nonumber
{\mathbb E_{P_0}[p(\Xi,\theta)^\gamma]}&= 
\exp\{\gamma\Lambda(\theta)\} \mathbb E_{P_0}\Big[\prod_{i=1}^M \lambda(S_i,\theta)^\gamma \Big]\\[3mm]
&= 
\exp\Big[\int_A \{\gamma \lambda(s,\theta)+\lambda_0(s)(\lambda(s,\theta)^\gamma-1)\} {\rm d}s\Big]
\end{align}
due to Proposition \ref{random}.
However, it is difficult to give an empirical expression of $H_{\gamma}(P_0,P_\theta)$   for a given realization  $(M,\{S_1,...,S_M\})$ generated from $P_0$.
In accordance, we consider another type of divergence.

Consider the log $\gamma$-divergence between $P_0$ and $P_\theta$ that is defined by 
\begin{align}\nonumber
\Delta_{\gamma}(P_0,P_\theta)=-\frac{1}{\gamma}\log \frac{\mathbb E_{P_0}[p(\Xi,\theta)^\gamma]}
{\{\mathbb E_{P\theta}[p(\Xi,\theta)^\gamma]\}^\frac{\gamma}{\gamma+1}\{\mathbb E_{P_0}[p_0(\Xi)^\gamma]\}^\frac{1}{\gamma+1}}.
\end{align}
This is written as
\begin{align}\label{log-gamma}
{\Delta}_{\gamma}(P_0,P_\theta)=-\frac{1}{\gamma}
\int_A \Big\{\lambda_0(s)\lambda(s,\theta)^\gamma-\frac{\gamma}{\gamma+1}\lambda(s,\theta)^{\gamma+1}
-\frac{1}{\gamma+1}\lambda_0(s)^{\gamma+1}    \Big\}{\rm d}s.
\end{align}
Therefore, the loss function is induced as
\begin{align}\nonumber
L_\gamma(\theta)=-\frac{1}{\gamma}
\sum_{i=1}^M \lambda(S_i,\theta)^\gamma +\frac{1}{\gamma+1}\int_A\lambda(s,\theta)^{\gamma+1} {\rm d}s
\end{align}
for a SDM \eqref{SDM}.
This loss function has totally different from the negative log-likelihood.
In a regression model, $D_\gamma(P,Q)$ and $\Delta_\gamma(P,Q)$ yield the same loss function; while in a PPP model
only $\Delta_\gamma(P,Q)$  yield the loss functions $L_\gamma(\theta)$.

We observe that the property about random sum and product leads to delicate differences among one-to-one transformed
divergence measures. So, we consider a divergence measure directly defined on the space of intensity functions other than on that  of probability distributions of PPPs. 
The $\beta$-divergence is given by
\begin{align}\label{beta-intensity}
{D}_{\beta}(\lambda,\eta)=-\frac{1}{\beta}
\int_A \Big\{\lambda(s)\eta(s)^\beta-\frac{\beta}{\beta+1}\eta(s)^{\beta+1}
-\frac{1}{\beta+1}\lambda(s)^{\beta+1}    \Big\}{\rm d}s;
\end{align}
The $\gamma$-divergence is given by
\begin{align}\nonumber
{D}_{\gamma}(\lambda,\eta)=-\frac{1}{\gamma}
\frac{\int_A \lambda(s)\eta(s)^\gamma {\rm d}s}{\{\int_A\eta(s)^{\gamma+1}{\rm d}s\}^\frac{\gamma}{\gamma+1}}
+\frac{1}{\gamma}\Big\{\int_A\lambda(s)^{\gamma+1}  {\rm d}s  \Big\}^\frac{1}{\gamma+1}.
\end{align}
The loss functions corresponding to these are given by
\begin{align}\label{beta-loss-PPP}
L_\beta(\theta)=-\frac{1}{\beta}\sum_{i=1}^M  \lambda(S_i,\theta)^\beta +\frac{1}{\beta+1}\int_A\lambda(s,\theta)^{\beta+1} {\rm d}s
\end{align}
and
\begin{align}\label{gamma-loss-PPP}
 L_\gamma(\theta)=-\frac{1}{\gamma}\frac{\sum_{i=1}^M  \lambda(S_i,\theta)^\gamma}{\{ \int_A\lambda(s,\theta)^{\gamma+1} {\rm d}s\}^\frac{\gamma}{\gamma+1}}.
\end{align}
And, the estimating functions corresponding to these are given by
\begin{align}\nonumber
{S}_\beta(\theta)=\sum_{i=1}^M  \lambda(S_i,\theta)^\beta \frac{\partial}{\partial\theta}
\log \lambda(S_i,\theta) -\int_A\lambda(s,\theta)^{\beta+1}\frac{\partial}{\partial\theta}
\log \lambda(s,\theta) {\rm d}s
\end{align}
and
\begin{align}\nonumber
{S}_\gamma(\theta)=&\sum_{i=1}^M \frac{ \lambda(S_i,\theta)^\gamma}{\{ \int_A\lambda(s,\theta)^{\gamma+1} {\rm d}s\}^\frac{\gamma}{\gamma+1}}\\[3mm]
&\times \Big\{\frac{\partial}{\partial\theta}
\log \lambda(S_i,\theta) -\int_A\frac{\lambda(s,\theta)^{\gamma+1}}
{\int_A\lambda(s,\theta)^{\gamma+1}{\rm d}s}\frac{\partial}{\partial\theta}
\log \lambda(s,\theta) {\rm d}s
\Big\}. \label{gamma-intensity}
\end{align}

A divergence measure between two PPPs is written by a functional with respect to intensity functions induced by them.
We observe an interesting relationship from this viewpoint.

\begin{proposition}\label{gamma-beta}
Let $P$ and $Q$ be probability distributions for PPPs with intensity functions $\lambda$ and $\eta$, respectively.
Then, the log $\gamma$-divergence $\Delta_\gamma(P,Q)$ in \eqref{log-gamma}is equal to the $\beta$-divergence $D_\beta(\lambda,\eta)$ in \eqref{beta-intensity}
when $\gamma=\beta$.
\end{proposition}
Proof is immediate by definition.

Essentially, $\Delta_\gamma(P,Q)$ satisfies the scale invariance which expresses an angle between $P$ an$Q$ rather than a distance between them; $D_\beta(\lambda,\eta)$ does not satisfy such invariance in the intensity function space.
Thus, they are totally different characteristics, however the connection of probability distributions  and their intensity functions for PPPs entails this coincidence.                 
It follows from Proposition \ref{gamma-beta} that the GM-divergence $D_{\rm GM}(P,Q)$ equals the Itakura-Saito divergence, that is
\begin{align}\nonumber
{D}_{\rm GM}(P,Q)=\int_A \Big\{\frac{\lambda(s)}{\eta(s)}-\log \frac{\lambda(s)}{\eta(s)}-1   \Big\}{\rm d}s.
\end{align}
Hence, the GM-loss function is given by
\begin{align}\nonumber
L_{\rm GM}(\theta)=\sum_{i=1}^M  \frac{1}{\lambda(S_i,\theta)} +\int_A\log \lambda(s,\theta) {\rm d}s.
\end{align}
and the estimating function is
\begin{align}\label{HM0}
{S}_{\rm GM}(\theta)=\sum_{i=1}^M  \frac{1}{\lambda(S_i,\theta)}\frac{\partial}{\partial\theta}
\log \lambda(S_i,\theta) +\int_A\frac{\partial}{\partial\theta}
\log \lambda(s,\theta)  {\rm d}s.
\end{align}
\begin{proposition}\label{GMfast}
Assume a log-linear model $\lambda(s,\theta)=\exp\{\theta_1^\top x(s)+\theta_0\}$ with a feature vector $x(s)$.
Then, the estimating function of the $GM$-estimator is given by
\begin{align}\label{HMestfn}
{S}_{\rm GM}(\theta)=\sum_{i=1}^M \exp\{-\theta_1^\top x(S_i)-\theta_0\}
 x_0(S_i)
-\int_A
x_0(s)
 {\rm d}s.
\end{align}
where $x_0(s)=(1,x(s)^\top)^\top$.
\end{proposition}
\begin{proof}
Proof is immediate. Equation \eqref{HMestfn} is seen by applying \eqref{HM0} to the log linear model. 
\end{proof}
Equating  ${S}_{\rm GM}(\theta)$  to zero satisfies
\begin{align}\label{GMequ}
  \sum_{i=1}^M \exp\{- \theta ^\top  x_0(S_i)\}x_0(S_i) = \sum_{j=1}w(S_j)x_0(S_j)
\end{align} 
by the quadrature approximation. This implies that the inverse intensity weighted mean for presence data $\{x(S_i)\}_{i=1}^M$ is equal to the region
mean for $\{x(S_j)\}_{j=1}^n$. 
The learning algorithm to solve the estimating equation \eqref{GMequ} to get the GM-estimator needs only the update of the inverse intensity mean for the presence during without any updates for the region mean during the iteration process.  
In this regard, the computational cost for the $\gamma$-estimator  frequently becomes huge for a large set of quadrature points.    
For example,  it needs to evaluate the quadrature approximation in the likelihood equation \eqref{like} during iteration process for obtaining the ML-estimator.
On the other hand, such evaluation step is free in the algorithm for obtaining the GM-estimator.

Finally, we look into an approach to the minimum divergence defined on the space of pdf's.
A intensity function $\lambda(s,\theta)$ determines  the pdf for an occurrence of a point $s$ by $p(s,\theta)=\lambda(s,\theta)/\Lambda(\theta)$. 
From a point of this, we can consider the divergence class of pdfs, which has been discussed in Chapter 2.
However, this approach has a weak point such that, in a log-linear model $\lambda(s,\theta)=\exp\{\theta_1 x(s)+\theta_0\}$, such a pdf transformation cancels the
intercept parameter $\theta_0$ as found in 
\begin{align}\nonumber
p(s,\theta)=\frac{\exp\{\theta_1 x(s)\}}{\int_A
\exp\{\theta_1 x(\tilde s)\}{\rm d} \tilde s}.
\end{align}  
Therefore, the maximum entropy method is based on such an approach, so that the intercept parameter cannot be consistently estimated.  See \cite{renner2013equivalence} for the detailed discussion in a rigorous framework.
Here we note that  the $\gamma$-divergence between $p=\lambda/\Lambda$ and $q=\eta/H$ is essentially equal to
that between $\lambda $ and $\eta$, that is 
$$
{D}_{\gamma}(\lambda,\eta) =\Lambda {D}_{\gamma}(p,q).
$$
This implies that  the intercept $\theta_0$ is  not estimable in the estimating function \eqref{gamma-intensity}.  
In deed, the $\gamma$-loss function \eqref{gamma-loss-PPP} for the log-linear model is reduced to
\begin{align}\nonumber
 L_\gamma(\theta)=-\frac{1}{\gamma}\frac{\sum_{i=1}^M  \exp\{\gamma(\theta_1^\top x(S_i)\}}{\{ \int_A
 \exp\{(\gamma+1)\theta_1^\top x(s)\} {\rm d}s\}^\frac{\gamma}{\gamma+1}},
\end{align}
which is constant in $\theta_0$.
From the scale invariance of the log $\gamma$ divergence,  $\Delta_\gamma(p,q)=\Delta_\gamma(\lambda,\eta)$ noting $p$ and $q$ equals $\lambda$ and $\eta$ up to a scale factor.
Similarly, the intercept parameter is not identifiable.
On the other hand, the $\beta$-loss function \eqref{beta-loss-PPP} is written down as
\begin{align}\nonumber
L_\beta(\theta)=-\frac{1}{\beta}\sum_{i=1}^M  \exp \{\beta(\theta_1^\top x(S_i)+\theta_0)\} +\frac{1}{\beta+1}\int_A\exp \{(\beta+1)(\theta_1^\top x(s)+\theta_0)\} {\rm d}s
\end{align}
in which $\theta_0$ is estimable.


\section{Robust divergence method}

We discuss robustness for estimating the SDM defined by a parametric intensity function $\lambda(s,\theta)$.
In particular, a log-linear model $\lambda(s,\theta)=\theta_1 x(s)+\theta_0$, where $\theta=(\theta_0,\theta_1)$ and
$x(s)$ is environmental feature vector influencing on a habitat of a target species.
In Section \ref{Div-on-Int} we discuss the minimum divergence estimation for the SDM in which
power divergence measures are explored  on the space of the PPP distributions, on that of intensity functions, or
on that of pdfs in exhaustive manner.
In the examination, the minimum $\beta$-divergence method defined on the space of intensity functions is recommended for its reasonable property for the consistency of estimation.

We look at the $\beta$-estimating function for a given dataset $(M,\{S_1,...,S_M\})$ that is defined as
\begin{align}\label{b-esteqPPP}
{S}_\beta(\theta)=\sum_{j=1}^n \{Z_j-w_j\exp(\theta_1^\top x(S_j)+\theta_0)\} \exp \{\beta(\theta_1^\top x(S_j)+\theta_0)\}x_0(S_j) ,
\end{align}
where $x_0(S_j) =(1,x(S_j)^\top )^\top$, $w_j$ is a quadrature weight on the set $\{S_1,...,S_n\}$ combined with presence and background grid centers, and  $Z_j$ is an indicator for presence, that is, $Z_j=1$ if $1\leq j \leq M$ and $0$ otherwise.
Note that taking $\beta=0$ as a specific choice yields the likelihood equation \eqref{like}.
Alternatively, taking a limit of $\beta$ to $-1$ entails the GM-estimating function as 
\begin{align}\label{GM-esteqPPP}
{S}_{\rm GM}(\theta)&=\sum_{j=1}^n \{Z_j\exp \{-(\theta_1^\top x(S_j)+\theta_0)\}-w_j\} x_0(S_j) ,\\[3mm]
&=\sum_{i=1}^M\exp \{-(\theta_1^\top x(S_i)+\theta_0)\}x_0(S_i)-\bar x_0 ,
\end{align}
where $\bar x_0= \sum_{j=1}^n w_jx_0(S_j)$.
This leads to a remarkable cost reduction for the computation in
the learning algorithm as discussed after Proposition \ref{GMfast}.
Here the computation in \eqref{GM-esteqPPP} is only for the first term of presence data with one evaluation $\bar x_0$
using background data.
For any $\beta$, the $\beta$-estimating function is unbiased.
Because 
\begin{align}\nonumber
\mathbb E_\theta[{S}_{\beta}(\theta)] =&\int_A \exp \{(\beta+1)(\theta_1^\top x(s)+\theta_0)\}x_0(s){\rm d}s\\[3mm]\nonumber
&-\sum_{j=1}^n w_j\exp \{(\beta+1)(\theta_1^\top x(S_j)+\theta_0)\} x_0(S_j) ,
\end{align}
which is equal to a zero vector if the the quadrature approximation is proper, where $\mathbb E_\theta$ denote expectation under the log-linear model $\lambda(s,\theta)=\exp(\theta_1 x(s)+\theta_0)$. 
This unbiasedness property guarantees the consistency of the $\beta$-estimator for $\theta$.
In accordance with this, we would like to select the most robust estimator for model misspecification in the class of
$\beta$-estimators.  
The difference of the $\beta$-estimator with the ML-estimator is focused only on the estimating weight $\exp \{\beta(\theta_1^\top x(S_j)+\theta_0)\}$ in \eqref{b-esteqPPP}.
We contemplate that the estimating weight would be not effective for any data situation regardless of whether $\beta$ is positive or negative.
Indeed, the estimating weight becomes unbounded and has no comprehensive understanding for misspecification. 

We consider another estimator rather than the $\beta$-estimator for seeking robustness of misspecification based on the discussion above \cite{saigusa2024robust}.
A main objective is  to change the estimating weight for the $\beta$-estimator into a more comprehensive form.
Let $F$ be a cumulative distribution function defined on $[0,\infty)$.
Then, we define an estimating function
\begin{align}\label{b-esteqPPP}
{S}_F(\theta)=\sum_{j=1}^n \{Z_j-w_j\exp(\theta_1^\top x(S_j)+\theta_0)\}
F(\sigma \exp(\theta_1^\top x(S_j)+\theta_0))x_0(S_j),
\end{align}
where $\sigma>0$ is a hyper parameter.
We call the $F$-estimator by $\hat\theta_F$ defined by a solution for  the equation  ${S}_F(\theta)=0$.
Immediately, the unbiasedness for ${S}_F(\theta)$ can be confirmed.
In this definition, the estimating weight is given as $F(\sigma \exp(\theta_1^\top x(S_j)+\theta_0))$.
For example, we will use a Pareto cumulative distribution function 
$$
F(t)=1-(1+\eta ct)^{-\frac{1}{\eta}}, 
$$
where 
 a shape parameter $\eta>0$ is fixed to be $1$ in a subsequent experiment.
We envisage that $F$ expresses the existence probability of the intensity for the presence  of the target species. 
Hence, a low value of the weight implies a low probability of   the presence.  
The plot of the estimating weights $F(\sigma \lambda(S_j,\theta))$ for $i=1,...,M$
would be helpful if we knew the true value.

Suppose that the log-linear model for the given data would be misspecified.
We consider a specific situation for misspecification such that
\begin{align}\label{out}
\lambda(s) = (1-\epsilon)\lambda(s,\theta)+\epsilon \lambda_{\rm out}(s).
\end{align}
This implies there is a contamination of a subgroup with a probability $\epsilon$ in a major group with
the intensity function correctly specified.
Here the subgroup has the intensity function $\lambda_{\rm out}(s)$ that is far away from the log-linear model $\lambda(s,\theta)$.
Geometrically speaking, the underlying intensity function $\lambda(s)$ in \eqref{out} is a tubular neighborhood
surrounding the model   $\{\lambda(s,\theta):\theta\in\Theta\}$ with a radius $\epsilon$ in the space of all intensity functions.
In this misspecification, we expect that the estimating wights for the subgroup should be suppressed than
those for the major group.  
It cannot be denied in practical situations that there is always a risk of model misspecification.
it is comparatively easy to find outliers for presence records or background data cause by mechanical errors.
Standard data preprocessing procedures helpful for data cleaning, however it is difficult to find outliers under such a latent structure for misspecification.
In this regard, the approach by the $F$-estimator is promising to solve the problem for such difficult situations. 
The hyper parameter $\sigma$ should be selected by a cross validation method to give its effective impact on the estimation process. 
We will discuss on enhancing the clarity and practical applicability of the concepts in this approach as a future work.

We have a brief study for numerical experiments.
Assume that a feature vector set  $\{X(s_j)\}_{j=1}^n$ on presence and background grids is generated from a bivariate normal distribution ${\tt Nor}(0,{\rm I})$, where $\rm I$ denotes a 2-dimensional identity matrix.
Our simulation scenarios for the intensity function was organized  as follows.

\

\noindent
(a). {Specified model:} 
{$\hspace{6mm} \lambda(s)=\lambda(s,\theta_0)$, where $\lambda(s,\theta_0)=\exp\{\theta_{01}^\top X(s)+\theta_{00}\}.$

\

\noindent
(b). Misspecified model: 
{$\hspace{.9mm} \lambda(s)=(1-\epsilon)\lambda(s,\theta_0)+\epsilon\lambda(s,\theta_*)$, where $\theta_*=(\theta_{00},-\theta_{01}).$

\

Here parameters were set as $\theta_0=0.5$, $\theta_1=(1.5,-1.0)^\top$,  and  $\pi=0.1$.
In case (b), $\lambda(s,\theta_*)$ is a specific example of $ \lambda_{\rm out}(s)$ in \eqref{out}, which implies that the subgroup has the intensity function with the negative parameter to the major group.
In ecological studies, a major group in the species might thrive under conditions where a few others do not, and vice versa. Using a negative parameter could imitate this kind of inverse relationship.
See Figure \ref{Presence} for the plot of presence numbers against two dimensional feature vectors.
The presence numbers were  generated from the misspecified model (b) with the simulation number $1000$. 

\vspace{-3.0mm}
\begin{figure}[htbp]
\begin{center}
  \includegraphics[width=100mm]{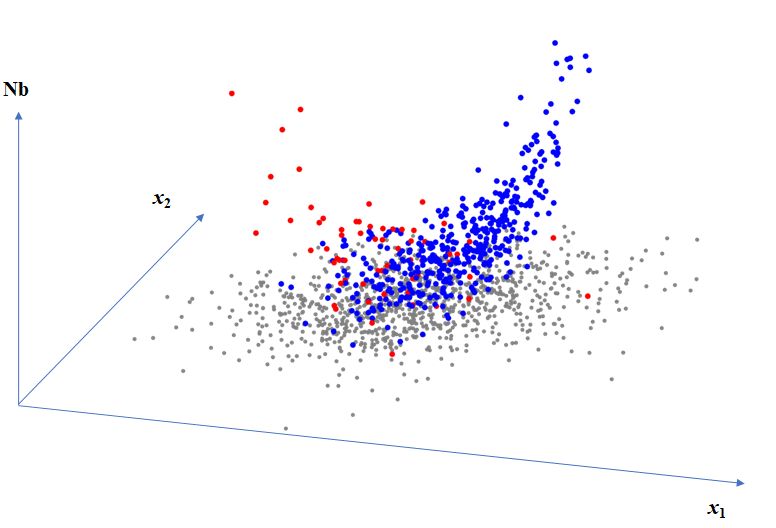}
 \end{center}
 \vspace{-2.0mm}\caption{ Plot of  presence numbers with major group (red), minor group (blue) and backgounds (gray).}
\label{Presence}
\end{figure} 

\vspace{5mm}
  
We compared the estimates the ML-estimator $\hat\theta_0$, the $F$-estimator $(\sigma=1.2)$   and the $\beta$-estimator $\hat\theta_{\beta}$  $(\beta=-0.1)$, where the simulation was conducted by 300 replications.
In the the case of specified model, the ML-estimator was slightly superior to the $F$-estimator and $\beta$-estimator in a point of the
root means square estimate (rmse), however the superiority over the $F$-estimator is just little.
Next, we suppose a mixture model of two intensity functions  in which one was the log-linear model as the above with the mixing ratio  $0.9$; the other was still a log-linear model  but the the slope vector was the minus one with the mixing ratio  $0.1$.
Under this setting, $F$-estimator was especially superior to the ML-estimator and the $\beta$-estimator, where the rmse of the ML-estimator is more than double that of the $F$-estimator.  The $\beta$-estimator shows less robust for this misspecification.
Thus, the ML-estimator is sensitive to the presence of such a heterogeneous subgroup; the $F$-estimator is robust.
It is considered that the estimating weight $F(\lambda(s,\theta)$ effectively performs to suppress the influence of the subgroup  in the estimating function of the $F$-estimator.    
See Table \ref{F-SDM} for details and and Figure \ref{F-Boxplot} for the plot of three estimators in the case (b).  
We observe in the numerical experiment that the $F$-estimator has almost the same performance as that of the ML-estimator when the model is correctly specified; the $F$-estimator has more robust performance than the ML-estimator when the model is partially misspecified
in a numerical example.
\begin{table}[hbtp]

\caption{Comparison among the ML-estimator, $F$-estimator  and  $\beta$-estimator. }

  \centering
\vspace{3.2mm}

(a). The case of specified model

\vspace{3mm}
  \begin{tabular}{ccc}
  \hline 

    Method  & estimate  &  rmse
 \\
    \hline \hline

    \vspace*{1mm}   
 ML-estimator & $({0.499 , 1.497 , -1.003 })$  & $ 0.152 $\\
    $F$-estimate   & $({0.501 , 1.496 , -1.002 })$  & $0.156 $ \\
    $\beta$-estimate   & $({0.707 , 1.498 , -1.004 })$  & $0.259 $ \\

   \hline
  \end{tabular}
\vspace{5.2mm}

(b). The case of  misspecified model

\vspace{3mm}
  \begin{tabular}{ccc}
  \hline 

    Method  & estimate  &  rmse \vspace{1mm}
 \\
    \hline \hline

    \vspace*{1mm}   
 ML-estimator & $({0.930 , 1.208 , -0.784 })$  & $0.869  $\\
    $F$-estimate   & $({0.562 , 1.410 , -0.935 })$  & $0.379  $ \\
    $\beta$-estimate   & $({1.211 , 1.163 , -0.750 })$  & $1.068 $ \\

   \hline
  \end{tabular}

\label{F-SDM}
\end{table}

\

\begin{figure}[htbp]
\hspace{5.28mm}
  \includegraphics[width=100mm]{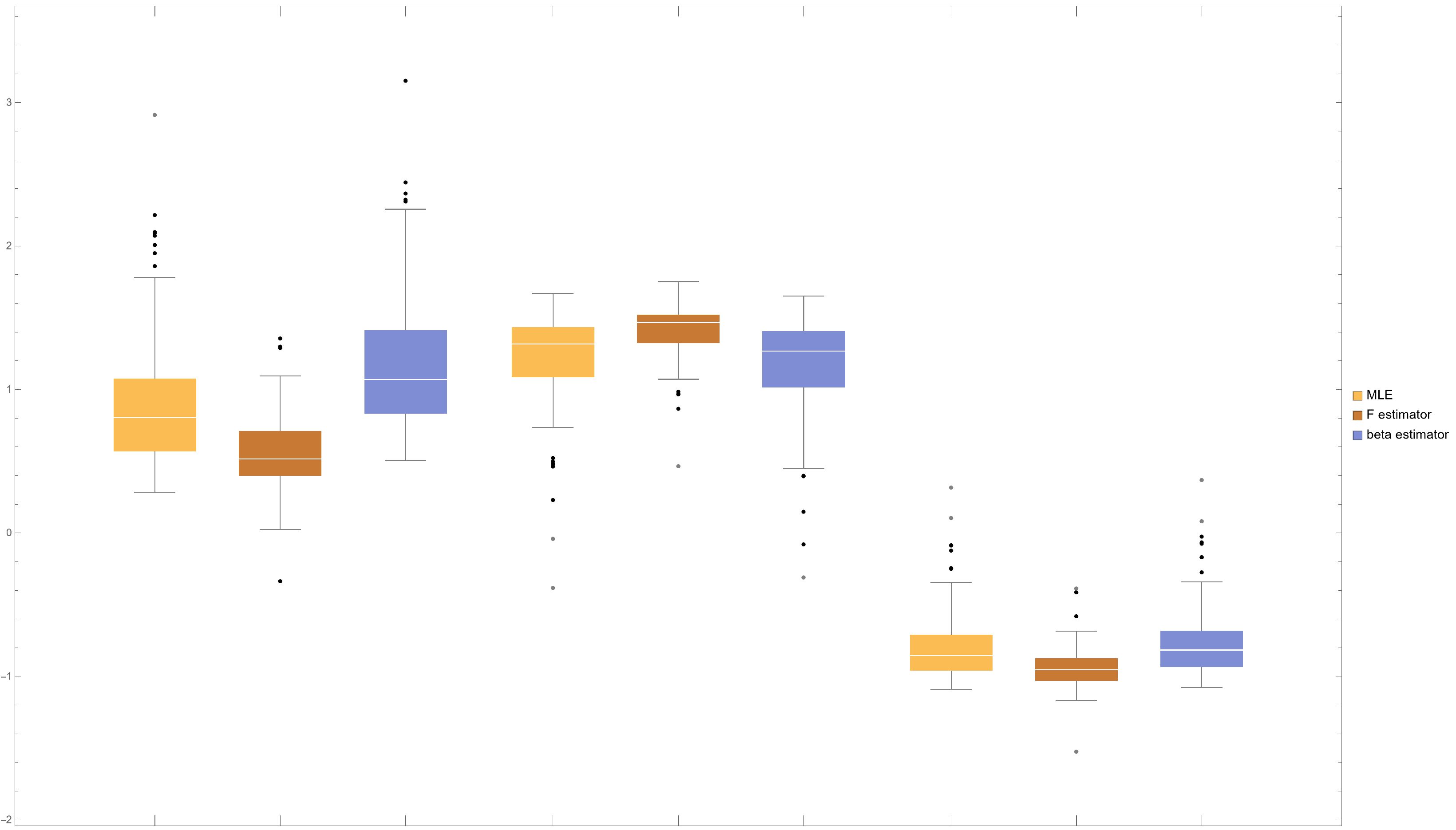}
 \vspace{-2mm}\caption{Box-whisker Plots of  the ML-estimator,  $F$-estimator and  $\beta$-estimator}
\label{F-Boxplot}
\end{figure}

We discuss a problem that the $F$-estimator $\hat\theta_F$ is defined by the solution of the equation  ${S}_F(\theta)=0$; while the objective function is not given.
However, it follows from Poincar\'{e} lemma that there is a unique objective function $L_F(\theta)$ such that
$\hat\theta_F=\argmin_{\theta\in\Theta}L_F(\theta)$.
See \cite{bott1982differential} for geometric insights in the Poincar\'{e} lemma.
Because ${S}_F(\theta)$ is integrable since the Jacobian matrix of ${S}_F(\theta)=0$ is symmetric.   
In effect, we have the solution as follows.

\

\begin{proposition}\label{LOSS}
Let $F$ be a cumulative distribution function on $[0,\infty)$.
Consider a loss function for a model $\lambda(s,\theta)$ defined by
\begin{align}\nonumber
L_F(\theta)=-\sum_{i=1}^M a_F(\lambda(S_i,\theta))+\sum_{j=1}^n w_jb_F(\lambda(S_j,\theta)),
\end{align}
where 
$ 
a_F(\lambda)=\int_0^\lambda \frac{F(z)}{z}dz$ 
and $b_F(\lambda)=\int_0^\lambda  {F(z)}dz.
$ 
Then, if the the model is a log-linear model $\lambda(s,\theta)=\exp\{\theta_1^\top x(s)+\theta_0\}$, the estimating function is given by ${S}_F(\theta)$ in \eqref{b-esteqPPP}.
\end{proposition} 
\begin{proof}
The gradient vector of $L_F(\theta)$ is given by
\begin{align}\nonumber
{S}_F(\theta)=-\sum_{i=1}^M a_F^\prime(\lambda(S_i,\theta))\frac{\partial}{\partial\theta}\lambda(S_j,\theta)
+\sum_{j=1}^n w_jb^{\prime}(\lambda(S_j,\theta))\frac{\partial}{\partial\theta}\lambda(S_j,\theta).
\end{align}
This is written as
\begin{align}\nonumber
{S}_F(\theta) &=\sum_{i=1}^M \frac{F(\lambda(S_i,\theta))}{\lambda(S_i,\theta)}
\frac{\partial}{\partial\theta}\lambda(S_j,\theta)
+\sum_{j=1}^n w_jF(\lambda(S_j,\theta))\frac{\partial}{\partial\theta}\lambda(S_j,\theta)\\[3mm]
&=\sum_{j=1}^n(Z_j -w_j{\lambda(S_i,\theta)})F(\lambda(S_j,\theta))\frac{\partial}{\partial\theta}\log\lambda(S_j,\theta),
\label{generalx}
\end{align}
where $Z_j$ is the presence indicator.  
Hence, we conclude that ${S}_F(\theta)$ is equal to that given in \eqref{b-esteqPPP} under a log-linear model.
\end{proof}
The $F$-estimator is derived by minimization of this loss function $L_F(\theta)$.
Hence, we have a question whether there is a divergence measure that induces to $L_F(\theta)$. 

\begin{remark}
We have a quick review of the Bregman divergence that is defined by 
\begin{align}\nonumber
D_U(\lambda,\eta)=\int[U(\lambda(s))-U(\eta(s))-U^\prime(\eta(s))\{\lambda(s)-\eta(s)\}]{\rm d}s
\end{align}
where $U$ is a convex function.
The loss function is given by
\begin{align}\nonumber
L_U(\theta)=-\sum_{i=1}^MU^\prime(\lambda(S_i,\theta))+\sum_{j=1}^n w_j\{\lambda(S_j,\theta)U^\prime(\lambda(S_j,\theta))-U(\lambda(S_j,\theta)))\}
\end{align}
and the estimating function is
\begin{align}\nonumber
{S}_U(\theta)=\sum_{j=1}^n\{Z_j- w_j\lambda(S_j,\theta)\}
{U^{\prime\prime}(\lambda(S_j,\theta))}{\lambda(S_j,\theta)}\frac{\partial}{\partial\theta}\log\lambda(S_j,\theta).
\end{align}
Therefore, we observe that, if $F(z)=U^{\prime\prime}(z)z$, then ${S}_F(\theta)={S}_U(\theta)$, where ${S}_F(\theta)$
is defined in \eqref{generalx}.
This implies that the divergence, $D_F(\lambda,\eta)$,  associated with ${S}_F(\theta)$ is equal to the Bregman divergence $D_U(\lambda,\eta)$ with
the generator $U$ satisfying $F(z)=U^{\prime\prime}(z)z$.
That is,
\begin{align}\nonumber
D_F(\lambda,\eta)=\int[A_F(\lambda(s))-A_F(\eta(s))+a_F(\eta(s))-b_F(\eta(s)) \lambda(s)]{\rm d}s,
\end{align}
where $a_F$ and $b_F$ are defined in Proposition \ref{LOSS} and $A_F(z)=\int_0^Z a_F(s){\rm d}s$.
\end{remark}

\section{Conclusion and Future Work}

Our study marks a significant advance in the field of species distribution modeling by introducing a novel approach that leverages Poisson Point Processes (PPP) and alternative divergence measures. 
The key contribution of this work is the development of the $F$-estimator, a robust and efficient tool designed to overcome the limitations of traditional ML-methods, particularly in cases of model misspecification. 

 The $F$-estimator, based on cumulative distribution functions, demonstrates superior performance in our simulations. This robustness is particularly notable in handling model misspecification, a common challenge in ecological data analysis.
Our approach provides ecologists and conservationists with a more reliable tool for predicting species distributions, which is crucial for biodiversity conservation efforts and ecological planning.
 We also explored the computational aspects of the $F$-estimator, finding it to be computationally feasible for practical applications, despite its advanced statistical underpinnings.
While our study offers significant contributions, it also opens up several avenues for future research:
Further validation of the $F$-estimator in diverse ecological settings and with different species is necessary to establish its generalizability and practical utility.
The integration of the $F$-estimator with other types of ecological data, such as presence-only data, would enhance its applicability.
There is scope for further refining the computational algorithms to enhance the efficiency of the $F$-estimator, making it more accessible for large-scale ecological studies.
Exploring the applicability of this method in other scientific disciplines, such as environmental science and geography, could be a fruitful area of research.
In conclusion, our work not only contributes to the theoretical underpinnings of species distribution modeling but also has significant practical implications for ecological research and conservation strategies.

The intensity function is modeled based on environmental variables, reflecting habitat preferences. This process typically involves a dataset of locations where species and environmental information have been observed, along with accurate and high-quality background data. With precise training on these datasets, reliable predictions can be derived using maximum likelihood methods in Poisson point process modeling. These predictions are easily applied using predictors integrated into the maximum likelihood estimators. 
While Poisson point process modeling and the maximum likelihood method can derive reliable predictions from observed data, predicting for `unsampled` areas that differ significantly from the observed regions poses a significant challenge \cite{yates2018outstanding,meyer2021predicting}.

The ability to predict the occurrence of target species in unobserved areas using datasets of observed locations, environmental information, and background data is a pivotal issue in species distribution modeling (SDM) and ecological research. Applying these models to regions that differ significantly from those included in the training dataset introduces several technical and methodological challenges.
When unobserved areas differ substantially from the observed regions,  predicting the occurrence of target species in unobserved areas remains a critical issue. 
To address this issue, exploring predictions based on the similarity of environmental variables is essential. One promising approach relies on ecological similarity rather than geographical proximity, making it particularly effective for species with wide distributions or fragmented habitats. Additionally, by adopting a weighted likelihood approach and linking Poisson point processes through a probability kernel function between observed and unobserved areas, it becomes possible to efficiently predict the probability of species occurrence in unobserved areas.
We believe that the methodologies developed in this study will inspire further innovations in statistical ecology and beyond.


\chapter{Minimum divergence in machine leaning}
\label{Minimum divergence in machine leaning}


\noindent{
We discuss  divergence measures and applications encompassing some areas of machine learning, Boltzmann machines, gradient boosting, active leaning and cosine similarity.
Boltzmann machines have wide developments for generative models by the help of statistical dynamics.
The ML-estimator is a basic device for data learning, but the computation is challenging for evaluating the partition functions. 
We introduce the GM divergence and the GM estimator for the Boltzmann machines.
The GM estimator is shown a fast computation thanks to free evaluation of the partition function.
Next, we focus on on active learning, particularly the Query by Committee method. It highlights how divergence measures can be used to select informative data points, integrating statistical and machine learning concepts.
Finally, we extend the $\gamma$-divergence on a space of real-valued functions.
This yields a natural extension of the cosine similarities, called $\gamma$-cosine similarities.
The basic properties are explored and demonstrated in numerical experiments compared to traditional cosine similarity. 
}


\section{Boltzmann machine}


Boltzmann Machines (BMs) are a class of stochastic recurrent neural networks that were introduced in the early 1980s, crucial in bridging the realms of statistical physics and machine learning, see \cite{hinton1986learning,hinton2012practical} for the mechanics of BMs, and 
\cite{goodfellow2016deep} for  comprehensive foundations in the theory underlying neural networks and deep learning. 
They have become fundamental for understanding and developing more advanced generative models.
Thus, BMs are statistical models that learn to represent the underlying probability distributions of a dataset. 
They consist of visible and hidden units, where the visible units correspond to the observed data and the hidden units capture the latent features.
Usually, the connections between these units are symmetrical, which means the weight matrix is symmetric. 
The energy of a configuration in a BM is calculated using an energy function, typically defined by the biases of units and the weights of the connection between units.
The partition function is a normalizing factor used to ensure that the probabilities sum to 1  summing  exponentiated negative energy over all possible configurations of the units
\cite{goodfellow2016deep}. 

Training a BM involves adjusting the parameters (weights and biases) to maximize the likelihood of the observed data. 
This is often done via stochastic maximum likelihood  or contrastive divergence.
 The log-likelihood gradient has a simple form, but computing it exactly is intractable due to the partition function. Thus, approximations or sampling methods like Markov chain Monte Carlo  are used.
 BMs have been extended to more complex and efficient models like Restricted BMs and deep belief networks.
They have found applications in dimensionality reduction, topic modeling, and collaborative filtering among others.
We overview the principles and applications of BMs, especially in exploring the landscape of energy-based models and the geometrical insights into the learning dynamics of such models. 
The exploration of divergence, cross-entropy, and entropy in the context of BMs might yield profound understandings, potentially propelling advancements in both theoretical and practical domains of machine learning and artificial intelligence.






Let $\mathcal P$ be the space of all  probability mass functions defined on a finite discrete set ${\mathcal X}=\{-1,1\}^d$, that is
\begin{align}\nonumber
{\mathcal P}=\Big\{p( x): p( x)>0 \ (\forall  x\in{\mathcal X})\ , \sum_{x\in{\mathcal X}}p( x)=1\Big\},
\end{align} 
in which  $p( x)$ is called a $d$-variate Boltzmann distribution.
A standard BM in $\mathcal P$ is introduced as
\begin{align}\nonumber
p( x, \theta)=\frac{1}{Z_{ \theta}}\exp\{-E( x, \theta)\}
\end{align}
for $ x\in{\mathcal X} $, where $E( x, \theta)$ is the energy function defined by
\begin{align}\nonumber
E( x, \theta)=- b^\top x- x^\top {  W}  x
\end{align}
with a parameter  $ \theta=( b, {  W})$.
Here  $Z_{ \theta}$ is the partition function defied by
\begin{align}\nonumber
Z_{ \theta}=\sum_{x\in{\mathcal X}}\exp\{-E( x, \theta)\}.
\end{align}
The Kullback-Leibler (KL) divergence is written as
\begin{align}\nonumber
D_{0}(p, p(\cdot, \theta))=\sum_{  x\in{\mathcal X}}p( x)
\log \frac{p( x)}{p( x, \theta)}
\end{align}
which involves the partition function $Z_{ \theta}$.
The negative log-likelihood function for a given dataset $\{ x_i\}_{i=1,...,N}$ is written as
\begin{align}\nonumber
L_0( \theta)=\sum_{i=1}^N E( x_i, \theta)+ N\log Z_{ \theta}
\end{align}
and the estimating function is given by
\begin{align}\nonumber
{S}_{\rm ML}( \theta)=\frac{1}{N}\sum_{i=1}^N  
\begin{bmatrix}
\!\!\!\! x_i -\mathbb E_{ \theta} [ X] \vspace{1mm}\\
 x_i  x_i^\top -\mathbb E_{ \theta} [ X X^\top] 
\end{bmatrix}
\end{align}
where $\mathbb E_{ \theta}$ denotes the expectation with respect to $p( x, \theta)$.
In practice, the likelihood computation is known as an infeasible problem because of a sequential procedure with large summation including $\mathbb E_{ \theta}$ or $Z_{\theta}$. 
There is much literature to discuss approximate computations such as variational approximations and Markov-chain Monte Carlo simulations
\cite{hinton2002training}.

On the other hand, we observe that the computation for the GM-divergence $D_{\rm GM }$
does not need to have any evaluation for the partition function as follows:
the GM-divergence is defined by
 \begin{align}\label{gm1}
D_{\rm GM }(p, p(\cdot, \theta))=\sum_{  x\in{\mathcal X}} \frac{p( x)}{p( x, \theta)}\prod_{ x\in{\mathcal X}}p( x, \theta)^{\frac{1}{m}}.
\end{align}
This is written as
 \begin{align}\nonumber
D_{\rm GM }(p, p(\cdot, \theta))=
\sum_{  x\in{\mathcal X}} {p( x)}{\exp\{E( x, \theta)-\bar{E}( \theta)\}}
\end{align}
where $\bar{E}( \theta)$ is an averaged energy given by
$ 
\bar{E}( \theta)={\frac{1}{m}}\sum_{  x\in{\mathcal X}}E( x, \theta).
$ 
Note that the averaged energy is written as
 \begin{align}\label{bias}
\bar{E}( \theta)=  b^\top \bar{ x}+ {\rm tr}({  W}\overline{ x x^\top}),
\end{align}
where $ \bar{ x}={\frac{1}{m}}\sum_{  x\in{\mathcal X}}  x$ and
$ \overline{ x  x^\top}={\frac{1}{m}}\sum_{  x\in{\mathcal X}} x  x^\top$.
\vspace{1mm}
We observe that $D_{\rm GM }(p, p(\cdot, \theta))$ is free from the partition term $Z_{ \theta}$ due to the cancellation of $Z_{ \theta}$ in multiplying the two terms in the right-hand-side of \eqref{gm1}.

For a given dataset $\{ X_i\}_{i=1,...,N}$, the GM-loss function for $ \theta$ is
defined by
 \begin{align}\nonumber
L_{\rm GM }( \theta) =\frac{1}{N}
\sum_{i=1}^N {\exp\{E( X_i, \theta)-\bar{E}( \theta)\}}
\end{align}
and the minimizer $\hat{ \theta}_{\rm GM }$ is called the GM-estimator.
The estimating function is given by
 \begin{align}\nonumber
{S}_{\rm GM }( \theta) =\frac{1}{N}
\sum_{i=1}^N {\exp\{E( X_i, \theta)-\bar{E}( \theta)\}}
\begin{bmatrix}
 X_i -\bar{X}\vspace{1mm}\\
 X_i  X_i^\top- \overline{ XX^\top} 
\end{bmatrix}
.
\end{align}
In accordance, the computation for finding $\hat{ \theta}_{\rm GM }$
is drastically lighter than that for the ML-estimator.
For example, a stochastic gradient algorithm can be applied in feasible manner.
In some cases, a Newton-type algorithm may still be applicable, which is suggested as
 \begin{align}\nonumber
 \theta\longleftarrow
\Big\{\sum_{i=1}^N {\exp\{E( X_i, \theta)-\bar{E}( \theta)\}}{  S}(X_i)\Big\}^{-1}{S}_{\rm GM }( \theta)
\end{align}
where 
 \begin{align}\nonumber
{  S}( X_i)=
\begin{bmatrix}
 X_i -\bar{X}\vspace{1mm}\\
 X_i  X_i^\top- \overline{XX^\top}  
\end{bmatrix}
\begin{bmatrix}
X_i^\top -\bar{X}^\top ,
 X_i  X_i^\top- \overline{XX^\top}  
\end{bmatrix}.
\end{align}
 This discussion is applicable for the deep BMs with restricted BMs.

Here we have a small simulation study to demonstrate the  fast computation for the GM estimator compared to the ML-estimator.
Basically, the computation time can vary based on the hardware specifications and other running processes.
This simulation is done by Python program on the Google Colaboratory \\ (https://research.google.com/colaboratory).   
Keep in mind that the computation of the partition function $Z_\theta$ can be extremely challenging for large dimensions due to the exponential number of terms. 
For simplicity, this implementation will not optimize this calculation and might not be feasible for very large dimensions. 
It is notable that the computation time for the log-likelihood increased significantly with the higher dimension, which is expected due to the exponential increase in the number of states that need to be summed over in the partition function. 
On the other hand, the computation time for the GM loss remains relatively low, which reinforces its computational efficiency, particularly in higher dimensions.  
it is not feasible to directly compute the log-likelihood times for dimensions up to 20 within a reasonable time frame using the current method. 
As shown, the computation time for the log-likelihood increases significantly with the dimension, reflecting the computational complexity due to the partition function. 
On the other hand, the GM loss computation time remains relatively low and stable across these dimensions. 
This trend suggests that while the GM estimator maintains its computational efficiency in higher dimensions, the ML-estimator becomes increasingly impractical due to the exponential growth in computation time. For dimensions beyond this range, especially approaching $D=20$, one might expect the computation time for the log-likelihood to become prohibitively long, further emphasizing the advantage of the GM loss method in high-dimensional settings. 
Figure \ref{BM1} focuses on computing and plotting the computation times for log-likelihood and GM loss across dimensions $d= 5,...,10$ of the BM, where the sample size $n$ is fixed as $100$.
This result is consistent with our observation that the GM loss might offer a more computationally efficient alternative to the ML-estimator, especially as the dimensionality of the problem increases.
For a case of the higher dimension mare than $10$, the naive gradient algorithm for the ML-estimator cannot converge in the limited time; that for the GM estimator works well if $d\leq100$.
When $d=100$ and $N=5000$, the computation time is approximately 0.811 seconds. 

\begin{figure}[htbp]
\begin{center}
  \includegraphics[width=80mm]{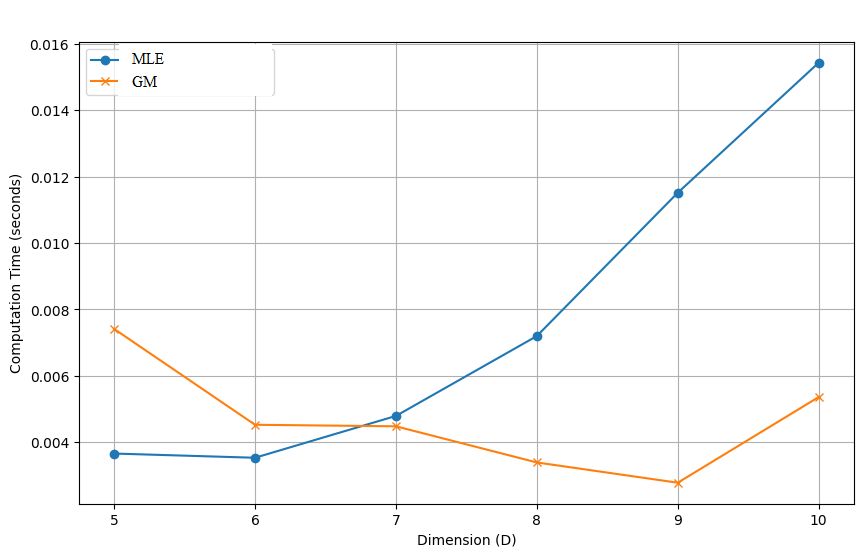}
 \end{center}
 \vspace{-2mm}\caption{Computation time for  the ML-estimator and GM-estimator for d=5,...,10}
\label{BM1}
\end{figure}

Consider a Boltzmann distribution with visible and hidden units
\begin{align}\nonumber
p( x,  h, \theta)=\frac{1}{Z_{ \theta}}\exp\{-E( x,  h , \theta)\}
\end{align}
for $( x, h)\in{\mathcal X}\times{\mathcal H}$, where ${\mathcal X}=\{0,1\}^{d}$, ${\mathcal H}=\{0,1\}^{\ell}$ and $E( x,  h , \theta)$ is the energy function defined by
\begin{align}\nonumber
E( x,  h, \theta)=- b^\top x- c^\top  h- x^\top {  W}  h
\end{align}
with a parameter  $ \theta=( b,  c, {  W} )$.
Here  $Z_{ \theta}$ is the partition function defied by
\begin{align}\nonumber
Z_{ \theta}=\sum_{( x, h)\in{\mathcal X}\times{\mathcal H} }\exp\{-E( x, h, \theta)\}.
\end{align}
The marginal distribution is given by
\begin{align}\nonumber
p( x, \theta)=\frac{1}{Z_{ \theta}}
\sum_{  h\in{\mathcal H}}\exp\{-E( x,  h , \theta)\},
\end{align}
and the GM-divergence is given by
 \begin{align}\nonumber
D_{\rm GM }(p, p(\cdot, \theta))&=
\sum_{  x \in{\mathcal X} } {p( x )}\Big[\sum_{ h\in{\mathcal H}}\exp\{-E( x,  h,   \theta)\}
\Big]^{-1}\prod_{  x \in{\mathcal X} }\Big[\sum_{  h\in{\mathcal H}}\exp\{-E( x,  h,   \theta)\}
\Big]^{\frac{1}{m}}
\\[2mm]\label{gm2}
&=  \sum_{  x \in{\mathcal X} } {p( x )}{\exp\{\tilde E( x , \theta)-\bar{  E}( \theta)\}}
\end{align}
where $
\bar{  E}( \theta)={\frac{1}{m}}\sum_{  x\in{\mathcal X}}\tilde
E( x, \theta)
$ 
and
 \begin{align}\nonumber
\tilde{E}( x, \theta)=- \log \sum_{  h\in{\mathcal H}}\exp\{-E( x,  h,   \theta)\}.
\end{align}
Note that the bias term $\bar {  E}( \theta)$ is not written by the sufficient statistics as in \eqref{bias}.
For a given dataset $\{ x_i\}_{i=1,...,N}$, the GM-loss function for $ \theta$ is
defined by
 \begin{align}\nonumber
L_{\rm GM }( \theta)&= \frac{1}{N}
\sum_{i=1}^N 
\Big[\sum_{  h\in{\mathcal H}}\exp\{-E( X_i,  h,   \theta)\}
\Big]^{-1}\prod_{  x \in{\mathcal X} }\Big[\sum_{ h\in{\mathcal H}}\exp\{-E( x,  h,   \theta)\}
\Big]^{\frac{1}{m}}\\
&= \frac{1}{N}
\sum_{i=1}^N {\exp\{\tilde E(X_i, \theta)-\bar{E}( \theta)\}}.
\end{align}
The estimating function is given by
 \begin{align}
{S}_{\rm GM }( \theta) &= \frac{1}{N}
\sum_{i=1}^N {\exp\{\tilde E( X_i, \theta)-\bar{E}( \theta)\}}
\begin{bmatrix}
X_i  -\bar{X}\vspace{1.5mm}\\
\mathbb E_{ \theta}[ H|X_i]-\overline{\mathbb E_{ \theta}[ H| X]}\vspace{1.5mm}\\
\mathbb E_{ \theta}[ X_i  H^\top| X_i] -
\overline{  \mathbb E_{ \theta}[X H^\top|X ]} \label{h-esti},
\end{bmatrix}
\end{align}
where
\begin{align}\label{h-condi}
p( h| x, \theta)=\frac{p( x, h, \theta)}{p( x, \theta)}=
\frac{\exp\{-E( x, h, \theta)\}}{\exp\{-\tilde E( x , \theta)\}}
\end{align}
and
\begin{align}\nonumber
\overline{  \mathbb E_{ \theta}[ H| X ]}={\frac{1}{m}}\sum_{  x\in{\mathcal X}}
\mathbb E_{ \theta}[ H| x], \ \ \
\overline{  \mathbb E_{ \theta}[ X H^\top|X ]}={\frac{1}{m}}\sum_{ x\in{\mathcal X}} \mathbb E_{\theta}[ x H^\top| x ].
\end{align}
{
For $ H=(H_k)_{k=1}^\ell$, the conditional distributions of $H_k$'s given $ x$ are conditional independent as seen in
\begin{align}\label{h-condi}
p( h| x, \theta) &= 
\frac{\exp\{ b^\top x+ c^\top  h+ x^\top {  W}  h\}}
{\sum_{  h '\in{\mathcal H}}\exp\{ b^\top x+ c^\top  h'+ x^\top {  W}  h'\}}\\[2mm]\nonumber
&= \prod_{k=1}^{\ell}
\frac{\exp\{c_k  h_k+ \sum_{j=1}^d x_j {W}_{jk} h_k\}}
{1+ \exp\{c_k + \sum_{j=1}^d x_j {W}_{jk}\}},
\end{align}
and hence,
\begin{align}\nonumber
\overline{  \mathbb E_{ \theta}[H_k | X ]}={\frac{1}{m}}\sum_{ x } 
\frac{\exp\{c_k + \sum_{j=1}^d x_j {W}_{jk}\}}
{1+ \exp\{c_k + \sum_{j=1}^d x_j {W}_{jk}\}}
\end{align}
and
\begin{align}\nonumber
\overline{  \mathbb E_{ \theta}[   H_k | X_{(-j)} ]}={\frac{1}{m}}\sum_{\ \   x_{(-j)} }
\frac{\exp\{c_k + {W}_{jk}+\sum_{j'\ne j} x_{j' }{W}_{j'k}\}}
{1+ \exp\{c_k + {W}_{jk}+\sum_{j'\ne j} x_{j'} {W}_{j'k}\}},
\end{align}
where $ x_{(-j)}=(x_1,...,x_{j-1},x_{j+1},...,x_d)$.

}
Note that the conditional expectation $\mathbb E_{ \theta}[\ \cdot  \ | x]$ in the estimating function  \eqref{h-esti} can be evaluated by
$p( h| x, \theta)$ in \eqref{h-condi} that is free from the partition function $Z_{ \theta}$.
A stochastic gradient algorithm is easily implemented in a fast computation.

Next, consider a Boltzmann distribution connecting visible and hidden units to an output
variable as
\begin{align}\nonumber
p( x,  h,y, \theta)=\frac{1}{Z_{ \theta}}\exp\{-E( x,  h,y , \theta)\}
\end{align}
for $( x, h,y)\in{\mathcal X}\times{\mathcal H}\times{\mathcal Y}$, where $E( x,  h ,y, \theta)$ is the energy function defined by
\begin{align}\nonumber
E( x,  h,y, \theta)=- b^t x- c^t  h  d^t 
- x^t {  W}  h- h^t {  U}  e(y)
\end{align}
with $ e(y)=(0,...,0,\overset{(y)}{1},0,...,0)^t$ for $y\in{\mathcal Y} $ and a parameter  $ \theta=( b,  c,  d,{  W},{  U} )$.
Here  $Z_{ \theta}$ is the partition function defied by
\begin{align}\nonumber
Z_{ \theta}=\sum_{( s x, s h,y)\in{\mathcal X}\times{\mathcal H}\times{\mathcal Y} }\exp\{-E( x, h,y, \theta)\}.
\end{align}
The marginal distribution is given by
\begin{align}\nonumber
p( x,y, \theta)=\frac{1}{Z_{ \theta}}
\sum_{  h\in{\mathcal H}}\exp\{-E( x,  h ,y, \theta)\},
\end{align}
Similarly, for a given dataset $\{( x_i,y_i)\}_{i=1,...,N}$, the GM-loss function for $ \theta$ is
defined by
 \begin{align}\nonumber
L_{\rm GM }( \theta) =\frac{1}{N}
\sum_{i=1}^N {\exp\{\tilde E( x_i,y_i, \theta)-\bar{E}( \theta)\}}.
\end{align}
where
 \begin{align}\nonumber
\tilde{E}( x,y, \theta)=- \log \sum_{  h\in{\mathcal H}}\exp\{-E( x,  h,y,   \theta)\}
\end{align}
and
 \begin{align}\nonumber
\bar{  E}( \theta)={\frac{1}{m'}}\sum_{(x,y)\in{\mathcal X}\times{\mathcal Y}}\tilde
E( x,y, \theta).
\end{align}
with the cardinal number $m'$ of ${\mathcal X}\times{\mathcal Y}$.
In accordance with this, we can apply the GM-loss function for the Boltzmann distribution
with supervised outcomes, and for that with the multiple hidden variables.  
In accordance with these, we point that the GM divergence and the GM estimator has advantageous property over the KL divergence and the ML-estimator in theoretical formulation.
A numerical example shows the advantage in a small scale of experiment. 
However, we have not discussed sufficient experiments and practical applications to confirm the advantage. 
For this, it needs further investigation for comparing the GM method with the current methods discussed for the deep brief network \cite{zhao2022intelligent}.  

\section{Multiclass AdaBoost}

AdaBoost is a part of ensemble learning algorithms that combine the decisions of multiple base learners, or weak learners to produce a strong learner. The core premise is that a group of "weak" models can be combined to form a "strong" model. 
AdaBoost \cite{freund1997decision} and its variants have found applications across various domains, including bioinformatics and statistical ecology, where they help in creating robust predictive models from noisy or incomplete data.
AdaBoost has been extended to handle multiclass classification problems. 

An example is Multiclass AdaBoost or AdaBoost.M1, an extension adapting the algorithm to handle more than two classes.
There are also other variants like SAMME (Stagewise Additive Modeling using a Multiclass Exponential loss function) which further extends AdaBoost to multiclass scenarios
\cite{hastie2009multi}. 
 Random forests and gradient boosting machines (GBM) can be mentioned as popular and efficient methods of ensemble learning \cite{breiman2001random,friedman2001greedy}.
Random forests exhibit a good balance between bias and variance, primarily due to the averaging of uncorrelated decision trees. 
GBMs are highly flexible and can be used with various loss functions and types of weak learners, though trees are standard.
AdaBoost excels in situations where bias reduction is crucial, while Random Forests are a robust, all-rounder solution. GBMs offer high flexibility and can achieve superior performance with careful tuning, especially in complex datasets.
Interestingly, there is a connection between AdaBoost and information geometry
 \cite{Murata2004}. 
The process of re-weighting the data points can be seen as a form of geometrically moving along the manifold of probability distributions over the data. 
This geometric interpretation might tie back to concepts of divergence and entropy, which are core to information geometry.
We focus on discussing various loss functions that are derived from the class of power divergence.









We discuss a setting of a binary class label.
Let $X$ be a covariate  with a value  in a subset ${\mathcal X}$ of $\mathbb R^d$, and  $Y$  be an outcome  with a value  in ${\mathcal Y}=\{-1,1\}$.
Let $f(x)$ be a predictor such that the prediction rule is given by $h(x)={\rm sign}(f(x))$. 
The exponential loss is proposed for the  Adaboost
\cite{freund1997decision,schapire2013boosting}.
As one of the most characteristic points, the optimization is conducted in the function space of a set of weak classifiers.
The exponential loss functional plays a central role as a key ingredient, which  is defined on  a space of  predictors as
\begin{align}\nonumber
L_{\exp }(f)=\frac{1}{n}\sum_{i=1}^n \exp\{-Y_if(X_i)\},
\end{align}
where $f$ is a predictor on ${\mathcal X}$.
If we take expectation under a conditional distribution
\begin{align}\nonumber
p_0(y|x)=\frac{\exp\{-yf_0(x)\}}{\exp\{f_0(x)\}+\exp\{-f_0(x)\}}, 
\end{align}
then the expected exponential loss is 
\begin{align}\nonumber
\mathbb E [L_{\exp }(f)]=\frac{ \exp\{f(x)-f_0(x)\}+\exp\{-f(x)+f_0(x)\}}{\exp\{f_0(x)\}+\exp\{-f_0(x)\}},
\end{align}
which is greater than or equal to ${ 2}/{\exp\{f_0(x)\}+\exp\{-f_0(x)\}}$.
The equality holds if and only if $f=f_0$.
This implies that the minimizer of the expected GM loss is equal to the true predictor $f_0$, namely, 
\begin{align}\nonumber
f_0=\argmin_{f\in{\cal F}}\mathbb E [L_{\exp}(f)].
\end{align}
The functional optimization is practically implemented by a simple algorithm. 
The stagewise learning algorithm is given as follows (Freund \& Schapire, 1995):

\noindent
\hrulefill

\noindent
(1).  Provide 
${\mathcal H}_J:=\{h_j:{\mathcal X}\rightarrow \{-1 ,1\}; j\in J\}$. Set as $w_{0,i}=\frac{1}{n}$ and $h_0(x)=0$.

\noindent
(2).  For step $t=1,...,T$\vspace{2mm}

(2.a).  $\displaystyle h_t=\argmin_{ h \in{\mathcal H}_J} {\rm Err}_t(h)$, where $\displaystyle{\rm Err}_t(h)=\sum_{i=1}^n w_{t-1,i}{\mathbb I}(h(X_i)\ne Y_i)$. \vspace{2mm}

(2.b).   $\displaystyle\alpha_t=\argmin_{\alpha\in\mathbb R}  L_{\rm GM }(f_{t-1}+\alpha h_t)$, where $\displaystyle f_{t-1}(x)=\sum_{j=1}^{t-1}\alpha_j h_j(x)$.\vspace{2mm}

(2.c).   $w_{t,i}= w_{t-1,i}\exp\{-\alpha_tY_i h_t(X_i)\}.$\vspace{4mm}

\noindent
(3).  Set $\displaystyle h_T(x)={\rm sgin}\Big(\sum_{t=1}^T \alpha_t h_t(x)\Big).$

\noindent
\hrulefill\vspace{4mm}

Note that substep (2.b) is calculated as a comprehensive form: the half log odds of error rate, 
\begin{align}\nonumber
\alpha_t=\half\log \frac{1-{\rm err}_t(h_t)}{{\rm err}_t(h_t)},
\end{align}
where ${\rm err}_t(h)={\rm Err}(h)/\sum_{i=1}^n w_{t-1,i}$.
The algorithm  is an  elegant and simple form, in which the mathematical operation is just defined by
elementary functions of $\exp$ and $\log$.
On the other hand, the iteratively reweighted least square algorithm needs the operation of matrix inverse
even for a linear logistic model. 
Let us apply the $\gamma$-loss functional for the boosting algorithm, cf. Chapert \ref{Minimum divergence for regression model} for a general form of the $\gamma$-loss.
First of all, confirm 
\begin{align}\nonumber
    p^{(\gamma)}(y|f(x))=\frac{\exp\{(\gamma+1)yf(x)\}}{\exp\{(\gamma+1)f(x)\}+\exp\{-(\gamma+1)f(x)\}}
\end{align}
as the $\gamma$-expression.    
Hence, the $\gamma$-loss functional is written by
\begin{align}\nonumber
  L_\gamma(f)=  \sum_{i=1}^n\big \{ p^{(\gamma)}(Y_i|f(X_i))\big\}^\frac{\gamma}{\gamma+1}.
\end{align}
Let us discuss the gradient boosting algorithm based on the $\gamma$-loss functional.
The stagewise learning algorithm $f_{t+1}=f_{t}+\alpha^* f^*$ for $t=0,1,...T$ is given as follows:
\begin{align}\nonumber
        (\alpha^*, f^*)=\argmin_{(\alpha,f)\in\mathbb R\times{\mathcal F}}
L_\gamma(f_{t}+\alpha f),
\end{align}
where $f_0$ is an initial guess and $T$ is determined by an appropriate stopping rule.
However, the joint minimization is expensive in the light of the computation. 
For this, we use the gradient as
\begin{align}\nonumber
\nabla L_\gamma(f_{t})=\partial_\alpha L_\gamma(f_{t}+\alpha f)\Big|_{\alpha=0},
\end{align}
which is written as
\begin{align}\nonumber
\sum_{i=1}^n \pi_{\gamma}(Y_i,f_t(X_i)){\mathbb I}(Y_i \ne {\rm sign}(f(X_i)))+C_{t},
\end{align}
where $C_t$ is a constant in $f$, and
\begin{align}\nonumber
 \pi_{\gamma}(Y_i,f_t(X_i))=
 \frac{p^{(\gamma)}(+1|f(X_i)) p^{(\gamma)}(-1|f(X_i))}{\big \{ p^{(\gamma)}(Y_i|f(X_i))\big\}^{\frac{1}{\gamma+1}}}.
\end{align}
In accordance, the $\gamma$-boosting algorithm is parallel to AdaBoost as:

\noindent
\hrulefill

\noindent
(1).  Provide 
${\mathcal H}_J:=\{h_j:{\mathcal X}\rightarrow \{-1 ,1\}; j\in J\}$. Set as $h_0(x)=0$.

\noindent
(2).  For step $t=1,...,T$\vspace{2mm}

(2.a).  $\displaystyle h_{t+1}=\argmin_{ h \in{\mathcal H}_J} {\rm Err}_t(h)$, where $\displaystyle{\rm Err}_t(h)=\sum_{i=1}^n \pi_\gamma(Y_i,f_t(X_i)){\mathbb I}(h(X_i)\ne Y_i)$. \vspace{2mm}

(2.b).   $\displaystyle\alpha_{t+1}=\argmin_{\alpha\in\mathbb R}  L_{\gamma }(f_{t}+\alpha h_{t+1})$, where $\displaystyle f_{t}(x)=\sum_{j=1}^{t}\alpha_j h_j(x)$.\vspace{2mm}

(2.c).    $\displaystyle{\rm Err}_{t+1}(h)=\sum_{i=1}^n \pi_\gamma(Y_i,f_t(X_i)+\alpha_{t+1}h_{t+1}(X_i)){\mathbb I}(h(X_i)\ne Y_i)$\vspace{4mm}

\noindent
(3).  Set $\displaystyle h_T(x)={\rm sgin}\Big(\sum_{t=1}^T \alpha_t h_t(x)\Big).$

\noindent
\hrulefill\vspace{4mm}

The-GM functional, $L_{\rm GM}(f)$, and the HM-loss functional, $L_{\rm HM}(f)$,  are derived by setting $\gamma$ to $-1$ and $-2$, respectively.
It can be observed that the exponential loss functional is nothing but the GM-loss functional.
We consider a situation of an imbalanced sample, where $\pi_{-1}\gg \pi_1$ for the probability $\pi_y$ of $Y=y$.
We adopt the adjusted exponential (GM) loss functional in \eqref{iw} as
\begin{align}\nonumber
L_{\exp }^{(\rm w)}(f)=\frac{1}{n}\sum_{i=1}^n \pi_{1-Y_i}\exp\{-\pi_{Y_i}Y_if(X_i)\}.
\end{align}
The learning algorithm is given by replacing  substeps (2.b) and (2.c) to\vspace{2mm}

(2.b$^*$).   $\displaystyle \alpha_t=\argmin_{\alpha\in\mathbb R}  L_{\exp }^{\rm ( w)}(f_{t-1}+\alpha h_t)$, where $\displaystyle f_{t-1}(x)=\sum_{j=1}^{t-1}\alpha_j h_j(x)$.\vspace{2mm}

(2.c$^*$).   $w_{t,i}= w_{t-1,i}\exp\{-\pi_{Y_i}Y_i\alpha_t h_t(X_i)\}.$\vspace{2mm}

\noindent
We observe in (2.b$^*$):
\begin{align}\nonumber
L_{\exp}^{(\rm w)}(f_{t-1}+\alpha h)\propto 
e^{\pi_{1}\alpha}{\rm err}_{t1}+e^{-\pi_{1}\alpha}(1-{\rm err}_{t1})+e^{\pi_{0}\alpha}{\rm err}_{t0}  +e^{-\pi_{0}\alpha}(1-{\rm err}_{t0}),
\end{align}
where 
\begin{align}\nonumber
{\rm err}_{ty}=\sum_{i=1}^n \pi_{1-y}w_{t-1,i}
\{e^{\pi_{1}\alpha} {\mathbb I}(Y_i=y,Y_i\ne f(X_i))/\sum_{i=1}^n w_{t-1,i}.
\end{align}

We discuss a setting of a multiclass label.
Let $X$ be a feature vector in a subset $\mathcal X$ of $\mathbb R^d$ and $Y$ be a label in ${\mathcal Y}=\{1,...,k\}$.
The major objective is to predict $Y$ given $X=x$, in which there are spaces of classifiers and
predictors, namely,  ${\cal H}=\{h:{\mathcal X}\rightarrow{\mathcal Y}\}$ and
\begin{align}\nonumber
  {\mathcal F}=\{f(x)=(f_1(x),...,f_k(x)) \in\mathbb R^k:  \sum_{y=1}^k f_y(x)=0\}.
\end{align}
A classifier $h(x)$ is  introduced by a predictor $f(x)$ as
\begin{align}\nonumber
h_f(x)=\argmax_{y\in{\mathcal Y}}f_y(x);
\end{align}
a predictor $f_h(x)$ is  introduced by a predictor $h(x)$ as
\begin{align}\label{embed}
 f_h(x)=\Big({\mathbb I}(h(x)=y)-\frac{1}{k}\Big)_{y\in{\mathcal Y}}.
\end{align}
Note that ${\mathcal H}=\{h_f:f\in{\mathcal F}\}$; while $\{f_h:h\in{\mathcal H}\}$ is a subset of $\mathcal F$.
In the learning algorithm discussed below, the predictor is updated by the linear span of predictors embedded  by selected classifiers in a sequential manner. 
The conditional probability mass function (pmf) of $Y$ given $X=x$ is assumed as a soft-max function
\begin{align}\nonumber
    p(y|f(x))=\frac{\exp\{f_y(x)\}}{\sum_{j\in{\mathcal Y}}\exp\{f_j(x)\}}
\end{align}
where $f(x)$ is a predictor of $\mathcal F$.
We notice that $p(y|f(x))$ and $f_y(x)$ are one-to-one as a function of $y$.
Indeed, they are connected as $f_y(x)=\log p(y|f(x))-\frac{1}{k}\sum_{j=1}^k \log p(y|f(x))$.
We note that this assumption is in the framework of the GLM as in the conditional pmf \eqref{76} with a different parametrization  discussed  in Section \ref{subsec-Multiclass} if $f (x)$ is a linear predictor.
However, the formulation is nonparametic, in which the model is written by ${\mathcal M}=\{
p(y|f(x)): f\in {\mathcal F}\}$.    
Similarly, the $\gamma$-loss functional for $f$ is 
\begin{align}\nonumber
  L_\gamma(f)=  \sum_{i=1}^n\Big \{\frac{\exp\{(\gamma+1)f_{Y_i}(X_i)\}}{\sum_{y \in{\mathcal Y}}\exp\{(\gamma+1)f_y (X_i)\}}\Big\}^\frac{\gamma}{\gamma+1}.
\end{align}
The minimum of the expected the $\gamma$-loss functional for $f$ 
is attained at $f=f^{(0)}$ taking expectation under a conditional distribution
\begin{align}\nonumber
p_0(y|x)=\frac{\exp\{f_y^{(0)}(x)\}}{\sum_{j=1}^k \exp\{f_j^{(0)}(x)\}}.
\end{align}
Thus, we conclude that the minimizer of the expected $\gamma$-loss in $f$ is equal to the true predictor $f^{(0)}$, namely, 
\begin{align}\nonumber
f^{(0)}=\argmin_{f\in{\cal F}}\mathbb E [L_{\gamma }(f)].
\end{align}
Thus, the minimization of the $\gamma$-loss functional on the predictor space $\cal F$ yields non-parametric consistency. 
Similarly, the stagewise learning algorithm $f_t=f_{t-1}+\alpha^* f_{h^*}$ for $t=1,...T$ is given as follows:
\begin{align}\nonumber
        (\alpha^*, h^*)=\argmin_{(\alpha,h)\in\mathbb R\times{\mathcal H}}
L_\gamma(f_{t-1}+\alpha f_h),
\end{align}
where $f_0$ is an initial guess and $f_h$ is defined in \eqref{embed}.
For any fixed $\alpha$, we observe
\begin{align}\nonumber
L_\gamma(f_{t-1}+\alpha f_h)
= \sum_{i=1}^n w_{t-1,i}{\mathbb I}(Y_i \ne h(X_i))+C_{t-1},
\end{align}
where 
\begin{align}\nonumber
w_{t-1,i}=\Big \{\frac{\exp\{(\gamma+1)f_{Y_i}(X_i)\}}{\sum_{y \in{\mathcal Y}}\exp\{(\gamma+1)f_y (X_i)\}}\Big\}^\frac{\gamma}{\gamma+1},
\end{align}

\noindent
\hrulefill

\noindent
(1).  Provide 
${\mathcal H}_J:=\{h_j:{\mathcal X}\rightarrow \{1 ,...,k\}; j\in J\}$. Set as $w_{0,i}=\frac{1}{n}$ and $h_0(x)=0$.

\noindent
(2).  For step $t=1,...,T$\vspace{2mm}

(2.a).  $\displaystyle h_t=\argmin_{ h \in{\mathcal H}_J} {\rm Err}_t(h)$, where $\displaystyle{\rm Err}_t(h)=\sum_{i=1}^n w_{t-1,i}{\mathbb I}(h(X_i)\ne Y_i)$.  \vspace{2mm}

(2.b).   $\displaystyle\alpha_t=\argmin_{\alpha\in\mathbb R}  L_{\gamma}(f_{t-1}+\alpha f_{h_t})$, where $\displaystyle f_{t-1}(x) =\sum_{j=1}^{t-1}\alpha_j f_{h_j}(x)$ with the

 \hspace{10mm}embedded predictor $f_h(x)$ defined in \eqref{embed}. \vspace{2mm}

(2.c).   $\displaystyle w_{t,i}= w_{t-1,i}\Big\{\frac{\exp\{(\gamma+1) \alpha_t f_t{}_{Y_i}(X_i)\}}{\sum_{y \in{\mathcal Y}}\exp\{(\gamma+1)f_t{}_y (X_i)\}}\Big\}^\frac{\gamma}{\gamma+1}.$\vspace{4mm}

\noindent
(3).  Set $\displaystyle h_T(x)=\argmax_{y\in{\mathcal Y}}f_T(y,x).$

\noindent
\hrulefill\vspace{4mm}

The GM-loss functional is given by
\begin{align}\nonumber
  L_{\rm GM}(f)=  \sum_{i=1}^n {\exp\{-f_{Y_i}(X_i)\}}
\end{align}
due to the normalizing condition $\sum_{y\in{\mathcal Y}}f_y(x)=0.$
This is essentially the same as the exponential loss \cite{hastie2009multi}, in which
the class label $y$ is coded as similar to \eqref{embed}.
Thus,  the equivalence of the GM-loss and the exponential loss also holds for a multiclass classification.
We can discuss the problem of imbalanced samples similarly as given for a binary classification.
Let $\pi_y=P(Y=y)$ and 
\begin{align}\nonumber
\pi_y^{\rm inv}=\frac{\frac{1}{\pi_y}}{\sum_{j=1}^k\frac{1}{\pi_j} }.
\end{align}
The adjusted exponential (GM) loss functional in \eqref{iw} as
\begin{align}\nonumber
L_{\exp }^{(\rm w)}(f)=\frac{1}{n}\sum_{i=1}^n \pi_{Y_i}^{\rm inv}\exp\{-\pi_{Y_i}f_{Y_i}(X_i)\}.
\end{align}
The learning algorithm is given by a minor change for substeps (2.b) and (2.c). 
The HM-loss functional is given by
\begin{align}\nonumber
  L_{\rm HM}(f)=   \sum_{i=1}^n\Big \{\frac{\exp\{-f_{Y_i}(X_i)\}}{\sum_{y \in{\mathcal Y}}\exp\{-f_y (X_i)\}}\Big\}^2.
\end{align}

GBMs  are highly flexible and can be adapted to various loss functions and types of weak learners, although decision trees are commonly used as the base learners. This flexibility is one of the key strengths of GBM, allowing it to be tailored to a wide range of problems and data types.
The loss functions discussed above can be applied to GBMs.
Its require careful tuning of several parameters (e.g., number of trees, learning rate, depth of trees), which can be time-consuming.
This discussion primarily focuses on the minimum divergence principle from a theoretical perspective.
In future projects, we aim to extend our discussion to develop effective GBM applications for a wide range of datasets.
     

\section{Active learning}

Active learning is a subfield of machine learning that focuses on building efficient training datasets, see \cite{settles2009active} for a comprehensive survey. 
Unlike traditional supervised learning, where all labels are provided upfront, active learning aims to select the most informative examples for labeling, thereby potentially reducing the number of labeled examples needed to achieve a certain level of performance, cf.
\cite{cohn1996active} for understanding how statistical methods are integrated into active learning algorithms. 
Active learning is a fascinating area where statistical machine learning and information geometry can intersect, offering deep insights into the learning process.
One of the primary goals is to reduce the number of labeled instances required to train a model effectively. Annotation can be expensive, especially for tasks like medical image labeling, natural language processing, or any domain-specific task requiring expert knowledge.
In scenarios where data collection is expensive or time-consuming, active learning aims to make the most out of a small dataset.
By focusing on ambiguous or difficult instances, active learning improves the model's performance faster than random sampling would.
In this way,  the active learning has been attracted attentions in a situation relevant in today's data-rich but label-scarce environments. This could set the stage for the technical details that follow.

The query by committee (QBC) method is a popular method in active learning in which there are another approaches the uncertainty sampling, the expected model change and Bayesian Optimization, see \cite{seung1992query} for the theoretical underpinnings of the QBC approach. 
We focus on the QBC approach, a "committee" of models is trained on the current labeled dataset. 
When it comes to selecting the next data point to label, the committee "votes" on the labels for the unlabeled data points. 
The data point for which there is the most disagreement among the committee members is then selected for labeling. 
The idea is that this point lies in a region of high uncertainty and therefore would provide the most information if labeled.
From an information geometry perspective, one could consider the divergence or distance between the probability distributions predicted by each model in the committee for a given data point. 
The point that maximizes this divergence could be considered the most informative.







Let $X$ be a feature vector in a subset $\mathcal X$ of $\mathbb R^d$ and $Y$ be a label in ${\mathcal Y}=\{1,...,k\}$.
The conditional probability mass function (pmf) of $Y$ given $X=x$ is assumed as a soft-max function
\begin{align}\nonumber
    p(y|\xi(x))=\frac{\exp\{\xi_y(x)\}}{\sum_{j\in{\mathcal Y}}\exp\{\xi_j(x)\}}
\end{align}
where $\xi(x)$ is a predictor vector with components $\{\xi_y(x)\}_{y=1}^k$ satisfying $\sum_{y\in{\mathcal Y}} \xi_y(x)=0.$
The prediction is conducted by 
\begin{align}\nonumber
h(x)=\argmin_{y\in{\mathcal Y}}\xi_y(x)
\end{align} 
noting $p(y|\xi(x))$ and $\xi_y(x)$ are one-to-one as a function of $y$.
In effect, they are connected as $\xi_y(x)=\log p(y|\xi(x))-\frac{1}{k}\sum_{j=1}^k \log p(y|x\xi(x))$.
We note that this assumption is in the framework of the GLM as in the conditional pmf \eqref{76} with a different parametrization  discussed  in Section \ref{subsec-Multiclass} if $\xi_y(x)$ is a linear predictor.

We aim to design a sequential family of datasets $\{S_t\}_{t=0}^T$ such that the $(t+1)$-th dataset is updated as 
\begin{align}\nonumber
S_{t+1}=\{(X_{t+1},Y_{t+1})\}\cup S_{t}
\end{align} 
for $t, 0\leq t\leq T-1$, where $S_0$ is an appropriately chosen datasets.
Given $S_t$, we conduct an experiment to get $(X_{t+1},Y_{t+1})$ in which  $X_{t+1}$ is explored to improve the performance of the prediction of the label $Y$, and the outcome $Y_{t+1}$ is sampled from the conditional distribution given $X_{t+1}$.
Thus, the  active leaning  proposes such a good update pair $(X_{t+1},Y_{t+1})$ that encourages the $t$-th
prediction result to strengthen the performance in a universal manner.
The key of the active leaning is to  build the efficient method to get the feature vector $X_{t+1}$ that compensates for the weakness of the prediction based on $S_t$. 
For this, it is preferable that the distribution of $Y$ given $X_{t+1}$ is separate from that given $(X_1,...,X_t)$.   
Here, let us take the QBC approach in which an acquisition function  plays a central role.

Assume that there are $m$ committee members or machines such that the $l$-th member employs a predictor $\xi^{(tl)}_y(x)$ for a feature vector $x$ and a label $y$ based on the dataset $S_t$, and thus the prediction for $Y$ given $x$ is performed by $\argmax_{y\in{\mathcal Y}}\xi^{(tl)}_y(x)$. 
We define an acquisition function defined on a feature vector $x$ of $\mathcal X$ 
\begin{align}\label{A-D}
    A^{(t)}(x)=\sum_{l=1}^m w_l D(P(\cdot|\xi^{(tl)}(x)), P(\cdot|\hat\xi^{(t)}(x)))
\end{align}
adopting a divergence measure $D$, where $\xi^{(tl)}(x)$ is the predictor vector 
learned by the $l$-th member at stage $t$; $\hat\xi^{(t)}(x)$ is the consensus predictor vector combining among $\{\xi^{(tl)}\}_{l=1}^m$.
The consensus predictor is given by
\begin{align}\nonumber
    \hat\xi_0^{(t)}(x)=\argmin_{\xi\in\Xi} \sum_{l=1}^m w_l D(P(\cdot|\xi^{(tl)}(x)), P(\cdot|\xi(x))),
\end{align}
where $\Xi$ is the set of all the predictor vectors.
Such an optimization problem is discussed around Proposition \ref{A-func} associated with the generalized mean  \cite{hino2023active}. 
In accordance, the new feature vector is selected as 
\begin{align}\label{new-vector}
X^{(t+1)}=\argmax_{x\in{\mathcal X}^{(t)}} A^{(t)}(x),
\end{align} 
where ${\mathcal X}^{(t)}$ is a subset of possible candidates of $\mathcal X$ at stage $t$.

The standard choice of $D$ for  \eqref{A-D} is the KL-divergence $D_0$ in \eqref{KL}, which yields the consensus distribution with the pmf
\begin{align}\nonumber
    \hat p_0^{(t)}(y|x) =\frac{\exp\big\{\sum_{l=1}^m w_l \xi^{(tl)}_{y}(x)\big\}}{\sum_{j=1}^k\exp\big\{\sum_{l=1}^m w_l \xi^{(tl)}_{j}(x)\big\}},
\end{align}
or equivalently $ \hat \xi_0^{(t)}  (x) =\sum_{l=1}^m w_l \xi^{(tl)}(x)$ as the consensus predictor.
Alternatively, we adopt the dual $\gamma$-divergence $D^*_\gamma$ defined in \eqref{gamma-star} and thus,
\begin{align}\nonumber
    A^{(t)}_\gamma(x)=\sum_{l=1}^m w_l D_\gamma(P(\cdot|\xi^{(tl)}(x)), P(\cdot|\hat\xi^{(t)}(x));C)
\end{align}
where
\begin{align}\nonumber
   D_\gamma(P(\cdot|\xi^{(tl)}(x)), P(\cdot|\hat\xi^{(t)}(x));C)= \sum_{y\in{\mathcal Y}}\{p^{(\gamma)}(y|\xi(x))\}^\frac{1}{\gamma+1}p(y|\xi^{(tl)}(x))^\gamma.
\end{align}
Here, $p^{(\gamma)}(y|\xi(x))$ is the $\gamma$-expression defined in \eqref{gamma-model}, or
\begin{align}\nonumber
    p^{(\gamma)}(y|\xi(x))=\frac{\exp\{(\gamma+1)\xi_y(x)\}}{\sum_{j=1}^k\exp\{(\gamma+1)\xi_j(x)\}}.
\end{align}
This yields 
\begin{align}\nonumber
    \hat p_\gamma^{(t)}(y|x) =\frac{\Big[\sum_{l=1}^m w_l \exp\{\gamma\xi^{(tl)}_{y}(x)\}\Big]^\frac{1}{\gamma}}
    {\sum_{j=1}^k \Big[\sum_{l=1}^m w_l \exp\{\gamma\xi^{(tl)}_{j}(x)\}\Big]^\frac{1}{\gamma}}.
\end{align}
as  the pmf of the consensus distribution and 
\begin{align}\nonumber
    \hat \xi_\gamma{}_{\ y}^{(t)}{}(x) =\frac{1}{\gamma} \log \Big[\sum_{l=1}^m w_l \exp\{\gamma\xi_y^{(tl)}(x)\}\Big]
\end{align}
as the consensus predictor up to a constant in $y$.
We note that the consensus predictor $\hat \xi_\gamma^{(t)}{}_{\!\!y}(x)$ has a form of  log-sum-exp mean.
This has the extreme forms as
\begin{align}\nonumber
  \lim_{\gamma\rightarrow-\infty}  \hat \xi_\gamma^{(t)} (x) =\min_{1\leq l\leq m} \xi^{(tl)} (x) \quad \text{      and  } \quad
  \lim_{\gamma\rightarrow\infty}  \hat \xi_\gamma^{(t)} (x) =\max_{1\leq l\leq m} \xi^{(tl)} (x).
\end{align}
Let us look at the decision boundary of the consensus predictors $\hat\xi_0(x)$ and $\hat\xi_\gamma(x)$ combining two linear predictors in a two dimensional space, see Figure \ref{DB-fig}.

\begin{figure}[htbp]
\begin{center}
 \includegraphics[width=100mm]{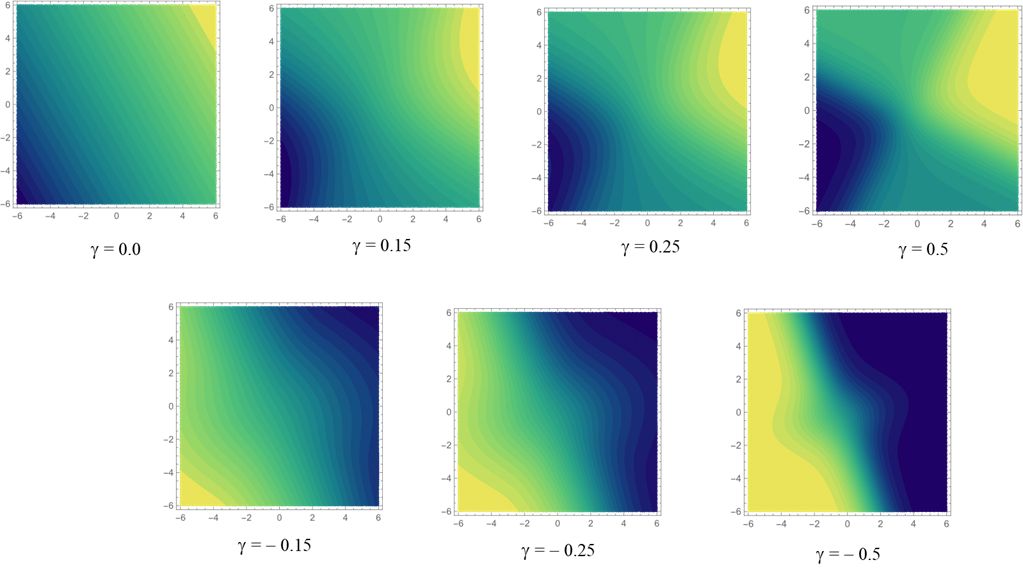}
 \end{center}
 \vspace{-5mm}\caption{Plots of decision boundaries of dual KL, dual $\gamma$-divergence measures}
\label{DB-fig}.
\end{figure} 

If all the committee machines have linear predictors, 
then the $0$ consensus predictor is still a linear predictor, but the $\gamma$-consensus predictor is a nonlinear predictor 
according to the value of $\gamma\ne0$ as in Figure \ref{DB-fig}.
Hence, we can explore the nonlinearity at every stage learning an appropriate value of $\gamma$. 
Needless to say, the objective is to find a good feature vector $X_{t+1}$ in \eqref{new-vector}, and hence we have to pay attentions to
the leaning procedure conducted by a minimax process
\begin{align}\nonumber
    \max_{x\in{\mathcal X}} \min_{\xi\in\Xi} \Big\{ \sum_{l=1}^m w_l D_\gamma(P(\cdot|\xi^{(tl)}(x)), P(\cdot| \xi (x)))\Big\}.
\end{align}
It is possible to monitor the minimax value each stage, which can evaluate the learning performance.
In effect, the minimax game of the cross entropy with Nature and a decision maker is nicely discussed \cite{grunwald2004game}. 
The mini-maxity is solved in the zero sum game: Nature wants to maximize the cross entropy under a constrain with a fixed expectation; the decision maker wants to minimize one o the full space. 
However, our minimax process is not relevant to this observation straightforward.
It is necessary to make further discussion to propose a selection of the optimal value of $\gamma$ based on $S_t$.

\



\section{The $\gamma$-cosine similarity}


The $\gamma$-divergence is defined on the probability measures dominated by a reference measure.
We have studied statistical applications focusing on regression and classification.
We would like to extend that the $\gamma$-divergence can be defined on the Lebesgue ${\rm L}_p$-space,
${\cal L}_p(\Lambda)=\{f(x): \|f\|_p<\infty\}$ for an exponent $p$, $1\leq p\leq\infty$, where 
the ${\rm L}_p$-norm is defined by
\begin{align}
 \|f\|_p=\Big(\int |f(x)|^p {\rm d}\Lambda(x)\Big)^\frac{1}{p},
\end{align}
where $\Lambda$ is a $\sigma$-finite measure.
There is a challenge for the extension: a function $f(x)$ can take a negative value.
If we adopt a usual power transformation $f(x)^{\gamma}$, it indeed poses a problem.
When $f(x)<0$ as raising a negative number to a fractional power can lead to complex values, which would not be meaningful in this context.
For this,  we introduce a sign-preserved power transformation as
\begin{align}\nonumber
f(x)^{\ominus\gamma}={\rm sign}(f(x))|f(x)|^\gamma.
\end{align}
The log $\gamma$-divergence \eqref{Gamma-log} is extended as
\begin{align}\nonumber
\Delta_\gamma(f,g;\Lambda)=-\frac{1}{\gamma}
\log \bigg|\int \frac{f (x)}{\|f\|_p}\Big(\frac{ g(x)}{\|g\|_p}\Big)^{\ominus\gamma} {\rm d}\Lambda(x)\bigg|
\end{align}
for $f$ and $g$ of ${\cal L}_p(\Lambda)$, where $p=\gamma+1$.
This still satisfies the scale invariance.
There would be potential developments to utilize $\Delta(f,g;\Lambda)$ in a field of statistics and machine learning, by which singed functions like a predictor function or functional data can be directed evaluated.
In particular, we explore this idea to a context of cosine similarity.








Cosine similarity is a measure used to determine the similarity between two non-zero vectors in an inner product space, which includes Hilbert spaces. This measure is particularly important in many applications, such as information retrieval, text analysis, and pattern recognition.
In a Hilbert space, which is a complete inner product space, cosine similarity can be defined in a way that generalizes the concept from Euclidean spaces:
\begin{align}\label{H-cosine}
   \cos(f,g)=\Big\langle \frac{f}{\|f\|},\frac{g}{\|g\|}\Big\rangle,
\end{align}
where $\|f \|=\sqrt{\langle f,f\rangle}$ and $\langle\ ,\ \rangle$ is the inner product, namely, $\langle f,g\rangle=\int f(x)g(x){\rm d}\Lambda(x)$.
Thus, in the Hilbert space, ${\cal H}(\Lambda)$, $\cos(\tau f,\sigma g)=\cos(f,g)$ for any scalars $\tau$ and $\sigma$.
The Cauchy-Schwartz inequality yields $|\cos(f,g)|\leq1$, and $|\cos(f,g)|=1$ if and only if there exists a scalar $\sigma$ such that $g(x)=\sigma f(x)$ for $x$ almost everywhere.

We extend the cosine measure \eqref{H-cosine}  on the L$_p$-space by analogy with the extension of the log $\gamma$-divergence, see \cite{luenberger1997optimization,conway2019course} for foundations for function analysis.
For this, we observe the H\"{o}lder inequality  implies $|{\rm H}(u,v)|\leq1$, where
\begin{align}\label{H-func}
   {\rm H}(u,v)=\Big\langle \frac{u}{\|u\|_p},\frac{v}{\|v\|_q}\Big\rangle,
\end{align}
for $u\in{\cal L}_p$ and $v\in{\cal L}_q$, where $q$ is the conjugate exponent to $p$ satisfying
$\frac{1}{p}+\frac{1}{q}=1$.
The dual space (the Banach space of all continuous linear functionals) of the
L$_{p}$-space for $ 1< p < \infty$ has a natural isomorphism with 
L$_{q}$-space.
The isomorphism associates with the functional $\iota_p(v)\in{\cal L}_p(\Lambda)^*$ defined by
$u\mapsto \iota_p(v)(u)=\int uv {\rm d}\Lambda$.
Thus, the the H\"{o}lder inequality guarantees that $\iota_p(v)(u)$ is well defined and continuous, and hence ${\cal L}_q(\Lambda)$ is said to be the continuous dual space of ${\cal L}_p(\Lambda)$.
Apparently, ${\rm H}(u,v)$ seems a surrogate for $\cos(f,g)$.
However, the domain of $\cos$ is
${\cal H}(\Lambda)\times{\cal H}(\Lambda)$; that of $\rm H$ is ${\cal L}_p(\Lambda)\times{\cal L}_q(\Lambda)$.
Further,
$|\cos(f,g)|=1$ means $f\propto g$; $|{\rm H}(u,v)|=1$ means $|u|^p\propto |v|^q$.
Thus, the functional ${\rm H}(u,v)$ has inappropriate characters as a cosine functional to measure an angle between vectors in a function space.
For this, consider a transform  $\kappa_p$ from ${\cal L}_p$ to ${\cal L}_q$  by 
$ \kappa_p(v)= v^{\ominus \frac{p}{q}}$ noting $|\kappa_p(v)|^q=|v|^p$.
Then, we can define ${\rm H}(u,\kappa_p(v))$ for $u$ and $v$ in ${\cal L}_p(\Lambda).$
Consequently, we define a cosine measure defined on ${\cal L}_p$ as
\begin{align}\label{cos-g}
\cos_\gamma(f,g)= \Big\langle\frac{f\ }{\|f\|_p }, \Big(\frac{g \ }{\|g\|_p}\Big)^{\ominus\frac{p}{q}}\Big\rangle 
\end{align}
linking $\gamma$ to as  $p=\gamma+1$, where $q$ is the conjugate exponent to $p$.
A close connection with the log $\gamma$-divergence  is noted as 
\begin{align}\nonumber
|\cos_\gamma(f,g)|=\exp\{-\gamma \Delta_\gamma(f,g)\} .
\end{align}
This implies $\cos_\gamma(f,g)=0\ \Longleftrightarrow \  \Delta_\gamma(f,g)=\infty$, in which both express  quantities when $f$ and $g$ are the most distinct.
In this formulation, $\cos_\gamma(f,g)$, called the $\gamma$-cosine,  ensures mathematical consistency across all real values of 
$g(x)$, which is vital for the measure's applicability in a wide range of contexts.
Note that, if $p=2$, then $\cos_\gamma(f,g)=\cos(f,g)$, in which $g(x)^{\ominus\frac{p}{q}}$ reduces to $g(x)$.
Further, a basic property is summarized as follows:

\begin{proposition}
Let $f$ and $g$ be in ${\cal L}_p(\Lambda)$.
Then, $|\cos_\gamma(f,g)|\leq 1$, and equality holds if and only $g$ is proportional to
 $f$.
\end{proposition}
\begin{proof}
By definition, $\cos_\gamma(f,g)={\rm H}(f,\kappa_p(g))$, where ${\rm H}$ is defined in \eqref{H-func}. This implies $|\cos_\gamma(f,g)\leq1$.  The equality  holds if and only if
$|f|^p\propto|g^q|^{\frac{p}{q}}$, that is,  there exists a scalar $\sigma$
such that $g(x)=\sigma f(x)$ for everywhere $x$.
\end{proof}
In this way, the $\gamma$-cosine is defined by the isomorphism between ${\cal L}_p$ and ${\cal L}_p^*$.
We note that 
\begin{align}\nonumber
\cos_\gamma(\tau f,\sigma g)={\rm sign}(\tau\sigma)\cos_\gamma(f,g).
\end{align}
Accordingly, $\cos_\gamma(f,g)$ is a natural extension of the cosine functional $\cos(f,g)$.
As a special character, $\cos_\gamma(f,g)$ is asymmetric in $f$ and $g$ if $\gamma\ne1$.
The asymmetry remains, akin to divergence measures, providing a directional similarity measure between two functions.
We have discussed the cosine measure extended on ${\cal L}_{1+\gamma}(\Lambda)$
relating to the log $\gamma$-divergence.
In effect, the divergence is defined for the applicability of any empirical probability measure for a given dataset.
However, such a constraint is not required in this context.
Hence we can define a generalized variants
\begin{align}\label{cos-gen}
\cos_{(\beta,\gamma)}(f,g)= \Big\langle\Big(\frac{f\ }{\|f\|_{\beta+\gamma} }\Big)^{\ominus\beta}, \Big(\frac{g \ }{\|g\|_{\beta+\gamma}}\Big)^{\ominus{\gamma}}\Big\rangle 
\end{align}
for $f$ and $g$ in ${\cal L}_{\beta+\gamma}(\Lambda)$, called the $(\beta,\gamma)$-cosine measure,  with tuning parameters $\beta\geq1$ and $\gamma\geq1.$
Specifically, it is noted $\cos_{\gamma}(f,g)=\cos_{(1,\gamma)}(f,g)$.
We note that the information divergence associated with $\cos_{(\beta,\gamma)}(f,g)$ is given by
\begin{align}\nonumber
\Delta_{(\beta,\gamma)}(f,g)=-\frac{1}{\beta\gamma}\log |\cos_{(\beta,\gamma)}(f,g)|.
\end{align}
In statistical machine learning, this measure could be used to compare probability density functions, regression functions, or other functional forms, especially when dealing with asymmetric relationships.
It might be particularly relevant in scenarios where the sign of the function values carries important information, such as in economic data, signal processing, or environmental modeling.

The formulation defined on the function space is easily reduced on a Euclidean space as follows.
Let $x$ and $y$ be in $\mathbb R^d$.
Then, the cosine similarity is defined by
\begin{align}\label{cos}
\cos(x,y)=\Big\langle \frac{x}{\|x\|},\frac{y}{\|y\|}\Big\rangle=\langle e(x),e(y)\rangle, 
\end{align}
where $e(x)=x/\| x\|$ and $\langle\cdot,\cdot\rangle$ and $\|\cdot\|$ denote  the Euclidean inner product and norm on $\mathbb R^d$.
The $\gamma$-cosine function in $\mathbb R^d$ is introduced as
\begin{align}\nonumber
\cos_\gamma(x,y)&= \Big\langle \frac{x}{\|x\|_{p}},\Big(\frac{y}{\|y\|_{p}}\Big)^{\ominus\gamma}\Big\rangle=\langle e_\gamma(x),e_\gamma^*(y)\rangle, 
\\[3mm]
&= \frac{\sum_{i=1}^d x_i {\rm sign}(y_i)|y_i|^{\gamma}}
{\{\sum_{i=1}^d |x_i|^{p}\}^{\frac{1}{p}}\{\sum_{i=1}^d |y_i|^{p}\}^{\frac{1}{q}}},
\end{align}
for a power parameter $\gamma>0$.
We can view the plot of the sign-preserving power transformation $x^{\ominus\gamma}$ for $\gamma=\frac{1}{5},\frac{2}{5},\frac{3}{5},\frac{4}{5},1$ in Fig. \ref{sign-power}:

\

\begin{figure}[htbp]
\begin{center}
  \includegraphics[width=70mm]{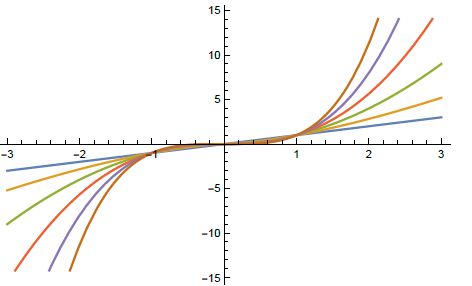}
 \end{center}
 \vspace{-1mm}\caption{Plots of the signed power function.}
\label{sign-power}
\end{figure} 

As for the generalized measure, the $(\beta,\gamma)$-cosine measure  is given by
\begin{align}\nonumber
\cos_{(\beta,\gamma)}(x,y)= \Big\langle \Big(\frac{x}{\|x\|_{\beta+\gamma}}\Big)^{\ominus\beta},\Big(\frac{y }{ \| y\|_{\beta+\gamma}}\Big)^{\ominus\gamma}\Big\rangle,
\end{align} 
see the functional form \eqref{cos-g}.

We investigate properties of the $\gamma$-cosine and $(\beta,\gamma)$ cosine comparing to the standard cosine.
Let $\displaystyle e_{(\beta,\gamma)}(x)=\Big(\frac{x}{\|x\|_{\beta+\gamma}}\Big)^{\ominus\beta}$ and
$\displaystyle e_{(\beta,\gamma)}^*(y)=\Big(\frac{y}{\|y\|_{\beta+\gamma}}\Big)^{\ominus\gamma}$.
Then, the $(\beta,\gamma)$-cosine is written as
\begin{align}\nonumber
\cos_{(\beta,\gamma)}(x,y)= \langle e_{(\beta,\gamma)}(x),e_{(\beta,\gamma)}^*(y)\rangle,
\end{align} 
We observe the following behaviors where $\gamma$ has an extreme value.

\begin{proposition}\label{prop1}
Let $x$ and $y$ be in $\mathbb R^d$.
Then,\\[2mm]
{\rm (a)}.\hspace{15mm} $\displaystyle \lim_{\gamma\rightarrow0}\cos_{(\beta,\gamma)}(x,y) =\Big\langle \Big(\frac{x}{\|x\|_\beta}\Big)^{\ominus\beta},{\rm sign}(y)\Big\rangle,$\\[2mm]
where ${\rm sign}( y)=({\rm sign} (y_i))_{i=1}^d$. Further,\\[2mm]
{\rm (b)}.\hspace{15mm} $\displaystyle \lim_{\gamma\rightarrow\infty}\cos_{(\beta,\gamma)}(x,y) =\Big\langle\Big( \frac{x}{\|x\|_\infty}\Big)^{\ominus\beta},\frac{{\rm sign}_\infty(y)}{\|{\rm sign}_\infty(y)\|_1}\Big\rangle,$\\[2mm]
where  the $i$-th component of ${\rm sign}_{\infty}(y)$ denotes ${\rm sign}(y_i){\mathbb I}(|y_i|=\|y\|_\infty)$ for $i=1...,d$ with $\|y\|_{\infty}=\max_{1\leq i\leq d}|y_i|$.

\end{proposition}
\begin{proof}
By definition, $\lim_{\gamma\rightarrow0}y^{\ominus\gamma}={\rm sign}(y)$.
This implies (a).
Next, if we divide both the numerator and the denominator of $e_{(\beta,\gamma)}^*(y)$ by $\|y\|_\infty$, then
\begin{align}\nonumber
 \lim_{\gamma\rightarrow\infty} e_{(\beta,\gamma)}^*(y)= 
\lim_{\gamma\rightarrow\infty}\Big(\frac{y}{\|y\|_\infty}\Big)^{\ominus\gamma}
\Big(\frac{\|y\|_{\beta+\gamma}}{\|y\|_\infty}\Big)^{-\gamma}.
\end{align}
Hence, for $i=1,...,d$
\begin{align}\nonumber
{\rm sign}(y_i)\lim_{\gamma\rightarrow\infty}\Big(\frac{|y_i|}{\|y\|_\infty}\Big)^{\gamma}
={\rm sign}(y_i){\mathbb I}(|y_i|=\|y\|_\infty);
\end{align}
\begin{align}\nonumber 
\lim_{\gamma\rightarrow\infty}\Big(\frac{\|y\|_{\beta+\gamma}}{\|y\|_\infty}\Big)^{-\gamma}=
\lim_{\gamma\rightarrow\infty}\bigg[\sum_{i=1}\Big(\frac{y_i}{\|y\|_\infty}\Big)^\gamma\bigg]^{-1}=\Big[\sum_{i=1}{\mathbb I}(|{y_i}|={\|y\|_\infty})\Big]^{-1}.
\end{align}
Consequently, we conclude (b).
\end{proof}
We remark that 
\begin{align}\nonumber 
 \lim_{\beta\rightarrow0,\gamma\rightarrow0}\cos_{(\beta,\gamma)}(x,y) =\frac{1}{d}\big\langle {\rm sign}(x),{\rm sign}(y)\big\rangle
\end{align}
Alternatively, the order of taking limits of $\beta$ and $\gamma$ to $\infty$ with respect to $\cos_{(\beta,\gamma)}(x,y)$ results in different outcomes:
\begin{align}\nonumber
 \lim_{\beta\rightarrow\infty} \lim_{\gamma\rightarrow\infty}\cos_{(\beta,\gamma)}(x,y) =\frac{\langle{\rm sign}_\infty(x),{\rm sign}_\infty(y)\rangle}{\|{\rm sign}_\infty(y)\|_1} ;
\end{align}
\begin{align}\nonumber
 \lim_{\gamma\rightarrow\infty} \lim_{\beta\rightarrow\infty}\cos_{(\beta,\gamma)}(x,y) =\frac{\langle{\rm sign}_\infty(x),{\rm sign}_\infty(y)\rangle}{\|{\rm sign}_\infty(x)\|_1} ;
\end{align}
\begin{align}\nonumber
 \lim_{\gamma\rightarrow\infty} \cos_{(\gamma,\gamma)}(x,y) =\frac{\langle{\rm sign}_\infty(x),{\rm sign}_\infty(y)\rangle}{\|{\rm sign}_\infty(x)\|_{2\ }\|{\rm sign}_\infty(y)\|_2} .
\end{align}

Note that  ${\rm sign}_\infty(y)$ is a sparse vector as it has $\rm sign$ only at the components with the maximum absolute value with 0's elsewhere .
Thus, $\cos_{(\beta,\infty)}(x,y)$ is proportional to the Euclidean inner product between $x^{\ominus\beta}$ and the sparse vector
${\rm sign}_\infty(y)$.
This is contrast with the standard cosine similarity, in which the orthogonality with $\cos_\infty(x,y)$ is totally different from that with $\cos(x,y)$.
In effect, $\cos(x,y)=0\Leftrightarrow\langle x,y\rangle=0$; 
$\cos_{(\beta,\infty)}(x,y)=0\Leftrightarrow\langle x^{\ominus\beta},{\rm sign}_\infty(y)\rangle=0$.
The orthogonality with $\cos_{(\beta,\infty)}(x,y)$ is reduced to the inner product of the
$d_\infty$-dimensional Euclidean space, where $d_\infty$ is 
the cardinal number of $\{i\in\{1,...,d\}: |y_i|=\|y\|_\infty\}$.
Note that the equality condition  in the limit case of $\gamma$
is totally different from that when $\gamma$ is finite. 
Indeed, $x=\pm{\rm sign}_\infty(y)$ if and only if $\cos_{(\beta,\infty)}(x,y)=\pm1$, where  $\cos_\infty(x,y)$ can be viewed the arithmetic mean of relative ratios in $1_\infty(y)$. 
 It is pointed that the cosine similarity has poor performance in a high-dimensional data.
Then, values of the cosine similarity becomes small numbers near a zero, and hence they cannot extract important characteristics of vectors.
It is frequently observed in the case of  high-dimensional data that only a small part of components involves important information for a target analysis; the remaining components are non-informative. 
The standard cosine similarity equally measures all components; while
the power-transformed cosine  ($\gamma$-cos) can focus on only the small part of essential components.
Thus,  the $\gamma$-cos neglects unnecessary information with the majority components, so that  the $\gamma$-cos can extract essential information involving with principal components.
In this sense, the $\gamma$-cos does not need any preprocessing procedures for dimension reduction such as principal component analysis.

\begin{proposition}\label{prop2}
Let $x=(x_0,x_1)$ and $y=(y_0,y_1)$, respectively, where
$x_0, y_0\in\mathbb R^{d_0}$; $x_1, y_1\in \mathbb R^{d_1}$
with $d=d_0+d_1$.
If $\|x_0\|_\infty>\|x_1\|_\infty$ and $\|y_0\|_\infty>\|y_1\|_\infty$, then,
\begin{align}\label{result2}
 \cos_{(\beta,\infty)}(x_0,y_0)=\cos_{(\beta,\infty)}(x,y).
\end{align}
\end{proposition}
\begin{proof}
From the assumption, $\|x_0\|_\infty=\|x\|_\infty$ and $1_\infty(y_0)=1_\infty(y)$.
This implies
\begin{align}\label{sum}
\cos_{(\beta,\infty)}(x,y) = \sum_{i=1}^{d} 
 \frac{\ x_{i}^{\ominus\beta}\ \ }{\|x_0\|_{\infty}}\frac{{\rm sign}(y_i){\mathbb I}(|y_i| =\|y_0\|_{\infty}\big)}{|1_\infty(y_0)|},
\end{align}
which is nothing but $\cos_{(\beta,\infty)}(x_0,y_0)$ since
all the summands are zeros  in the summation of $i$ from $d_0+1$ to $d$ in \eqref{sum}. 
\end{proof}
In Proposition \ref{prop2}, the infinite-power cosine similarity is viewed as a robust measure in the sense that  $ \cos_\infty(x_0,y_0)=\cos_\infty((x_0,x_1),(y_0,y_1))$ for any
minor components $x_1$ and $y_1$.
However, we observe that the robustness looks extreme as seen in the following. 

\begin{proposition}\label{prop3}
Consider a function of  $\epsilon$ as
\begin{align}\nonumber
\Phi(\epsilon)=\cos_\infty(x,(y_0,\epsilon y_1)).
\end{align}
Then, if $\|y_1\|_{\infty}=\|y_0\|_{\infty}$, $\Phi(\epsilon)$ is not continuous at $\epsilon=1$.

\end{proposition}
\begin{proof}
It follows from Proposition 2 that, if $0< \epsilon<1$, then  
\begin{align}\nonumber
\Phi(\epsilon) = \sum_{i=1}^{d_0} 
 \frac{x_{i}\ \ }{\|x\|_{\infty}}\frac{{\mathbb I}(|y_i| =\|y_0\|_{\infty}\big)}{|1_\infty(y_0)|}
\end{align}
where $d_0$ is the dimension of $y_0$. 
On the other hand, 
\begin{align}\nonumber
\Phi(1) =  \sum_{i=1}^{d_0} 
 \frac{x_{i}\ \ }{\|x\|_{\infty}}\frac{{\mathbb I}(|y_i| =\|y_0\|_{\infty}\big)}{|1_\infty(y_0)|}
+\sum_{i=d_0+1}^{d} 
 \frac{x_{i}\ \ }{\|x\|_{\infty}}\frac{{\mathbb I}(|y_i| =\|y_1\|_{\infty}\big)}{|1_\infty(y_1)|}.
\end{align}
This implies the discontinuity of $\Phi(\epsilon)$ at $\epsilon =1$.

\end{proof}

We investigate statistical properties of the power cosine measure in comparison with the conventional cosine similarity.
For this we have scenarios to generate realized vectors $x$'s and $y$'s in $\mathbb R^d$ as follows.
Assume that the $j$-th replications $X_j$ and $Y_j$ are given by
\begin{align}\nonumber
X_j=\mu_1+\epsilon_1 \hspace{5mm}\mbox{and}\hspace{5mm}Y_j=\mu_2+\epsilon_2,
\end{align}
where $\epsilon_a$'s are identically and independently distributed as ${\tt Nor}(0,\sigma^2{\mathbb I}_d)$.
We conduct a numerical experiment with  $2000$   replications setting $d=1000$ and $\mu_1=(10,9,...,1,0,...,0)^\top$ with
$\mu_2$ fixed later for some $\sigma^2$'s.

First, fix as $\mu_2=\mu_1$ as a proportional case.
Then, the value of the cosine measure $\cos_{(\beta,\gamma)}(X,Y)$ is expected to be $1$ if the error terms are negligible.
When $(\beta,\gamma)=(1,1)$, then $\cos_{(\beta,\gamma)}(X,Y)$ has not a consistent mean even with small erros; when
$\beta>1 ,\gamma>1$, then $\cos_{(\beta,\gamma)}(X,Y)$ has a consistent mean near $0$ with resonable errors.
Table \ref{propto} shows  detailed outcomes with the variance $\sigma^2=0.05,0.1,0.3,0.5$, where Mean and Std denote the mean and standard deviations for $\cos_{(\beta,\gamma)}(X_j.Y_j)$'s with 2000 replications .
Second, fix as 
\begin{align}\nonumber
\mu_2=\mu_{20}-\frac{\langle \mu_1,\mu_{20}\rangle}{\|\mu_0\|^2}\mu_0,
\end{align} 
where $\mu_{20}=(1,2,...,10,0,...,0)^\top$. Note $\langle \mu_1,\mu_{20}\rangle=0$.
This means $\mu_1$ and $\mu_2$ are orthogonal in the L$_2$-sence.
Then, the value of the cosine measure $\cos_{(\beta,\gamma)}(X,Y)$ should be near $0$ if the error terms are negligible.
For all the cases $(\beta,\gamma)$'s, the mean of $\cos_{(\beta,\gamma)}(X,Y)$ is reasonably near $0$ with small standard deviations, see Table \ref{ortho} for details.


 \begin{table}[ht]
\caption{ $\cos_{(\beta,\gamma)}(X,Y)$ in a proportional case } \label{propto}
 
\begin{center}
  \begin{tabular}{ccc}
  \hline 

    $(\beta,\gamma)$  & Mean  &  Std
 \\
    \hline \hline

    $(1,1)$ & $ 0.885 $  & $  0.005$\\
    $(2,2)$   & $ 0.997$  & $ 0.001$ \\
    $(2,5)$   & $ 0.995 $  & $ 0.003$ \\
   \hline

  \end{tabular}

\

       {$\sigma^2=0.05$}

\

 \begin{tabular}{ccc}
  \hline 

    $(\beta,\gamma)$  & Mean  &  Std
 \\
    \hline \hline

     $(1,1)$ & $0.793$  & $0.008 $\\
    $(2,2)$   & $0.993  $  & $0.003$ \\
    $(2,5)$   & $ 0.991$  & $0.007$ \\
   \hline

  \end{tabular}

\

     $\sigma^2=0.1$
  
\

  \begin{tabular}{ccc}
  \hline 

    $(\beta,\gamma)$  & Mean  &  Std
 \\
    \hline \hline

    $(1,1)$ & $ 0.562$  & $ 0.018$\\
    $(2,2)$   & $ 0.975$  & $0.008$ \\
    $(2,5)$   & $0.976$  & $0.018$ \\
   \hline
  \end{tabular}

\

           {$\sigma^2=0.3$}
  
\
      
 \begin{tabular}{ccc}
  \hline 

    $(\beta,\gamma)$  & Mean  &  Std
 \\
    \hline \hline

    $(1,1)$ & $0.434  $  & $0.023$\\
    $(2,2)$   & $0.948$  & $0.014$ \\
    $(2,5)$   & $ 0.961 $  & $0.030$ \\
   \hline
  \end{tabular}

\

           {$\sigma^2=0.5$}

 \end{center}
 
   \end{table}


\begin{table}[ht]
\caption{ $\cos_{(\beta,\gamma)}(X,Y)$ in an orthogonal case }\label{ortho} 

      \begin{center}
  \begin{tabular}{ccc}
  \hline 

    $(\beta,\gamma)$  & Mean  &  Std
 \\
    \hline \hline

    $(1,1)$ & $0.000    $  & $ 0.029$\\
    $(2,2)$   & $ -0.086 $  & $0.045$ \\
    $(2,5)$   & $0.007 $  & $0.000$ \\
   \hline
  \end{tabular}

\

           {$\sigma^2=0.05$}

\

 \begin{tabular}{ccc}
  \hline 

    $(\beta,\gamma)$  & Mean  &  Std
 \\
    \hline \hline

     $(1,1)$ & $0.000$  & $0.015 $\\
    $(2,2)$   & $0.092  $  & $0.014$ \\
    $(2,5)$   & $ 0.006$  & $0.000$ \\
   \hline
  \end{tabular}

\

           {$\sigma^2=0.1$}
        
\

  \begin{tabular}{ccc}
  \hline 

    $(\beta,\gamma)$  & Mean  &  Std
 \\
    \hline \hline

    $(1,1)$ & $ 0.000$  & $ 0.020$\\
    $(2,2)$   & $ 0.093$  & $0.021$ \\
    $(2,5)$   & $0.006$  & $0.000$ \\
   \hline
  \end{tabular}

\

          {$\sigma^2=0.3$}
 
          \

 \begin{tabular}{ccc}
  \hline 

    $(\beta,\gamma)$  & Mean  &  Std
 \\
    \hline \hline

    $(1,1)$ & $0.000 $  & $0.027$\\
    $(2,2)$   & $-0.089 $  & $0.036$ \\
    $(2,5)$   & $0.006 $  & $0.000$ \\
   \hline
  \end{tabular}

\

         {$\sigma^2=0.5$}

\end{center}
   \end{table}

We applied these similarity measures to hierarchical clustering using a Python package. Synthetic data were generated in a setting of 8 clusters, each with 15 data points, in a 1000-dimensional Euclidean space. The distance functions used were $\cos_{(1,1)}(x,y)$ and $\cos_{(1,5)}(x,y)$, to compare performance in high-dimensional data clustering. The clustering criterion was set to {\tt maxclust} in {\tt fcluster} from the {\tt scipy.cluster.hierarchy} module. The silhouette score, ranging from -1 to +1, served as a measure of the clustering quality. The clustering was conducted with 10 replications.

For case (a), using the distance based on $\cos_{(1,1)}(x,y)$, the 10 silhouette scores had a mean of -0.038 with a standard deviation of 0.001, indicating poor clustering quality. Alternatively, for case (b), with the distance based on $\cos_{(1,5)}(x,y)$, the scores had a mean of 0.833 and a standard deviation of 0.015, suggesting good clustering quality. Thus, the hierarchical clustering performance using $(\beta,\gamma)=(1,5)$-cosine similarity was significantly better than that using standard cosine similarity, as illustrated in typical dendrograms (Fig. \ref{fig3}).

\begin{figure}[h]
\vspace*{5mm}
\begin{center}
  \includegraphics[width=120mm]{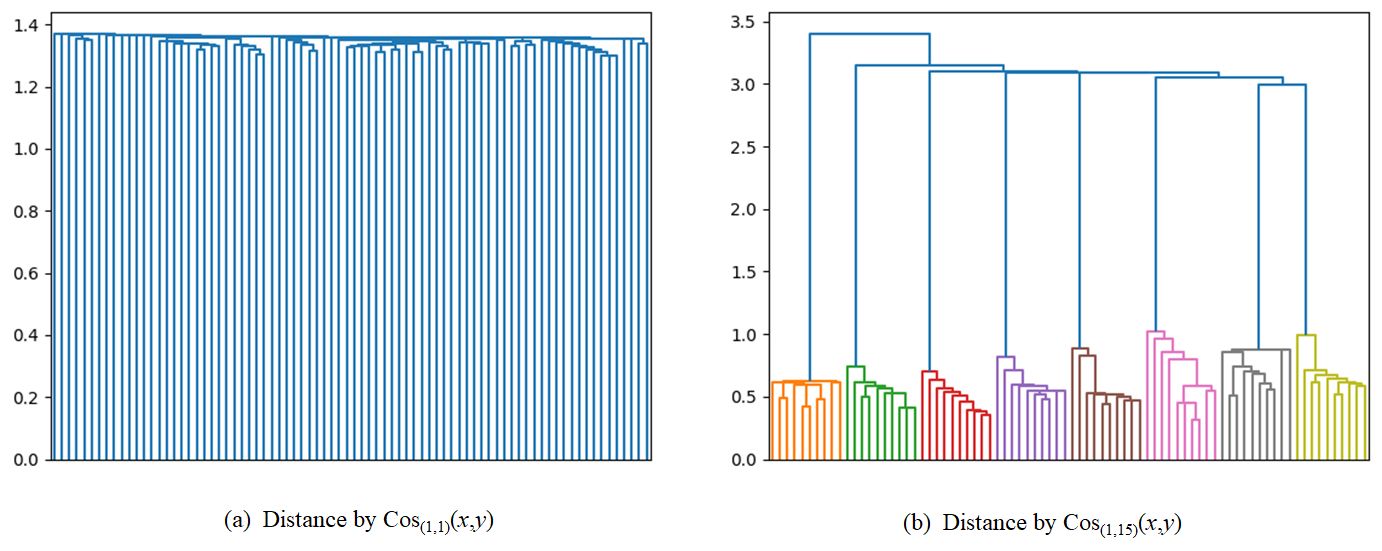}
 \end{center}
 \vspace{-1mm}\caption{The dendrograms by the distances based on (a)  and (b).}
\label{fig3}
\end{figure}

 Let $X$ be a $d$-variate variable with a covariance matrix $\Sigma$, which is a $d$-dimensional symmetric, positive definite matrix.  Suppose the eigenvalues $\lambda_1,...,\lambda_d$ of $\Sigma$  are restricted as $\lambda_1\geq...\geq\lambda_{d_0}\geq\epsilon>\delta\geq\lambda_{d_0+1}\geq...\geq\lambda_d$.  
Given $n$-random sample $(X_1,...,X_n)$ from $X$, the standard PCA is given by solving the $k$-principal vectors $v_1,...,v_k$, and $x\approx \sum_{j=1} \lambda_j v_jv_j^\top x  $.
Suppose that $(X_1,...,X_n)$ is generated from ${\tt Nor}_d(0,\Sigma)$, where 
\begin{align}\label{structure}
\Sigma=\begin{bmatrix}\Sigma_0&O\\O^\top& \epsilon \>{\mathbb I}_{d-d_0}\end{bmatrix}.
\end{align}  
Here $\Sigma_0$ is a positive-definite matrix of size $d_0\times d_0$-matrix whose eigenvalues are $(\lambda_1,...,\lambda_{d_0})$ and $O$ is a zero matrix of size $d_0\times(d-d_0).$
We set as 
\begin{align}\nonumber
n=500, d=1000, d_0=10, (\lambda_1,...,\lambda_{d_0})=(5,4.5,...,,1.5,1),\epsilon=0.1
\end{align} 
Thus, the scenario is envisaged a situation where the signal is just of $10$ dimension with the rest of $990$-dimensional noise.

For this, the sample covariance matrix is defined by
\begin{align}\nonumber
   S=\frac{1}{n}\sum_{i=1}^n (X_i-\bar X)(X_i-\bar X)^\top
\end{align} 
and $\hat \lambda_j$ and $\hat v_j$ are obtained as the $j$-th eigenvalue and eigenvector of $S$, where $\bar X$ is the sample mean vector.
We propose the $\gamma$-sample covariance matrix as
\begin{align}\nonumber
   S ^{\ominus\gamma} =\frac{1}{n}\sum_{i=1}^n(X_i^{\ominus\gamma}- \overline {X^{\ominus\gamma}})(X_i^{\ominus\gamma}- \overline {X^{\ominus\gamma}})^\top,
\end{align} 
where the $\gamma$-transform for a $d$-vector $x$ is given by $x^{\ominus\gamma}=({\rm sign}(x_j)|x_j|^\gamma)_{j=1}^d$ and
\begin{align}\nonumber
   \overline {X^{\ominus\gamma}}=\frac{1}{n}\sum_{i=1}^n X_i^{\ominus\gamma}.
\end{align}
Thus, the $\gamma$-PCA is derived by solving the eigenvalues and eigenvectors of $S_\gamma$.

To implement the PCA modification in Python, especially given the specific requirements for generating the sample data following these steps:

\begin{itemize}
\item 
Generate Sample Data: Create a 1000-dimensional dataset where the first 10 dimensions are drawn from a normal distribution with a specific covariance matrix $\Sigma_0$, and the remaining dimensions have a much smaller variance.

\item 
Compute the $\gamma$-Sample Covariance Matrix: Apply the $\gamma$ transformation to the covariance matrix computation.

\item 
Eigenvalue and Eigenvector Computation: Compute the eigenvalues and eigenvectors of the $\gamma$-sample covariance matrix.

\end{itemize}

We conducted a numerical experiment according to these steps.
The cumulative contribution ratios are plotted in Fig \ref{g-PCA1}.
It was observed that the standard PCA $(\gamma=1$) had poor performance for the synthetic dataset, in which the cumulative contribution to $10$ dimensions was lower than $0.3$.
Alternatively, the $\gamma$-PCA effectively improves the performance 
as the cumulative contribution to $10$ dimensions was higher than $0.9$ for $\gamma=2.0$. 
We remark that this efficient property for the $\gamma$-PCA depends on the simulation setting where
the signal vector $X_0$ of dimension $d_0$ and the no-signal vector $X_1$ of dimension $d-d_0$ are independent as in \eqref{structure}, where
$X$ is decomposed as $(X_0,X_1)$.
If the independence is not assumed, then the good recovery by the $\gamma$-PCA is not observed.
In reality, there would not be a strong evidence whether the independence holds or not. To address this issue we need  more discussion with real data analysis. 
Additionally, combining PCA with other techniques like independent component analysis or machine learning algorithms can further enhance its performance in complex data environments. This broader perspective should enrich the discussion in your draft, especially concerning the real-world applicability and limitations of PCA modifications.

\begin{figure}[htbp]
\begin{center}
  \includegraphics[width=100mm]{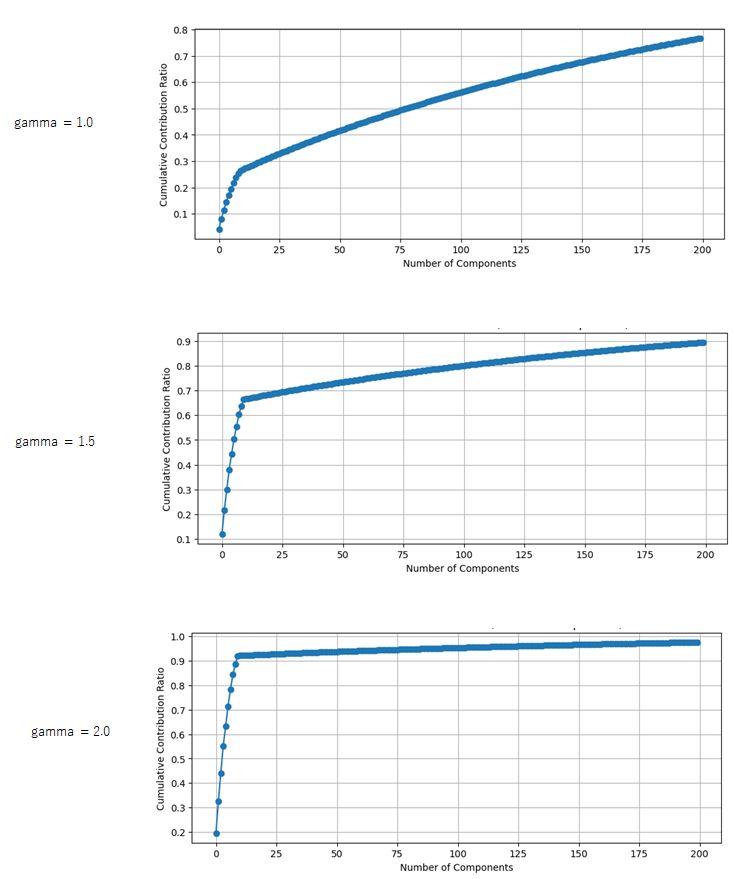}
 \end{center}
 \vspace{-1mm}\caption{Plot of cumulative contribution ratios with $\gamma=1.0,1.5,2.0$}
 \label{g-PCA1}
\end{figure}

We have discussed the extension of the $\gamma$-divergence to the Lebesgue L$_p$-space and introduces the concept of $\gamma$-cosine similarity, a novel measure for comparing functions or vectors in a function space. This measure is particularly relevant in statistics and machine learning, especially when dealing with signed functions or functional data.

The $\gamma$-divergence, previously studied in the context of regression and classification, is extended to the Lebesgue L$_p$-space. 
 To address the issue of functions taking negative values, a sign-preserved power transformation is introduced. This transformation is crucial for extending the log \(\gamma\)-divergence to functions that can take negative values.
 The concept of cosine similarity, commonly used in Hilbert spaces, is extended to the L$_p$-space. The $\gamma$-cosine similarity is defined as \(\cos_\gamma(f,g) = \langle f/\|f\|_p, (g/\|g\|_p)^{\ominus(p/q)}\rangle\), where \(p=\gamma+1\) and \(q\) is the conjugate exponent to \(p\). This measure maintains mathematical consistency across all real values of \(g(x)\).
Basic properties of $\gamma$-cosine similarity are explored:  \(|\cos_\gamma(f,g)| \leq 1\), with equality holding if and only if \(g\) is proportional to \(f\). It is also noted that \(\cos_\gamma(f,g) = 0\) if and only if \(\Delta_\gamma(f,g) = \infty\), indicating maximum distinctness between \(f\) and \(g\).
Generalized $(\beta,\gamma)$-cosine measure is given as a more general form of the cosine measure  is introduced for $L_{\beta+\gamma}(\Lambda)$ space, providing additional flexibility with tuning parameters \(\beta\) and \(\gamma\).
An application of these similarity measures in hierarchical clustering is demonstrated using Python. The $(\beta,\gamma)$-cosine similarity shows better performance in clustering high-dimensional data compared to the standard cosine similarity.
It can focus on essential components of the data, potentially reducing the need for preprocessing steps like principal component analysis.
The $\gamma$-PCA is defined, parallel to the $(\gamma,\gamma)$-cosine, and demonstrated for a good performance in high-dimensional situations. 
Therefore, the $\gamma$-cosine and $(\beta,\gamma)$-cosine measures could be particularly useful in statistical machine learning for comparing probability density functions, regression functions, or other functional forms, especially in scenarios where the sign of function values is significant.

In conclusion, the $\gamma$-cosine similarity and its generalized form, the $(\beta,\gamma)$-cosine measure, represent significant advancements in the field of statistical mathematics, particularly in the analysis of high-dimensional data and functional data analysis. These measures offer a more flexible and robust way to compare functions or vectors in various spaces, which is crucial for many applications in statistics and machine learning.

\section{Concluding remarks}

The concepts introduced in this chapter, particularly the GM divergence, $\gamma$-divergence, and $\gamma$-cosine similarity, offer promising avenues for advancing machine learning techniques, especially in high-dimensional settings. However, several areas warrant further exploration to fully understand and leverage these methodologies.

While the computational advantages of the GM divergence and $\gamma$-cosine similarity are demonstrated through simulations, real-world applications in domains such as bioinformatics, natural language processing, and image analysis could benefit from a deeper investigation. The scalability of these methods in extremely high-dimensional datasets, particularly those encountered in genomics or deep learning models, remains an open question.
Future research should focus on implementing these methods in large-scale machine learning pipelines to assess their performance and robustness compared to traditional methods. This could include exploring parallel computing strategies or GPU acceleration to handle the increased computational demands in practical applications.

The chapter primarily discusses the GM divergence and $\gamma$-divergence, but the potential to extend these ideas to other divergence measures, such as Jensen-Shannon divergence or Renyi divergence, could be fruitful. Investigating how these alternative measures interact with the GM estimator or can be integrated into ensemble learning frameworks like AdaBoost might yield novel insights and improved algorithms.
Moreover, a systematic comparison of these divergence measures across different machine learning tasks could provide clarity on their relative strengths and weaknesses.

While the $\gamma$-cosine similarity provides a novel way to compare vectors in function spaces, its theoretical underpinnings require further formalization. For instance, exploring its properties in different types of function spaces, such as Sobolev spaces or Besov spaces, might reveal new insights into its behavior and applications.
Additionally, the interpretability of the $\gamma$-cosine similarity in practical settings is a key aspect that should be addressed. How does this measure correlate with traditional metrics used in machine learning, such as accuracy, precision, and recall? Can it be used to enhance the interpretability of models, particularly in domains requiring high levels of transparency, such as healthcare or finance?

The methods discussed in this chapter are largely grounded in parametric models, particularly in the context of Boltzmann machines and AdaBoost. However, extending these divergence-based methods to non-parametric or semi-parametric models could open up new applications, particularly in statistical machine learning.
For example, exploring the use of GM divergence in the context of kernel methods, Gaussian processes, or non-parametric Bayesian models could provide new avenues for research. Similarly, semi-parametric approaches that combine the flexibility of non-parametric methods with the interpretability of parametric models could benefit from the computational advantages of the GM estimator.

To solidify the practical utility of the proposed methods, extensive empirical validation across a variety of datasets and machine learning tasks is essential. This includes benchmarking against state-of-the-art algorithms to evaluate performance in terms of accuracy, computational efficiency, and robustness.
Establishing a comprehensive suite of benchmarks, possibly in collaboration with the broader research community, could facilitate the adoption of these methods. Such benchmarks should include both synthetic datasets, to explore the behavior of these methods under controlled conditions, and real-world datasets, to demonstrate their applicability in practical scenarios.
6. Exploration of Hyperparameter Sensitivity
The introduction of $\gamma$ and $\beta$ parameters in the $\gamma$-cosine and $(\beta,\gamma)$-cosine measures adds a layer of flexibility, but also complexity. Understanding how sensitive these methods are to the choice of these parameters, and developing guidelines or heuristics for their selection, would be a valuable addition to the methodology.
Future work could explore automatic or adaptive methods for tuning these parameters, possibly integrating them with cross-validation techniques or Bayesian optimization to improve the ease of use and performance of the algorithms.
Conclusion
The introduction of GM divergence, $\gamma$-divergence, and$\gamma$-cosine similarity offers exciting opportunities for advancing machine learning and statistical modeling. However, their full potential will only be realized through continued research and development. By addressing the challenges outlined above, the field can better understand the theoretical implications, enhance practical applications, and ultimately, integrate these methods into mainstream machine learning practice.
\chapter*{Acknowledgements}
\addcontentsline{toc}{chapter}{Acknowledgements}

I also would like to acknowledge the assistance provided by ChatGPT, an AI language model developed by OpenAI. Its ability to answer questions, provide suggestions, and assist in the drafting process has been a remarkable aid in organizing and refining the content of this book. While any errors or omissions are my own, the contributions of ChatGPT have certainly made the writing process more efficient and enjoyable.




\end{document}